\documentclass[10pt,reqno]{amsart}
\usepackage{amsmath,amsfonts,amsbsy,amsthm,amssymb}
\usepackage{mathtools}
\usepackage{slashed}
\usepackage{eucal}
\usepackage{mathrsfs}
\usepackage{comment}
\usepackage{wasysym}
\usepackage{bbold}

\usepackage{tikz}
\usetikzlibrary{decorations.pathmorphing, decorations.pathreplacing,shapes,arrows,patterns}

\makeatletter
\renewcommand{\tocsection}[3]{%
  \indentlabel{\@ifnotempty{#2}{\bfseries\ignorespaces#1 #2\quad}}\bfseries#3}
\renewcommand{\tocsubsection}[3]{%
  \indentlabel{\@ifnotempty{#2}{\ignorespaces#1 #2\quad}}#3}

\newcommand\@dotsep{4.5}
\def\@tocline#1#2#3#4#5#6#7{\relax
  \ifnum #1>\c@tocdepth 
  \else
    \par \addpenalty\@secpenalty\addvspace{#2}%
    \begingroup \hyphenpenalty\@M
    \@ifempty{#4}{%
      \@tempdima\csname r@tocindent\number#1\endcsname\relax
    }{%
      \@tempdima#4\relax
    }%
    \parindent\z@ \leftskip#3\relax \advance\leftskip\@tempdima\relax
    \rightskip\@pnumwidth plus1em \parfillskip-\@pnumwidth
    #5\leavevmode\hskip-\@tempdima{#6}\nobreak
    \leaders\hbox{$\m@th\mkern \@dotsep mu\hbox{.}\mkern \@dotsep mu$}\hfill
    \nobreak
    \hbox to\@pnumwidth{\@tocpagenum{\ifnum#1=1\bfseries\fi#7}}\par
    \nobreak
    \endgroup
  \fi}
\AtBeginDocument{%
\expandafter\renewcommand\csname r@tocindent0\endcsname{0pt}
}
\def\l@subsection{\@tocline{2}{1pt}{2.5pc}{5pc}{}}
\makeatother

\newcommand{\verti}[1]{{\left\vert\kern-0.25ex\left\vert\kern-0.25ex\left\vert #1 
    \right\vert\kern-0.25ex\right\vert\kern-0.25ex\right\vert}}

\usepackage{xcolor}

\definecolor{BB}{rgb}{0, 0.35, .75}
\definecolor{BO}{rgb}{.7, .3, 0}
\definecolor{BeauBlue}{rgb}{0, 0.2, .85}
\definecolor{BeauOrange}{rgb}{.85, .2, 0}
\usepackage{hyperref}
\hypersetup{
	colorlinks = true,
	linkcolor = BeauBlue,
	urlcolor = cyan,
	citecolor = BeauOrange,
}

\usepackage{enumerate}
\usepackage{enumitem}
\setlist{leftmargin=*}
\usepackage{tikz}
\usetikzlibrary{decorations.pathmorphing,shapes,arrows}

\numberwithin{equation}{section}

\makeatletter
\newtheoremstyle{corsivo}
   {\medskipamount}{\medskipamount}%
   {\itshape}{}%
   {\bfseries}{}%
   { }
   {\thmname{#1}\thmnumber{\@ifnotempty{#1}{ }\@upn{#2}}%
    \thmnote{ {\bfseries\boldmath(#3)}}.}%
\makeatother

\theoremstyle{corsivo}
\newtheorem{theorem}{Theorem}[section]
\newtheorem*{theorem*}{Theorem}
\newtheorem{lemma}[theorem]{Lemma}
\newtheorem{corollary}[theorem]{Corollary}
\newtheorem{proposition}[theorem]{Proposition}
\newtheorem{definition}[theorem]{Definition}

\makeatletter
\newtheoremstyle{dritto}
   {\medskipamount}{\medskipamount}%
   {\rmfamily}{}%
   {\bfseries}{}%
   { }
   {\thmname{#1}\thmnumber{\@ifnotempty{#1}{ }\@upn{#2}}%
    \thmnote{ {\bfseries\boldmath(#3)}}.}%
\makeatother

\theoremstyle{dritto}

\newtheorem{remark}[theorem]{Remark}

\newcommand{\sub}[1]{_{\mathrm{#1}}}

\newcommand{\su}[1]{^{\mathrm{#1}}}

\newcommand{\eps}{\varepsilon}

\newcommand{\s}{\sigma} 

\newcommand{\Id}{\mathbb{1}}  
\newcommand{\ex}{\mathrm{e}}
\newcommand{\ee}{\mathrm{e}}
\newcommand{\iu}{\mathrm{i}}
\newcommand{\ii}{\mathrm{i}}   
\newcommand{\di}{\mathrm{d}}
\newcommand{\dd}{\mathrm{d}}

\newcommand{\Bs}{\mathcal{B}}                     

\newcommand{\cZ}{\mathcal{Z}}

\newcommand{\N}{\mathbb{N}}
\newcommand{\Z}{\mathbb{Z}}

\newcommand{\R}{\mathbb{R}}
\newcommand{\C}{\mathbb{C}}
\newcommand{\T}{\mathbb{T}}

\newcommand{\Do}{\mathcal{D}}

\newcommand{\Hi}{\mathcal{H}}

\newcommand{\Ss}{\mathcal{B}_{0}}

\newcommand{\F}{\mathcal{F}}

\newcommand{\scal}[2]{\left\langle #1 , #2 \right\rangle}                
\newcommand{\inner}[2]{\left\langle #1 , #2 \right\rangle}     
\newcommand{\norm}[1]{\left\| #1 \right\|}

\newcommand{\set}[1]{ \left\{  #1 \right\}} 

\DeclareMathOperator{\Tr}{Tr}         

\DeclareMathOperator{\re}{Re}

\DeclareMathOperator{\Span}{Span} 

\DeclareMathOperator{\End}{End}

\DeclareMathOperator{\Sp}{Spectrum}
\DeclareMathOperator{\dist}{dist}

\newcommand{\ie}{{\sl i.\,e.\ }}   
\newcommand{\eg}{{\sl e.\,g.\ }} 

\newcommand{\Or}{{\mathcal{O}}}

\newcommand{\abs}[1]{\left\lvert#1\right\rvert}

\newcommand{\virg}[1]{``#1''}

\renewcommand{\(}{\left(}
\renewcommand{\)}{\right)}
\newcommand{\Cr}{{\mathcal{C}}}

\let\oldfootnote\footnote
\renewcommand{\footnote}[1]{\oldfootnote{\  #1}}

\title[Spin transport and lack of quantisation]{Spin transport and lack of quantisation for time-reversal symmetric insulators on the honeycomb structure}
\author[L.~Fresta, G.~Marcelli]{Luca Fresta  \and Giovanna Marcelli}
\date{\today}

\begin{document}

\begin{abstract}
We investigate spin transport in a class of time-reversal symmetric insulators on the honeycomb structure, the Kane--Mele model being an emblematic example in this class. 
We derive the spin conductivity by the linear response \`a la Kubo and show that it is well-defined and independent of the choice of the spin current.
For models that do not conserve the spin, we demonstrate that the deviation of the spin conductivity from the quantised value is, at worst, quadratic in the spin-non-conserving terms, thus improving previous results. 
Additionally, we show that the leading-order corrections are actually quadratic for some models in the class, demonstrating that the spin conductivity is not universally quantised. Consequently, our results show that, in general, there is no direct connection between the spin conductivity and the Fu--Kane--Mele index.
\end{abstract}

\maketitle

\vspace{-12mm}
\tableofcontents
\goodbreak


\section{Introduction}
The study of topologically protected phases of matter has received a great deal of attention in the physics and mathematics community over the last decades.
The most well-known example of a topological effect in condensed-matter systems is the \textit{Integer Quantum Hall Effect} (IQHE), which consists in the quantisation in units of $e^{2}/h$ of the Hall conductivity in two-dimensional samples at low temperatures and sufficient level of impurities \cite{vonKlitzing}. 
The striking feature of the IQHE is that the Hall conductivity, a physical observable depending on the complex microscopic details of the Hamiltonian, is universal and linked with a topological invariant. 
For non-interacting electrons, this is the Chern number of the vector bundle associated with the Fermi projector \cite{TKNN,Avron83} in the translation invariant case, and a Fredholm index in the heterogeneous case \cite{Bellissard94,Avron94}. The quantisation of the Hall conductance is strikingly robust and persists for gapped interacting electron systems as well \cite{Hastings, Porta17}, see also \cite{AvronSeiler, Ryu-inv,Bachmann18}. 
The connection with topological invariants can also be established via the effective gauge theory which emerges by marginalising over the fermionic degrees of freedom \cite{FrohlichKerl,FrohlichZee}, see also the more recent work \cite{FW}.
The above findings are based on the so-called Kubo formula for the Hall conductivity, which was laid on mathematically solid ground in \cite{Teufel,Bachmann20II}.

Another primary example of topologically protected phases are the so-called \textit{time-reversal topological insulators}, which were theoretically predicted in condensed-matter physics in \cite{KaneMele2005} and subsequently observed in experiments \cite{Molenkamp,Hasanexp}.
Time-reversal topological insulators are time-reversal-symmetric materials that exhibit the \textit{Quantum Spin Hall Effect} (QSHE): even though they are normal insulators in the bulk, they carry edge modes whose signature is a robust non-zero spin current \cite{SchulzBaldes}. See also \cite{Frohlich,FrohlichST} for the prediction of this effect in non-relativistic many-body systems.
The topological properties of such materials are captured by a topological invariant known as the
Fu--Kane--Mele index \cite{FuKaneMele} and directly related with the edge modes of the system \cite{GrafPorta,Avila}; see also the recent work \cite{BBR} for the extension of this index to interacting many-body systems.

For spin-conserving systems, the spin Hall conductivity is connected with the Fu--Kane--Mele index, as is the case for the charge Hall conductivity with the Chern number: in particular, the spin Hall conductivity is quantised in units of $e/h$, the quantisation integer modulo two being the Fu--Kane--Mele index \cite{Bernevig2,HasanKane}.
However, in spin-non-conserving systems, it remains both experimentally and theoretically unclear whether any spin response function can be associated with the Fu–Kane–Mele index. On the one hand, experimental measurements are imprecise, likely because of the presence of magnetic impurities that partially break time-reversal symmetry \cite{Molenkamp,Mol2}.
On the other hand, theoretical problems already arise in identifying a suitable spin-current operator, the system lacking the associated conserved quantity \cite{ShiZhangXiaoNiu}. In this regard, two operators serve as natural generalizations of the spin current for spin-conserving systems: the \textit{conventional} spin current, defined as the product of the charge current with the spin operator, and the \textit{proper} spin current introduced in \cite{ShiZhangXiaoNiu,Zhang}, which is tied with a continuity equation and with the Onsager relations.

In \cite{MaPaTa}, the authors initiated a systematic study of spin transport in two-dimensional insulators, with the long-term goal of exploring possible connections with the Fu–Kane–Mele index. They showed that for any gapped, periodic, short-range, discrete, non-interacting Hamiltonian, one can choose a suitable definition of the spin current operator such that the Kubo-like terms for spin Hall conductivity and spin Hall conductance coincide. This equivalence relies on the vanishing of the mesoscopic average of the spin-torque response. The study of spin transport was further developed in subsequent work \cite{MaPaTe}, which considered gapped, periodic one-particle Hamiltonians in both discrete and continuum settings, in arbitrary spatial dimension $d$. Addressing both the conventional and the proper definitions of the spin current operator, the authors derived a general formula for the spin conductivity by constructing the non-equilibrium almost-stationary state (NEASS). In a specific class of lattice-periodic models with discrete rotational symmetry, they established the equality of the spin conductivity tensors arising from the two different definitions of the spin current operator.

In this work, using the Kubo formula, we continue the study of the spin conductivity initiated in \cite{MaPaTa, MaPaTe,MaMo}, by specialising to short-range, translation and $2\pi/3$-rotation symmetric models on the two-dimensional honeycomb structure.
The presence of spatial symmetries is crucial in our analysis to remove the ill-posedness problems and the ambiguities related to the choice of the spin current. 
Unlike \cite{MaPaTe}, in which the NEASS approach is used, here we follow Kubo's approach to transport coefficients. 
Moreover, by further exploiting spatial inversion symmetries of the honeycomb structure, we prove that the spin conductivity is an antisymmetric tensor. This property makes the spin conductivity invariant under rotations, and thus establishes that it measures an intrinsic transverse response, independent of the orientation of the laboratory \cite{Sinova}. 

In addition to reviewing key well-posedness findings, we prove two novel results for short-range, translation- and $2\pi/3$-rotation symmetric insulators that are time-reversal symmetric and that almost conserve the spin.
First, we show that the deviation of the spin conductivity from the quantised value (which is observed in spin-conserving systems) is, at worst, quadratic in the spin-non-conserving terms, improving previous results \cite{SchulzBaldes, MaPaTe}. Note that in theoretical units, which we henceforth adopt, quantisation of the spin conductivity means valued in $ \frac{1}{2 \pi} \mathbb{Z}$. 
Additionally, we show that in certain models the leading-order corrections to the quantised value are indeed quadratic, implying that spin conductivity is not universally quantised. 
As a result, our findings demonstrate that no general connection exists between spin conductivity and the Fu–Kane–Mele index when implementing the linear response by Kubo's formula. To the best of our knowledge, this is the first rigorous result in this direction.

Let us now state these results more precisely. We refer to insulators as pairs $(H,\mu)$ where $H$ is a Hamiltonian with a spectral gap, and where $\mu \in \mathbb{R}$ is in the spectral gap. An insulator $(H,\mu)$ is short-range, translation- and $2\pi/3$-rotation symmetric if $H$ is such, see \eqref{eq:gamma-per-rot-operators} and Section \ref{sec:Honeycomb_structure} for details.

Let us denote by  $S_{z}=\frac{1}{2}\s_{z}$ the spin operator in the $z$ direction. 
To quantify spin conservation, we follow \cite[Eq.~(1)]{SchulzBaldes} and split any Hamilton operator $H$ as follows
\begin{equation*}
H\su{sc}:=H+2 [S_z,H]S_z,\qquad H\su{snc}:=2 [H,S_z]S_z \;,
\end{equation*}
so that $[H\su{sc},S_{z}]= 0$ and $H=H\su{sc}+H\su{snc}$.
We will also need a stronger norm on the space of bounded operators, defined in \eqref{eqn:verti} and denoted by $\verti{\, \cdot \,}$.
An insulator $(H,\mu)$ almost conserves the spin when $\verti{[H,S_z]}$ or equivalently $\verti{H\su{snc}}$ is small depending on the distance between $\mu$ and $\sigma(H)$, see Definition \ref{def:almspincon} for details. In particular, we will always consider $H^{\mathrm{snc}}$ small enough so that also $(H\su{sc},\mu)$ is an insulator. 

We summarize our results in the following theorem.
\begin{theorem} \label{thm:intro}
\begin{enumerate}[label=(\roman*), ref=(\roman*)]
\item \label{it:i} For any short-range, translation- and $2\pi/3$-rotation symmetric insulator $(H,\mu)$, the spin conductivity tensor, which we denote by $\sigma^{s}(H,\mu) =\big(\sigma^{s}_{ij} (H,\mu) \big)_{i,j=1,2}$, is well-defined and independent of the choice of the spin current operator. Moreover, if $H$ is (spatially) inversion symmetric it holds true that $\sigma^{s}_{i j} (H,\mu) = - \sigma^{s}_{j i} (H,\mu)$, for any $i,j=1,2$.
\item\label{it:ii} For any short-range, translation- and $2\pi/3$-rotation symmetric insulator $(H,\mu)$, that is time-reversal symmetric and almost conserves the spin, there exists a constant $C$ independent of $\verti{[H,S_z]}$ such that

\begin{equation*}
\mathrm{dist}\Big(\sigma^{s}_{1 2}(H,\mu) \,,\,\frac{1}{2 \pi} \mathbb{Z} \Big) \leq  C \verti{[H,S_z]}^{2}.
\end{equation*}
\item\label{it:iii} There exist short-range, translation- and $2\pi/3$-rotation symmetric insulators $(H,\mu)$, that are time-reversal symmetric and almost conserve the spin, and a constant $C>0$ independent of $\verti{[H,S_z]}$ such that
\begin{equation*}
\mathrm{dist}\Big(\sigma^{s}_{1 2}(H,\mu) \,,\,\frac{1}{2 \pi} \mathbb{Z} \Big)  \geq C \verti{[H,S_z]}^{2}.
\end{equation*}
\end{enumerate}
\end{theorem}
Each item of Theorem \ref{thm:intro} is elaborated in different sections of the manuscript: part $(i)$ is addressed in Propositions \ref{thm:main} and \ref{prop:antisymmetry}, part $(ii)$ in Theorem \ref{cor:main1}, and finally part $(iii)$ in Theorem \ref{thm:main2}. 
As previously mentioned, the proof of $(i)$ follows the ideas of \cite{MaPaTa, MaPaTe}, but relies on Kubo’s approach to the computation of transport coefficients. 
Our analysis shows that the derivation of the spin conductivity by Kubo's formula agrees with the one obtained via NEASS, see Remark \ref{rmk: NEASS}. 
The proof of $(ii)$ involves more refined algebraic manipulations than in previous work, allowing us to extract commutators between operators that nearly commute with the spin operator, thus providing improved estimates. Iterating these manipulations does not lead to further cancellations or sharper bounds; yet, $(ii)$ does not rule out the possibility that the corrections could be universally of higher order. This is however not the case due to $(iii)$.
To prove $(iii)$, we explicitly construct an extension of the Kane--Mele model, introduced in \cite{KaneMele2005,KaneMele_graphene}, with next-to-nearest-neighbour Rashba interaction (see \eqref{eqn:hKM}) and analyse the phase diagram for suitable parameter choices. In Proposition \ref{prop:robust-dirac-points}, we show that, like in the Kane--Mele model, this diagram features three insulating phases separated by two semi-metallic phases (i.e., a metallic phase with a point-like Fermi surface). Using the ``imaginary-time'' representation of the Kubo formula for the spin conductivity, 
established via a Wick rotation (as in \cite{Porta16, Porta20}) and detailed in Theorem \ref{prop:imaginary-time-repr}, we compute the discontinuity of the spin conductivity across the phase transition in Proposition \ref{prop-discontinuity}.
The discontinuity is not quantised, with the leading-order corrections are quadratic in the spin-non-conserving terms. Using continuity arguments and exact results in the absence of Rashba interaction, we conclude that the spin conductivity is not universally quantised for such models. A key challenge in proving Proposition \ref{prop-discontinuity} is the presence of an arbitrarily small spectral gap, which is in particular smaller than the size of the spin-non conserving terms; in other words, $\|H\su{snc} \| \| (H\su{sc}  - \mu \mathbb{1})^{-1}\|$ is not small, making perturbative arguments useless.

Furthermore, we want to point out that the extension of the Kane--Mele model we propose draws inspiration from renormalization group considerations, compare with \cite{Porta16,Porta20} in the case of charge transport: the Kane--Mele model with an additional Hubbard-type interaction, exhibits large-scale properties at the critical energy that are effectively captured by a single-particle Hamiltonian which includes next-to-nearest-neighbour spin interactions.
Although the additional Rashba coupling is negligible in actual physical experiments and thus neglected in theoretical considerations, it turns out to be crucial for proving our non-universality result.
It is worth noting that the spin conductivity of the Kane–Mele model with two-body Hubbard interaction was studied in \cite{MPS}, but only in the spin-conserving case, hence the present work is complementary to their findings.  We plan to come back to the problem of studying the spin conductivity for the spin-non-conserving Kane--Mele model with Hubberd-type interaction in the future.

A last remark is in order. For quantum Hall systems at low temperature, disorder plays a fundamental role for the observation of quantised plateaus in the Hall conductivity \cite{Graf review,Bellissard94}. The role of disorder should  be important for topological insulators in general, although there seems to be a spectral gap in the systems studied experimentally \cite{Molenkamp,Mol2,Mol3}.
In any case, it would be interesting to extend the analysis of this work to models with a mobility gap, as is the case when suitable random potentials are added \cite{AizenmanWarzel}. Given the results in \cite{Bouclet}, we believe this extension to be at reach, and plan to come back to this problem in future work.
\medskip

\textbf{Structure of the paper}. In Section~\ref{sec:KM} we establish the framework for describing non-interacting fermions on the honeycomb structure and introduce the algebra of short-range, translation and $2\pi/3$-rotation symmetric linear operators $\mathcal{P}_{0}^{\hexagon}$. In Section~\ref{sec:spin-conductivity} we introduce the spin conductivity via the Kubo formula, and prove that for insulators in $\mathcal{P}_{0}^{\hexagon}$ they are well-defined and coincide for both \textit{proper} and \textit{conventional} spin current operator (Theorem \ref{thm:intro}\ref{it:i}). We then analyse time-reversal symmetric models in $\mathcal{P}_{0}^{\hexagon}$ for which the spin is almost conserved and show that the deviation of the spin conductivity from the quantised value is at worst quadratic in the spin-non-conserving terms (Theorem \ref{thm:intro}\ref{it:ii}). In Section~\ref{sec:phase-transition} we introduce the extended Kane--Mele model, study its phase diagram, and prove that the deviations from the quantised value are quadratic in the spin-non-conserving terms (Theorem \ref{thm:intro}\ref{it:iii}).
The latter result is based on an imaginary-time representation of the spin conductivity, which is proven in Appendix \ref{app:imaginary} via the Wick rotation.
\bigskip

{\bf Acknowledgements.}
We are deeply grateful to M.~
Porta for inspiring discussions on transport for electron systems and for his constant support, without which this work would not have been possible. G.~M.~is grateful to G.~Panati for introducing her to this research topic and for stimulating discussions. We also acknowledge G.M.~Graf and D.~Monaco for valuable and stimulating discussions.
We acknowledge the financial support from the European Research Council (ERC), under the European Union's Horizon 2020 research and innovation programme (ERC Starting Grant MaMBoQ, no. 802901).
The work of L.~F.~has been supported by the German Research Foundation (DFG) under Germany's Excellence Strategy - GZ 2047/1, Project-ID 390685813, and under SFB 1060 - Project-ID 211504053.
G.~M. acknowledges financial support from the Independent Research Fund Denmark--Natural Sciences, grant DFF–10.46540/2032-00005B and from the European Research Council through the ERC CoG UniCoSM, grant agreement n.724939.
We finally thank the anonymous referee for useful comments.

\section{The honeycomb structure}
\label{sec:KM}
\label{sec:Honeycomb_structure}

In this section, we introduce the Hilbert space $\mathcal{H}$ to describe electrons on the honeycomb structure, the relevant symmetries
and the subalgebra $\mathcal{P}_{0}^{\hexagon}(\mathcal{H})$ of short-range translation- and $2\pi/3$-rotation symmetric operators. Models in $\mathcal{P}_{0}^{\hexagon}(\mathcal{H})$ have desirable properties for studying spin transport, see Section \ref{sec:spin-conductivity}.  
Our notation (as well as Figure \ref{fig:honeycomb}) is largely borrowed from \cite[Appendix A]{MaPaTa}.
\medskip

We consider non-interacting electrons on a two-dimensional honeycomb structure $\Cr\subset \R^2$, see Figure \ref{fig:honeycomb}. This space configuration is characterized by the \emph{nearest-neighbour vectors}  
\begin{equation}\label{eq:d-vectors}
 d_1 = \frac{1}{2}( 1,-\sqrt{3}), \qquad d_2 = \frac{1}{2}(1, \sqrt{3}), \qquad  d_3 =(-1 , 0),
\end{equation}
and by the \emph{next-to-nearest-neighbour vectors} 
\begin{equation}\label{eq:a-vectors}
a_1  =  d_2 -  d_3 , \qquad  a_2 =  d_3- d_1 , \qquad  a_3 =  d_1- d_2  \;.
\end{equation}
\begin{figure}[htb]
\centering
\includegraphics{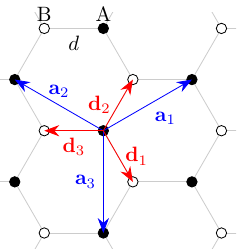}
\caption{\label{fig:honeycomb}The honeycomb structure.}
\end{figure}

\noindent
The vectors\footnote{Clearly, one of the $a_i$'s is superfluous, since it is an integer linear combination of the two others.} $\{ a_i\}_{i=1,2,3}$ generate the Bravais lattice $\Gamma := \Span_\Z\{  a_1,  a_2,  a_3\} \cong \Z^2$. 
The Bravais lattice $\Gamma$ acts on the honeycomb structure $\Cr$ by translations, \ie for any $\gamma\in\Gamma$
\begin{equation*}
\mathrm{T}_\gamma x:= x+\gamma\quad \text{ for all $x\in\Cr$}
\end{equation*}
which defines a group action $\mathrm{T}:\Gamma\times\Cr\to\Cr$.\\
\noindent
To describe non-interacting electrons with spin degrees of freedom we consider the Hilbert space
\begin{equation*}
\mathcal{H}:= \ell^2(\Cr) \otimes \mathbb{C}^{2} \;.
\end{equation*}
Let 
$\upsilon_A$ and $\upsilon_B$ denote respectively any vector reaching 
an $A$-site and a $B$-site (any black and white dot in Figure \ref{fig:honeycomb}) and introduce the triangular lattices
\[
\Gamma_\sharp:=\Gamma + \upsilon_\sharp,\qquad \sharp\in\{A,B\}.
\]
Observe that
\begin{equation*}
\Cr=\Gamma_A \cup \Gamma_B=\Gamma\, \cup\,(\Gamma+d_i)\quad\text{for every $i = 1,2,3$,}
\end{equation*}
where in the last equality we have chosen $\upsilon_A=(0,0)$ (thus, $\Gamma_A\equiv \Gamma$) and $\upsilon_B\in\{d_1,  d_2, d_3\}$, as shown in Figure \ref{fig:honeycomb}. Equivalently, one writes that
\[
\label{eqn:dimCr}
\Cr=\set{x=\gamma_x+[x],\,\text{ with }\, \gamma_x\in \Gamma,\, [x]\in\{0, d_i\}}\quad\text{for every $i=1,2,3$.}
\]
The above rewriting of elements in $\Cr$ is called \emph{dimerization} and depends on the chosen $ d_i$. With each dimerization one associates the \emph{dimerization isomorphism} 
\begin{equation}\label{eqn:dim}
\mathcal{D}_{i}\colon\ell^2(\Cr) \to \ell^2(\Gamma, \C^2)\,\text{ such that }\, (\mathcal{D}_{i}\psi)(\gamma):=
\begin{pmatrix}
\psi(\gamma)\\ 
\psi( d_i+\gamma)
\end{pmatrix}\text{for all $\gamma\in\Gamma$}
\end{equation}
which is extended to $\mathcal{H}$ by setting, with abuse of notation, $\mathcal{H} \ni \psi = (\psi_{\uparrow}, \psi_{\downarrow}) \mapsto {\mathcal{D}}_{i} \psi:= \begin{pmatrix}
\mathcal{D}_{i}\psi_{\uparrow}\\
\mathcal{D}_{i}\psi_{\downarrow}
\end{pmatrix}$. Then, one has that
\begin{equation*}
\Hi \cong\ell^2(\Gamma)\otimes \C^4\cong\ell^2(\Z^2)\otimes \C^4.
\end{equation*}
We denote the algebra of linear bounded operators acting on $\Hi$ by $\Bs(\Hi)$ and by $\| \cdot \|$ the operator norm. For any operator $A \in \Bs(\Hi)$, we denote by ${[A]}_{{\mathcal{D}}_{i}}:={\mathcal{D}}_{i}\, A \, {\mathcal{D}}_{i}^{-1}$ the corresponding operator acting on $\ell^2(\Z^2)\otimes \C^4$. We also adopt the following notation: for every $n,m\in\Z^2$
\begin{equation}\label{eqn:dimnot}
A^{(i)}_{n,m}\equiv {\big([A]_{{\mathcal{D}}_i}\big)}_{n,m}:=\scal{\delta_n}{[A]_{{\mathcal{D}}_i}\delta_m}
\end{equation}
where $\delta_n$ is defined as usual ${(\delta_n)}_m=\delta_{n,m}$ and $\scal{\,\cdot\,}{\,\cdot\,}$ is the standard scalar product on $\ell^2(\Z^2)$. Obviously, $A$ is characterized by the family of matrices $\{A^{(i)}_{n,m}\}_{n,m\in\Z^2}\subset \End\big( \C^4\big)$. Unless specified, we will consider the $i=3$ dimerization and, by abuse of notation, refer simply to $\{ A_{n,m}\}_{n,m\in\Z^2}$.
For such matrices we use the notation $\abs{A_{n,m}}$ for the matrix norm. 

A special type of bounded operators are \emph{short-range operators}. A linear operator $A \in \Bs(\Hi)$ is called short-range if and only if there exist constants $C, \xi >0$ such that
\begin{equation}\label{eq: exp-decay}
\abs{{A}_{n,m}}\leq C\ex^{-\abs{n-m}_1/\xi}\qquad\forall\,n,m\in\Z^2,    
\end{equation}
where we define $\abs{n}_1:=\abs{n_1}+\abs{n_2}$ for every $n=(n_1,n_2)\in\Z^2$. Short-range operators are indeed bounded operators by the well-known Schur--Hölmgren estimate
\begin{equation}\label{eqn:SH-bound}
\|A\| \leq \max \bigg\{ \sup_{n \in \Z^{2}} \sum_{m \in \Z^{2}} |A_{n,m}| , \sup_{m\in \Z^{2}} \sum_{n \in \Z^{2}} |A_{n,m}|\bigg\} \;.
\end{equation}
We denote the subalgebra of short-range operators by
\begin{equation*}\label{eq:short-range-set}
\Ss(\Hi) := \{ A \in \Bs(\Hi) \,| \, \text{$A$ is short-range} \} \;.
\end{equation*}

A particular role is also played by $\Gamma$-periodic operators in $\Bs(\Hi)$. To provide a precise definition, we introduce the set of \emph{translation operators} $(T_{ v})_{v \in \R^{2}}$ on $\ell^2(\mathcal{C})$, defined as
\begin{equation}\label{eq:transl} 
(T_{ v} \psi)( x) = \begin{cases}
\psi( x -v)             & \text{ if }  { x- v} \in \mathcal{C}\\
0                                    & \text{ otherwise}
\end{cases} \qquad \text{for all } \psi \in \ell^2(\mathcal{C}) \;.
\end{equation}
A linear operator $A \in \Bs(\Hi)$ is said to be $\Gamma$-\emph{periodic} (or simply \emph{periodic}) if and only if
\begin{equation*}
[A,T_\gamma]=0 \quad \forall \gamma\in\Gamma \;.
\end{equation*}
We denote the subalgebras of $\Gamma$-periodic bounded and short-range operators respectively by
\begin{equation}\label{eq:gamma-per-operators}
\begin{split}
\mathcal{P}(\Hi)&:= \{A \in \Bs(\Hi) \,| \, \text{$A$ is $\Gamma$-periodic} \} \,, 
\\  
\mathcal{P}_{0}(\Hi) & :=\{A \in \Ss(\Hi) \,| \, \text{$A$ is $\Gamma$-periodic} \} \;.
\end{split}
\end{equation}
The analysis of the operators in $\mathcal{P}(\Hi)$ is typically performed with the aid of the so-called Bloch–Floquet transform, see \eg \cite{Kuchment, Teufel1} and references therein.
\begin{definition}[Bloch--Floquet Transform]
\label{def:BF}
Let $\T^2_*:= \R^2/\Gamma^*$ be the \emph{Brillouin torus}, where the dual lattice $\Gamma^*$ is defined by
\[ 
\Gamma^* := \{  k \in {\R^2} :\,   k   \gamma \in 2\pi\Z\quad \text{ for all } \gamma \in \Gamma\}. 
\]
For any $i=1,2,3$, the Bloch--Floquet transform 
\begin{equation*}
\F_{i} \colon \Hi\to L^2(\T^2_*, \C^4)
\end{equation*}
is the unitary extension of the following map defined on $\psi \in \Hi$ being compactly supported:
\begin{equation}\label{eq:Fourier-transf}
\(\F_i\psi\)(k):=\frac{1}{\sqrt{2\pi}}\sum_{\gamma\in\Gamma}\ee^{-\iu k\gamma}\(\mathcal{D}_{i}\psi\)(\gamma)=\frac{1}{\sqrt{2\pi}}\sum_{\gamma\in\Gamma}\ee^{-\iu k\gamma}\begin{pmatrix}
\psi_\uparrow(\gamma)\\ 
\psi_\uparrow( d_i+\gamma)\\
\psi_\downarrow(\gamma)\\ 
\psi_\downarrow( d_i+\gamma)
\end{pmatrix} \;,
\end{equation}
where $\mathcal{D}_{i}$ denoting the dimerization isomorphism introduced in \eqref{eqn:dim}.
\end{definition}

The Bloch--Floquet transform induces the following transformation on $\Bs(\Hi)$, $A \mapsto A_i:=\F_i\, A\, \F_i^*$, the latter acting on $L^2(\T^2_*, \C^4)$. For any $\Gamma$-periodic operator $A$, the operator $A_i$ is \emph{fibered} or \emph{decomposable} in the sense that for every $\psi\in L^2(\T^2_*, \C^4) $
\[
(A_i\psi)(k) = A_i(k) \psi( k)   \qquad \text{ for all }  k \in \T^2_* \;.
\]
The matrix $A_i(k)$ is called the \emph{fiber of the operator} $A$ at the point $k$, and, as the notation suggests, depends on the choice of the dimerization. For any $\Gamma$-periodic operator $A$ we will use the \emph{fiber direct integral} notation
\begin{equation}\label{eq:BF-fibration}
\F_i  \, A \, \F_i ^{*} = \int_{\T^2_*}^{\oplus}\di k\, A_i(k).
\end{equation}
Besides bounded linear operators, we will eventually consider densely-defined linear unbounded operators, constructed in terms of the position operator, see Sections \ref{sec:linear-response} and \ref{subsec:spin-cond-AII}.

Besides translations, other symmetries will play an important role in the rest of this work. These symmetries act non-trivially on the spin degrees of freedom. We let
\begin{equation*}
\s_{x}=\begin{pmatrix}
0 & 1 \\ 1 & 0
\end{pmatrix} \,,
\qquad
\s_{y}=\begin{pmatrix}
0 & -\ii \\ \ii & 0
\end{pmatrix} \,,
\qquad
\s_{z}=\begin{pmatrix}
1 & 0 \\ 0 & -1
\end{pmatrix} \,,
\end{equation*}
denote the Pauli matrices and let $S_{i}$ denote the spin operator in the $i$-th direction for $i=x,y,z$:
\begin{equation}\label{eqn:spin}
S_i:= \frac{1}{2}\Id_{\ell^2(\Cr)}\otimes \s_{i} \;,
\end{equation}
and write $\s := (\s_x,\s_y,\s_z)$ and $S:= (S_{x},S_{y},S_{z})$.
 \label{it:rot} We also let $\mathrm{R}_{2\pi/3}$ denote the counterclockwise rotation of $2\pi/3$ in the plane:
\begin{equation}
\label{eqn:r120}
\mathrm{R}_{2\pi/3}(x_1,x_2):=\(-\frac{1}{2}x_1-\frac{\sqrt{3}}{2}x_2, \frac{\sqrt{3}}{2}x_1-\frac{1}{2}x_2\)\quad\forall\,(x_1,x_2)\in\R^2.
\end{equation}
Since $\mathrm{R}_{2\pi/3}\Cr=\Cr$, the restriction of the above map to $\Cr$ is still bijective and will be denoted by the same symbol.
We introduce the $2\pi/3$-\emph{rotation} operator on $\Hi$ 
\[
\({R}_{2\pi/3} \psi\)(x):=\ex^{-\iu \frac{2\pi}{3}S_z}\psi(\mathrm{R}^{-1}_{2\pi/3}x),\quad\forall\,\psi\in\Hi.
\]
\begin{remark}
\label{rem:rot}
Note that we embed the configuration space $\Cr$ in $\R^3$ by the canonical injection $(x_{1},x_{2}) \mapsto (x_{1},x_{2},0)$ and denote by $\mathrm{R}_{\theta,\hat{n}}$ a counter-clockwise rotation of $\R^{3}$ by an angle $\theta$ around the axis identified by the unit vector $\hat{n}$. The action of such three-dimensional rotations is defined on $\Hi$ by setting
\begin{equation*}\label{eq:3d-rotations}
\({R}_{\theta,\hat{n}}\psi\)(x):=\ex^{-\iu \theta\, \hat{n} \cdot S}\psi(\mathrm{R}^{-1}_{\theta,\hat{n}}x),
\end{equation*}
whenever $x, \mathrm{R}^{-1}_{\theta,\hat{n}}x \in \Cr$. In particular, one has that ${R}_{2\pi/3}=R_{2\pi/3, e_3}$ where $e_3 = (0,0,1)$.
\end{remark}

A linear operator $A$ on $\Hi$ is called $2\pi/3$-\emph{rotation symmetric} if and only if
\begin{equation*}
[A,{R}_{2\pi/3}]=0 \;.
\end{equation*}
Together with translations, rotational symmetry plays a crucial role in understanding the spin-transport properties of the models we consider. We therefore introduce the subalgebra of short-range $\Gamma$-periodic, $2\pi/3$-\emph{rotation symmetric} operators
\begin{equation}\label{eq:gamma-per-rot-operators}
\begin{split}
\mathcal{P}^{\hexagon}_{0}(\Hi)  &:=\{A \in \mathcal{P}_{0}(\Hi) \,| \, [A,{R}_{2\pi/3}]=0 \} \;,
\end{split}
\end{equation}
which will be our main focus throughout the rest of this work.

Other important symmetries of the honeycomb structure are vertical and horizontal spatial inversions.
Let $a$ be either $0$ or $1$. We say that $\Pi_a$ is an inversion if and only if $\Pi_a$ is an operator on $\Hi$ such that, possibly after an overall translation\footnote{In fact, for $a=0$, that is, for the horizontal inversion, one needs to translate by $(-1/2,0)$.}, $\Pi_a^2=\Id$,
\begin{equation}
\label{eqn:Pi}
\Pi_a\, X_i\, \Pi_a^{-1}=(-)^{i+a} X_i\quad\text{ and }\quad \Pi_a\, S_z\, \Pi_a^{-1}=-S_z.
\end{equation}
The inversion $\Pi_0$ is referred to as horizontal, while $\Pi_1$ is referred to as vertical. An operator $A$ on $\Hi$ is called \emph{inversion symmetric} if and only if
\begin{equation*}
[A,\Pi_a]=0\qquad \text{for some $a\in\{0,1\}$}.
\end{equation*}
Spatial inversion symmetries will be used to prove that the spin conductivity is an antisymmetric tensor, see Proposition \ref{prop:antisymmetry}, and will be also present in the extended Kane--Mele model discussed in Section \ref{sec:extended Kane--Mele}.

Finally, we introduce time-reversal symmetry on $\mathcal{H}$. The \emph{time-reversal} operator is an anti-unitary operator that flips the spin and acts trivially on the spatial degrees of freedom. For spin $1/2$ fermions on the honeycomb structure it is given by
\begin{equation}\label{eq:time-rev-def}
(\Theta \psi)(x) = \ee^{\ii \pi S_{y}} \overline{\psi(x)}\;,\qquad \forall \psi \in \mathcal{H} \,,
\end{equation}
and satisfies $\Theta^{2}= - \mathbb{1}_{\mathcal{H}}$. 
An operator $A$ on $\Hi$ is called \emph{time-reversal symmetric} if and only if
\begin{equation*}
[A,\Theta]=0\;.
\end{equation*}

\section{Spin transport}\label{sec:spin-conductivity}

\subsection{Trace per unit volume}\label{sec:trace}

Since we are interested in studying the linear response of the system for extensive observable, it is natural to normalise by the volume and therefore compute expectations with the trace per unit volume functional. This functional takes into account the invariance or covariance by discrete lattice translations, and thus particularly appropriate in the setting of periodic or, more generally, ergodic operators \cite{AizenmanWarzel, Bel, Bouclet}.
\begin{definition}[Trace per unit volume]\label{def:trace-volume}
For every $L\in 2\N +1$, let a \emph{fundamental cell of side $L$} be defined as
\begin{equation*}
C_L:=\left\lbrace x\in \Cr : x = \alpha_1 \,  a_1 + \alpha_2 \,  a_2 \text{ with }  |\alpha_j|\leq L/2 \quad\text{for $ j=1,2$}\right\rbrace
\end{equation*}
and let $\chi_L$ be the multiplication operator by the characteristic function of $C_L$.
The trace per unit volume of a linear operator $A$ on $\mathcal{H}$ is defined as 
\begin{equation*}
\tau(A):=\lim_{\substack{L\to\infty\\L\in 2\N+1}}\frac{1}{\abs{C_L}}\Tr(\chi_L\, A\, \chi_L)
\end{equation*}
whenever the above limit exists.
\end{definition}
In general, even when defined, the trace per unit volume does depend on the choice of a fundamental cell, which, in turn, depends on the choice of a basis $\{a_1, a_2\}$ over $\Z$ of the Bravais lattice $\Gamma$. However, when restricting to $\mathcal{P}(\Hi)$ this ambiguity is removed. Actually, by restricting to $\mathcal{P}(\Hi)$, the trace per unit volume functional acquires a whole set of desirable properties.
\begin{lemma}
\label{lem:tau}
 Letting $A\in\mathcal{P}(\Hi)$, the following holds true:
\begin{enumerate}[label=(\roman*), ref=(\roman*)]
\item \label{it:tau1} $\tau(A)$ is well-defined and 
\[
\tau(A)=\frac{1}{\abs{C_1}}\Tr(\chi_1 A\chi_1).
\]
In particular, $\tau(\,\cdot\,)$ is continuous since
\begin{equation}\label{eq: tau-state}
\big| \tau(A) \big| \leq \| A\| \;.
\end{equation}
\item \label{it:tau2} Let $C_1$ and $\widetilde{C}_1$  be two fundamental cells of side $L=1$. Then there exist $\gamma_1,\gamma_2\in\Gamma$ such that\footnote{The symbol $\sqcup$ denotes the disjoint union.}
\begin{equation}
\label{eqn:decfundcell}
C_1=\mathcal{P}_{\gamma_1}\sqcup \mathcal{P}_{\gamma_2}\quad\text{ and }\quad \widetilde{C}_1=\mathrm{T}_{\gamma_1}\mathcal{P}_{\gamma_1}\sqcup\mathrm{T}_{\gamma_2}\mathcal{P}_{\gamma_2},
\end{equation}
where $\mathcal{P}_{\gamma_1}$ and $\mathcal{P}_{\gamma_2}$ are subsets of the honeycomb structure $\Cr$. 
Denoting by $\tau$ and $\tilde{\tau}$ the trace per unit volume with respect to the fundamental cells $C_{L}$ and $\widetilde{C}_{L}$ respectively, then one has that 
\[
{\tau}(A)=\tilde{\tau}(A).
\]

\item \label{it:tau3}  For any $i=1,2,3$, letting $A_i(k)$ denote the Bloch--Floquet fiber of the operator $A$, see \eqref{eq:BF-fibration}, then
\[
\Tr(\chi_1 A \chi_1) = \frac{1}{|\T^2_*|} \int_{\T^2_*}  \di k\,\Tr_{\C^4}(A_i(k)). 
\]
In particular, the right-hand side does not depend on the chosen dimerization $i=1,2,3$.
\end{enumerate}
Furthermore, we have:
\begin{enumerate}[label=(\roman*), ref=(\roman*),start = 4]
\item \label{it:tau4}
Let $A \colon \R \to \mathcal{P}(\Hi)$ be Bochner integrable, \ie $\| A\| \in L^{1}(\mathbb{R})$. Then,
\begin{equation*}
\tau\Big( \int_{\R} A(t) \dd t \Big) = \int_{\R} \tau(A(t)) \dd t \;.
\end{equation*}
\item \label{it:tau5} The trace per unit volume is cyclic on $\mathcal{P}(\Hi)$, \ie
\begin{equation*}
\tau(AB)=\tau(BA)\quad \text{for all $A,B \in \mathcal{P}(\Hi)$.}
\end{equation*}
\item \label{it:tau6}
Let $A \in \mathcal{P}(\Hi)$ such that $\Tr(A\chi_{\mathcal{P}_{\gamma_1}})=0=\Tr(A\chi_{\mathcal{P}_{\gamma_2}})$, where $\mathcal{P}_{\gamma_1}$ and $\mathcal{P}_{\gamma_2}$ are introduced in  \eqref{eqn:decfundcell}. For $j=1, 2$ assume that the operator $X_j A$ be densely defined. Then we have that
\begin{equation}
\label{eqn:vantau(XA)}
\tau(X_j A)=0.
\end{equation}
\end{enumerate}
\end{lemma}

\begin{proof} 
\ref{it:tau1} Since $\chi_L$ is finite-rank operator, then $\chi_L A \chi_L$ is trace class. \cite[Proposition 2.4.1]{MaPaTe} implies the statement.\\
\ref{it:tau2} It follows by \cite[Lemma A.1]{MaPaTe} together with \cite[Proposition A.2.1]{MaPaTe}.\\
\ref{it:tau3} \cite[Proposition 2.4.2]{MaPaTe} proves the statement.\\
\ref{it:tau4} We note that for finite $L$ by linearity and by exploiting the inequality $\Tr(\abs{BC})\leq \Tr(\abs{B})\norm{C}$ for any trace class operator $B$ and and any bounded operator $C$, we have
\begin{equation*}
\frac{1}{\abs{C_L}}\Tr\chi_{L} \int_{\R}  A(t) \dd t \chi_{L} =  \frac{1}{\abs{C_L}}\int_{\R} \Tr \chi_{L} A(t) \chi_{L} \dd t \;,
\end{equation*}
and thus, since $|C_{L}|^{-1}\Tr \chi_{L} A(t) \chi_{L} \leq c \| A(t) \|$ for some uniform  constant $c$, we can take the limit $L \to \infty$ inside of the integral by dominated convergence and obtain the claim.\\
\ref{it:tau5} The proof can be found \eg \cite[Lemma 3.22]{Bouclet}.\\
\ref{it:tau6} 
Observe that in view of Lemma \ref{lem:tau}\ref{it:tau1}, decomposition \eqref{eqn:decfundcell} and the cyclicity of the trace, we have that $\abs{C_1}\tau(A)=\Tr( A \chi_{\mathcal{P}_{\gamma_1}})+\Tr( A\chi_{\mathcal{P}_{\gamma_2}})=0$. Thus, by using \cite[Eq.~(2.16)]{MaPaTe}, we obtain that
\begin{equation}\label{eqn:orsymnotused}
\frac{1}{\abs{C_L}}\Tr(\chi_L X_j A \chi_L)=\frac{1}{\abs{C_1}}\Tr( \chi_1  X_j A  \chi_1)+\frac{\tau(A)}{L^2} \left(\sum_{\gamma\in \Gamma\cap C_L}\gamma_j\right)=\frac{1}{\abs{C_1}}\Tr( \chi_1  X_j A  \chi_1),
\end{equation}
where the right-hand side term does not depend on the side $L$.
Hence, by employing again decomposition \eqref{eqn:decfundcell} we get that
\begin{align*}
\tau\( X_j A\)&=\frac{1}{\abs{C_1}}\left[  \Tr(\chi_{\mathcal{P}_{\gamma_1}} X_j A \chi_{\mathcal{P}_{\gamma_1}}) + \Tr(\chi_{\mathcal{P}_{\gamma_2}} X_j A \chi_{\mathcal{P}_{\gamma_2}})  \right]\\
&=\frac{1}{\abs{C_1}}\left[  \lambda_{1j}\Tr( A \chi_{\mathcal{P}_{\gamma_1}}) + \lambda_{2j}\Tr( A \chi_{\mathcal{P}_{\gamma_2}})  \right]=0,
\end{align*}
where we have used that $\chi_{\mathcal{P}_{\gamma_i}} X_j=\lambda_{ij}\chi_{\mathcal{P}_{\gamma_i}}$.
\end{proof}

\begin{remark}\label{rem:vanrob}
In \eqref{eqn:vantau(XA)} the vanishing is a \emph{robust} property, in the sense that it does not depend on the particular choice of the exhaustion $C_L\nearrow\Cr$, being such that $C_L \cap \Gamma$ is invariant under the reflection with respect to the origin: $x \mapsto - x$. Indeed, this spatial symmetry is not exploited in the proof above, see \eqref{eqn:orsymnotused}. Also, note that the argument of the proof of Lemma \ref{lem:tau}\ref{it:tau6} was used in \cite[Proposition 5.9]{MaPaTe}.
\end{remark}

\subsection{Linear response \`a la Kubo}\label{sec:linear-response}

We are interested in studying the linear response coefficients of insulators in $\mathcal{P}_{0}^{\hexagon}(\mathcal{H})$, subject to an external homogeneous electric field which is small in its intensity and switched on adiabatically in time.
In previous work \cite{MaPaTe}, this goal was accomplished via the construction of the so-called non-equilibrium almost-stationary state (NEASS).
Although here we borrow some ideas from \cite{MaPaTe}, by exploiting the spatial symmetries we are able to employ the standard Kubo approach \cite{Kubo}, avoiding the construction of the NEASS.

As already anticipated in the Introduction, an \emph{insulator} is
the pair $(H,\mu)$ consisting of a Hamilton operator $H$ being self-adjoint  with a spectral gap and $\mu \in \mathbb{R}$ is in this spectral gap. 
More precisely, we are interested in zero-temperature insulators, which are described by the Fermi projection associated with an insulator $(H,\mu)$:
\begin{equation}\label{eqn:defnP}    P \equiv P(H,\mu):=\chi_{(-\infty, \mu]}(H).
\end{equation}
By the Riesz formula, the Fermi projection has the following representation:
\begin{equation}
\label{eqn:riesz}
P=\frac{\iu}{2\pi} \oint_{\mathcal{C}} \di z\, (H - z )^{-1},
\end{equation}
where $\mathcal{C}$ is a positively-oriented complex contour intersecting the real axis at $\mu$ and below the bottom of the spectrum of $H$.

The following lemma, whose proof can be found in \cite{AizenmanGraf, AizenmanWarzel}, see also Lemma \ref{lemma:shortrange}\ref{it:i-P0} below, will be repeatedly used.
\begin{lemma}\label{lem:Pi0shortrange}
Let $(H,\mu)$ be an insulator. If $H \in \mathcal{P}_{0}(\Hi)$, then $P \in \mathcal{P}_{0}(\Hi)$ and $(\mathbb{1}+\ee^{\beta(H-\mu)})^{-1} \in \mathcal{P}_{0}(\Hi)$ for any $\beta >0$.
\end{lemma}

To define the external homogeneous electric field, we shall now introduce the \emph{position operator}: for $i=1,2$ we let $X_i$ denote the position operator in the $i$-th direction, that is,
\begin{equation}\label{eqn:X}
 (X_i \psi)(x ) := x_i \psi(x ), \quad\mbox{for all $\psi\in \Do(X_i)$,}
\end{equation}
where $\Do(X_i)$ denotes its maximal dense domain. The following standard result holds.
\begin{lemma}\label{lemma:shortrange}
\begin{enumerate}[label=(\roman*), ref=(\roman*)]
\item \label{it:i-P0}Let $A,B\in \mathcal{P}_{0}(\Hi)$ then $AB\in \mathcal{P}_{0}(\Hi)$. Furthermore, if $f$ is analytic in a strip around the real axis, then $f(A) \in \mathcal{P}_{0}(\Hi)$.
    \item  \label{it:ii-P0}Let $A$ be in  $\Ss(\Hi)$ (resp. in $\mathcal{P}_{0}(\Hi)$). Then, $[A,X_i]$ is in  $\Ss(\Hi)$ (resp. in $\mathcal{P}_{0}(\Hi)$).
\end{enumerate}
\end{lemma}
\begin{proof}
The first claim in \ref{it:i-P0} follows by the triangle inequality. To prove the second claim in \ref{it:i-P0}, we note that by Cauchy integral formula we can write
\begin{equation*}
f(H) = \frac{\ii}{2 \pi} \oint_{\mathcal{C}} \dd z f(z) (H-z)^{-1} \;,
\end{equation*}
where $\mathcal{C}$ is a positively-oriented complex contour within the strip of analyticity of $f$ and crossing the real axis away from the spectrum of $H$. The claim then follows by Combes--Thomas estimates for the resolvent $(H-z)^{-1}$, which give exponential decay for such $z \in \mathcal{C}$.

We now prove \ref{it:ii-P0}. Recalling the dimerization notation in \eqref{eqn:dimnot}, we obtain that
\begin{equation}
\label{eqn:dimpos}    
{\(X_i\)}^{(j)}_{n,m}=m_i\delta_{n,m} \Id_{\C^4}+
d_{j,i}\delta_{n,m}\begin{pmatrix}
  E_4  & 0\\
0 &  E_4
\end{pmatrix},
\end{equation}
where $d_{j,i}$ denotes the $i$-th component of the vector $d_j$ and $E_4:=\begin{pmatrix}0 & 0\\ 0 & 1\end{pmatrix}$. Therefore, one has that
\[
{\([A,X_i]\)}^{(j)}_{n,m}=(m_i-n_i){\(A\)}^{(j)}_{n,m}+d_{j,i}\left[{\(A\)}^{(j)}_{n,m},\begin{pmatrix} E_4 &0\\0 & E_4   \end{pmatrix}\right],
\]
thus $[A,X_i] \in \Bs_{0}(\Hi)$. Finally, if $A$ is periodic, by the Jacobi identity for commutators we get
\begin{align*}
[[A,X_i],T_\gamma]&=  -[[T_\gamma,A],X_i]-[[X_i,T_\gamma],A] =\gamma_i [ T_\gamma,A]=0\, ,
\end{align*}
where we have used that $[X_i,T_\gamma]=\gamma_i T_\gamma$.
\end{proof}
\begin{remark}
Note that the position operator in \eqref{eqn:X} distinguishes between sub-lattice points, hence the presence of the second term in \eqref{eqn:dimpos}. The same definition of the position operator is used in \cite{MaPaTe}, while in \cite{MaPaTa} it assigns the same value across the whole discrete unit cell.
\end{remark}
Let us now discuss the linear response theory \`a la Kubo \cite{Kubo}. If $(H,\mu)$ is an insulator, we can describe the small and adiabatically switched on perturbation by an external homogeneous electric field in terms of the time-dependent perturbed Hamiltonian
\begin{equation*}
H\sub{pert}(t):=H-\eps \ex^{\eta t}X_j,\qquad\text{for $t\leq 0$},
\end{equation*}
where $\eps\in[0,1)$ is the strength of the electric field in the $j$-th direction and $\eta\in (0,1)$ is the time-adiabatic parameter. Thus, the state of the perturbed system is given by the density operator $\rho$ solving the following Cauchy problem
\begin{equation*}
\begin{cases}
\iu \frac{\di}{\di t}{\rho} (t) = [H\sub{pert}(t), \rho(t)]\quad\text{for $t\leq 0$}, \\
\rho(-\infty) = P \;,
\end{cases}
\end{equation*}
where $P=P(H,\mu)$ is the Fermi projection.
Observe that asymptotically at $t=-\infty$ the perturbation is completely switched off, while in $t=0$ the perturbation is fully turned on. We are interested in the state $\rho(0)$ of the system at the final time $t=0$.
By following the strategy in \cite[Appendix A.2]{AizenmanGraf} for implementing the Kubo formula \cite{Kubo}, we single out the linear term in the formal $\eps$ expansion of the perturbed state $\rho(0)$.
Specifically, we first exploit the interaction picture for the above Cauchy problem, and then apply the fundamental theorem of calculus, to obtain the formal series in $\varepsilon$
\begin{equation}\label{eqn:kubo}
\begin{aligned}
\rho(0)&=P+\iu\eps \int_{-\infty}^{0}\di s\, \ex^{\eta s} \ex^{\iu s H}[X_j,P]\ex^{-\iu s H}+\Or_\eta(\eps^2)\\
&=:P+\eps\, L_{\eta,j}+\Or_\eta(\eps^2),
\end{aligned}
\end{equation}
where we have used the shorthand $\Or_\eta(\eps^2)$ to refer to further terms in the formal $\varepsilon$ series whose operator norm is bounded by $ C_{\eta} \,\eps^{2}$, for some constant $C_{\eta}>0$. We shall refer to the operator $L_{\eta,j}$ as the \emph{linear response of the state $P$}, the subscript $j$ emphasizing its dependence on the $j$-th direction. For any extensive observable $O$, the expectation with respect to the state $\rho$ is thus written by using the trace per unit volume, resulting in the following formal series in $\eps$
\begin{equation*}
\tau \big(O \rho(0) \big) = \tau\big(O P \big) + \varepsilon \tau \big(O L_{\eta,j} \big) + O_\eta(\eps^2) \;,
\end{equation*}
where, likewise, $O_\eta(\eps^2)$ is the shorthand for further terms in the formal $\varepsilon$ series that are bounded by $C_{\eta} \varepsilon^{2}$.
The term $\tau\big(O P \big)$ is the \emph{persistent} value of the observable $O$, which is vanishing in the case of spin transport, see Section \ref{subsec:spin-cond-AII}. 
The linear term, in the adiabatic limit, is precisely the linear response coefficient associated with the observable $O$.
Because the term $\Or_\eta(\eps^2)$ is henceforth neglected, our goal will be the study of the limit
\begin{equation*}
\lim _{\eta \to 0^{+}} \tau \big( O L_{\eta,j} \big) \;, \qquad \quad j =1,2 \;,
\end{equation*}
for the spin current operators, and to suitably express this limit in terms of the initial state $P$, see Section \ref{subsec:spin-cond-AII}. 

Before delving into that task, we want to grasp a better understanding of the linear response of the state $P$. First of all, we note the crude bound
\begin{align*}
\norm{L_{\eta,j}}\leq \int_{-\infty}^{0}\di s\, \ex^{\eta s} \norm{[X_j,P]}= \frac{1}{\eta} \norm{[X_j,P]} \;, 
\end{align*}
which shows that, even though $[X_j,P]$ is bounded (see Lemma \ref{lemma:shortrange}), computing the adiabatic limit $\eta \to 0^+$ of the term $L_{\eta,j}$ is non-trivial. As it turns out, the fact that $(H,\mu)$ is an insulator has to be used in an essential way in order to perform such a limit. 
In preparation for the next results, it is convenient to introduce the following notation: for any linear operator $A$ on $\mathcal{H}$ (we will consider operators in $\Bs(\Hi)$ or the position operator $X$), such that $AP$ is densely defined, we denote its \emph{diagonal} and \emph{off-diagonal parts}, with respect to the orthogonal projection $P$, by
\begin{equation}\label{eqn:diagoffdiag}
A\su{D}:=P A P+ P^\perp A P^\perp,\qquad A\su{OD}:=P A P^\perp+ P^\perp A P\equiv [[A,P],P],
\end{equation}
where $P^\perp:=\Id_{\Hi}-P$. If $A\equiv A\su{D}$ (resp.~$A\equiv A\su{OD}$), we say that $A$ is \emph{diagonal} (resp.~$A$ is \emph{off-diagonal}) with respect to $P$\footnote{In the paper the diagonality/off-diagonality of an operator is always understood with respect to the Fermi projection $P$.}.
\begin{lemma}\label{lem:I}
Let $(H,\mu)$ be an insulator with $H\in \mathcal{P}_{0}(\Hi)$ and let $P = P(H,\mu)$ be its Fermi projection. Let $A\in \mathcal{P}(\Hi)$ such that $A=A\su{OD}$ with respect to $P$.
Then, we have that
\begin{enumerate}[label=(\roman*), ref=(\roman*)]
\item \label{it:IA} the operator 
\begin{equation}
\label{eqn:I}
I(A):=\lim _{\eta \to 0^{+}}\iu \int_{-\infty}^{0}\di s\, \ex^{\eta s} \ex^{\iu s H}A\ex^{-\iu s H}
\end{equation}
exists in $\mathcal{P}(\Hi)$ and  satisfies $I(A)=\mathcal{L}_H^{-1}\( A\)$, where the Liouvillian (super-)operator  is defined as $\mathcal{L}_H(B):=[H,B]$ for every $B\in\Bs(\Hi)$.
\item \label{it:etaI}
\[
\lim _{\eta \to 0^{+}}\iu \int_{-\infty}^{0}\di s\, \eta \ex^{\eta s} B\ex^{\iu s H}A\ex^{-\iu s H}=0\qquad\text{ for all $B\in\mathcal{P}(\Hi)$.}
\]
\end{enumerate}
\end{lemma}

\begin{proof}
In order to prove both statements {\it \ref{it:IA}} and  {\it \ref{it:etaI}}, we first notice that both integrals can be fibered via the Bloch--Floquet transform (see \eqref{eq:BF-fibration}) since all the involved operators are in $\mathcal{P}(\Hi)$. Let us restrict to the case in which $A=PAP^\perp$, as the other term $P^\perp AP$ can be treated in a similar way. By using the spectral decomposition $H(k)=\sum_{n=1}^4E_n(k)P_n(k)$, where $P_n(k)$ is the eigenprojector associated with the eigenvalue $E_n(k)$, we have that
\[
\ex^{\iu s H(k)}P(k)=\sum_{n:\, E_n(k)\leq \mu}\ex^{\iu s E_n(k)}P_n(k),\qquad \ex^{-\iu s H(k)}P^\perp(k)=\sum_{m:\, E_m(k)> \mu}\ex^{-\iu s E_m(k)}P_m(k).
\]
Thus, we get that
\begin{equation}
\label{eqn:Apart}
\(\F_i  \,\ex^{\iu s H}P A P^\perp \ex^{-\iu s H} \, \F_i ^{*}\)(k)= \sum_{\substack{n:\, E_n(k)\leq \mu\\m:\, E_m(k)> \mu } } \ex^{\iu s (E_n(k)-E_m(k))}\,P_n(k)A(k)P_m(k).   
\end{equation}

{\it \ref{it:IA}} By using \eqref{eqn:Apart}, we deduce that
\begin{equation}\label{eqn:repr-P1}
\begin{split}
\left(\F_i \, I(P A P^\perp )\,\F_i ^{*}\right)(k)&=\lim_{\eta\to 0^+}\iu \sum_{\substack{n:\, E_n(k)\leq \mu\\m:\, E_m(k)> \mu } }P_n(k)A(k)P_m(k)\int_{-\infty}^{0}\di s\, \ex^{ s (\eta +\iu (E_n(k)-E_m(k)))}\\
&=\lim_{\eta\to 0^+}\iu \sum_{\substack{n:\, E_n(k)\leq \mu\\m:\, E_m(k)> \mu } }P_n(k)A(k)P_m(k)\frac{1}{\eta +\iu (E_n(k)-E_m(k))}\\
&=\sum_{\substack{n:\, E_n(k)\leq \mu\\m:\, E_m(k)> \mu } }\frac{P_n(k)A(k)P_m(k)}{ E_n(k)-E_m(k)},
\end{split}
\end{equation}
in the last equality we have used that $E_n(k)-E_m(k)\neq 0$ for all $k\in \T^2_*$ due to the hypothesis on $\mu$ in a spectral gap \cite[Theorem XIII.85]{RS4} and the continuity of the eigenvalues $E_n(k)$ for all $1\leq n\leq 4$.
From this explicit computation in $k$-space, we have that
\[
I(PAP^\perp)=\mathcal{L}_H^{-1}\( PAP^\perp\).
\]
{\it \ref{it:etaI}} 
By using equality \eqref{eqn:Apart}, we have that
\begin{align*}
&\lim _{\eta \to 0^{+}}\iu \int_{-\infty}^{0}\di s\, \eta \ex^{\eta s} B(k)\ex^{\iu s H(k)}A(k)\ex^{-\iu s H(k)}\\
&\qquad=\lim _{\eta \to 0^{+}}\sum_{\substack{n:\, E_n(k)\leq \mu\\m:\, E_m(k)> \mu } } \,B(k)P_n(k)A(k)P_m(k)\int_{-\infty}^{0}\di s\, \eta\ex^{ s (\eta +\iu (E_n(k)-E_m(k)))}\\
&\qquad=\lim _{\eta \to 0^{+}}\sum_{\substack{n:\, E_n(k)\leq \mu\\m:\, E_m(k)> \mu } } \,B(k)P_n(k)A(k)P_m(k)\frac{\eta}{\eta +\iu (E_n(k)-E_m(k))}=0,
\end{align*}
by exploiting again the insulator hypothesis guaranteeing $E_n(k)-E_m(k)\neq 0$ for all $k\in \T^2_*$.
\end{proof}

\begin{remark}\label{rmk:Liouvillan}
The usage of (super-)operator $I$ in \eqref{eqn:I} is ubiquitous in papers dealing with the quantum Hall effect for both non-interacting and interacting fermions. In our paper, it is crucial to establish the well-posedness and explicit formulas for the spin conductivities in Proposition \ref{thm:main}. The strategy exploiting the map $I$ follows the one adopted in \cite[Appendix A]{AizenmanGraf}, which in turn is related to \cite[Section 4]{Bellissard94}. Also in the context of interacting fermions \cite{Bachmann18, Teufel} the map $I$ is fundamental since it preserves the quasi-locality of operator-families \cite{Hastings2}.      
\end{remark}

An immediate consequence of Lemma \ref{lem:I} is the following:
\begin{corollary}\label{cor:Pi_1}
Let $(H,\mu)$ be an insulator with $H \in \mathcal{P}_{0}(\Hi)$ and let $P = P(H,\mu)$ be its Fermi projection.
Let $L_{\eta,j}$ be defined as in \eqref{eqn:kubo}. Then
\begin{equation}\label{eqn:P1}
P_{1}:= \lim _{\eta \to 0^{+}} L_{\eta,j}\,\text{ exists in $\mathcal{P}(\Hi).$}    
\end{equation}
\end{corollary}
\begin{proof}
Observe that $[X_j,P]\in \mathcal{P}_{0}(\Hi)$ by Lemma \ref{lemma:shortrange} and Lemma \ref{lem:Pi0shortrange}. Note that the operator $[X_j,P]$ is off-diagonal with respect to the decomposition induced by the projection $P$, see \eqref{eqn:diagoffdiag}. Therefore, Lemma \ref{lem:I}\ref{it:IA} implies that $I([X_j,P])\equiv P_1$ exists in $\mathcal{P}(\Hi).$
\end{proof}
\begin{remark}
\label{rmk:well-posed}
\begin{enumerate} [label=(\roman*), ref=(\roman*)]
\item \label{it:P11} As a straightforward consequence of the above corollary and of Lemma \ref{lem:tau}\ref{it:tau1} the limit $ \lim _{\eta \to 0^{+}}\tau(O L_{\eta,j} )$ is well-defined for any operator $O \in \mathcal{P}(\Hi)$. This is however not enough, as the main goal of the adiabatic perturbation theory is to find an expression for such expectation value in terms of the equilibrium state $P$.
When it comes to the linear response of the spin current, an additional difficulty is present since the proper spin current is an unbounded and non-periodic operator.
\item \label{it:P12} Let us observe that the operator $P_1$ coincides with $\Pi_1$ in \cite[Proposition 4.1.2]{MaPaTe}. Indeed, both of these operators are defined by computing the inverse of the Liouvillian $I$ on $[X_j,P]$, with the difference that in the current paper the super-operator $I$ is written in terms of a time integral over the real line, while in \cite{MaPaTe} the map $I$ is given by an energy integral in the complex plane (see \eg \cite[Equation (4.2)]{MaPaTe}). Thus, in particular we have the following identity
\begin{equation}
\label{eqn:2expP_1}
P_1=\lim_{\eta\to 0^+}\iu \int_{-\infty}^{0}\di s\, \ex^{\eta s} \ex^{\iu s H}[X_j,P]\ex^{-\iu s H}=\frac{\iu}{2\pi} \oint_{\mathcal{C}} \di z\, (H - z )^{-1}[[X_j,P],P] (H - z )^{-1}
\end{equation}

\end{enumerate}
\end{remark}

\subsection{Spin conductivity}
\label{subsec:spin-cond-AII}

We shall now focus on the spin current observables.
\begin{definition}\label{defn:spincurrents}
The \emph{conventional} and the \emph{proper spin current} operators are defined respectively for every $ i=1, 2$ as
\begin{equation}
\label{eqn:defnJ}
\begin{aligned}
J_{i}^{\mathrm{conv}}&:=\frac{\iu}{2} [H, X_i]S_z + \frac{\iu}{2}S_z[H, X_i]\\
J_{i}^{\mathrm{prop}}&:=  \iu [H,  X_i\,S_z]=X_i\, \mathscr{T}_z+\iu [H,  X_i]\,S_z,
\end{aligned}
\end{equation}
where $H\in\Ss(\Hi)$ is self-adjoint and
\begin{equation}
\label{eq: spin-torque}
\mathscr{T}_z:=\iu [H, S_z]
\end{equation}
is called the \emph{spin-torque} operator. 
\end{definition}
Note that if the spin $S_{z}$ is a conserved quantity, that is, $ [H, S_z]=0$, the two definitions basically coincide, $J_{i}^{\mathrm{conv}} \equiv J_{i}^{\mathrm{prop}}$. However, when spin is not a conserved quantity, associating a current to it becomes ambiguous, and this remains a topic of debate in condensed matter physics. The first definition has been adopted by several works, such as \cite{Sinovaetalii, ShengShengTingHaldane2005, SchulzBaldes}, while the second, more recent definition was proposed by \cite{ShiZhangXiaoNiu}. Both definitions come with their own advantages and disadvantages, depending also on the context in which they are applied. For the moment, we point out that $J_{i}^{\mathrm{prop}}$ presents more technical challenges: whenever $H \in \mathcal{P}_{0}(\Hi)$, then $J_{i}^{\mathrm{conv}} \in \mathcal{P}_0(\Hi)$ by Lemma \ref{lemma:shortrange}, while $J_{i}^{\mathrm{prop}}$ is neither periodic nor bounded because of the presence of the term $ X_i\, \mathscr{T}_z$.
In our analysis we follow \cite{MaPaTe} and consider both definitions of the spin current.

We now define the spin conductivity as the corresponding linear response coefficient.
\begin{definition}[Spin conductivity]\label{def-spin-cond}
Let $(H,\mu)$ be an insulator with $H \in \mathcal{P}_{0}^{\hexagon}(\Hi)$. The conventional and proper spin conductivities, respectively denoted by $\sigma^{\mathrm{s},\mathrm{conv}}_{ij} $ and $\sigma^{\mathrm{s},\mathrm{prop}}_{ij} $, 
are defined as
\begin{equation}\label{eqn:sigma}
\sigma^{\mathrm{s},\sharp}_{ij}:=\lim_{\eta\to 0^+}\re\tau(  J_{i}^{\sharp} \, L_{\eta,j})\qquad\text{for $\sharp \in \{\mathrm{conv, prop}\}$,}
\end{equation}
where the operator $J_{i}^{\sharp}$ is as in \eqref{eqn:defnJ}.
\end{definition} 

While $\sigma^{\mathrm{s},\mathrm{conv}}_{ij}$ is clearly well-posed, see Remark \ref{rmk:well-posed}\ref{it:P11}, this is not obvious for $\sigma^{\mathrm{s},\mathrm{prop}}_{ij}$, given that $J_{i}^{\mathrm{prop}}$ is an unbounded and non-periodic operator. Well-posedness of the latter is shown in the following proposition, which is the main result of this section.
\begin{proposition}\label{thm:main}
Let $(H,\mu)$ be an insulator with $H \in \mathcal{P}_{0}^{\hexagon}(\Hi) $. Let $P= P(H,\mu)$ be its Fermi projection. Let $\sigma^{\mathrm{s},\sharp}_{ij} = \sigma^{\mathrm{s},\sharp}_{ij}(H,\mu)$ be the spin conductivity as defined in \eqref{eqn:sigma}, with $\sharp \in \{\mathrm{conv},\mathrm{prop} \}$. Then, we have that: 

\begin{enumerate} [label=(\roman*), ref=(\roman*)]
\item \label{it:main1} Both spin conductivies are well-defined and 
\[
\sigma^{\mathrm{s,\mathrm{prop}}}_{ij}=\sigma^{\mathrm{s},\mathrm{conv}}_{ij} \;.
\]
\item \label{it:main2} It holds true that
\begin{equation*}
\sigma^{\mathrm{s},\mathrm{conv}}_{ij}=\Sigma\su{sc}_{ij}+\Sigma\su{snc}_{ij}\;,
\end{equation*}
where the \emph{spin-commuting term} $\Sigma\su{sc}_{ij}$ is given by
\begin{equation}\label{eqn:sigmasc}
\Sigma\su{sc}_{ij}:=\re \tau\Big(\iu S_z P \big[[X_i,P],  [X_j,P]\big]\Big)
\end{equation}
and the \emph{spin-noncommuting term} $\Sigma\su{snc}_{ij}$ is given by
\begin{equation}\label{eqn:sigmasnc}
\Sigma\su{snc}_{ij}:=\re \tau\Big(\iu [H,X_i] \big[[S_z,P],P\big] P_{1}  +\iu \big[[X_i,P],P\big][S_z,H]P_{1}  + \iu P[X_i,P]\big[[P,S_z],X_j\big] \Big),
\end{equation}
where $P_1$ is defined as in \eqref{eqn:P1}.
\item \label{it:main3}
In particular, in the spin-commuting case, \ie $[H,S_z]=0$, we have that 
\begin{equation*}
\sigma^{\mathrm{s},\mathrm{conv}}_{ij}=\Sigma\su{sc}_{ij}=\re \tau\Big( \iu S_z P\big[[X_i,P],  [X_j,P]\big]\Big).
\end{equation*}
\end{enumerate}
\end{proposition}
\begin{remark}
Thus, for every insulator $(H,\mu)$ with $H \in \mathcal{P}_{0}^{\hexagon}(\Hi)$ we are allowed to adopt the notation 
$$
\sigma^{\mathrm{s}}_{ij}:=\sigma^{\mathrm{s},\mathrm{prop}}_{ij}=\sigma^{\mathrm{s},\mathrm{conv}}_{ij}.
$$
\end{remark}

Before proving Proposition \ref{thm:main}, we establish an important geometric property of the spin conductivity, namely that it is an antisymmetric tensor under further symmetry assumptions.
This property makes the spin conductivity invariant under rotations, and thus establishes that it measures an intrinsic transverse response, the spin Hall response, regardless of the orientation of the laboratory \cite{Sinova}. 
Unlike the case of charge conductivity, this property is not obvious from its definition and follows from using both $2\pi/3$-rotation and a (spatial) inversion symmetry.
\begin{proposition}\label{prop:antisymmetry}
Let $(H,\mu)$ be an insulator with $H \in \mathcal{P}_{0}^{\hexagon}(\Hi)$ inversion symmetric. The spin conductivity $\sigma^{\mathrm{s}}_{ij}(H,\mu)$ is an antisymmetric tensor, that is
$$
\sigma^{\mathrm{s}}=\begin{pmatrix}
    0& \sigma^{\mathrm{s}}_{12}\\
    -\sigma^{\mathrm{s}}_{12} &0 
\end{pmatrix}.
$$
\end{proposition}
\begin{proof}
Since $\sigma^{\mathrm{s}}_{ij}= \lim _{\eta \to 0^{+}} \tau(J_{i}^{\mathrm{conv}} L_{\eta,j})$, it suffices to prove that the quantity $\tau(J_{i}^{\mathrm{conv}} L_{\eta,j})$ is antisymmetric. 
First of all, we show $\tau(J_{i}^{\mathrm{conv}} L_{\eta,i}) = 0$. This is a simple consequence of the symmetry under inversion. Indeed, since 
\begin{equation*}
\label{eqn:rvhops}
\Pi_a \,H \,\Pi_a^{-1} = H,
\qquad \Pi_a\, X_i\, \Pi_a^{-1} =(-)^{i+a} X_i, \qquad\Pi_a\, S_{z}\, \Pi_a^{-1} = - S_{z}
\end{equation*}
we have\footnote{Since $J_{i}^{\mathrm{conv}}$ and $L_{\eta,j}$  involve commutators with the position operator, their transformation properties are independent of the choice of the origin; see discussion before \eqref{eqn:Pi}.}
\begin{equation*}
\Pi_a \,J_{i}^{\mathrm{conv}} \,\Pi_a ^{-1} = (-)^{i+a+1}J_{i}^{\mathrm{conv}}\;, 
\qquad
\Pi_a\, L_{\eta,j}\, \Pi_a^{-1} = (-)^{j+a}L_{\eta,j} \;.
\end{equation*}
Therefore, since $J_{i}^{\mathrm{conv}} L_{\eta,j} \in \mathcal{P}(\Hi) $ by using Lemma \ref{lem:tau}\ref{it:tau4} and the identity $\Pi_a \,\chi_1\, \Pi_a^{-1}=\chi_1$, we conclude that
\begin{equation*}
\begin{split}
\tau(J_{i}^{\mathrm{conv}} L_{\eta,i}) =  \tau(\Pi_a\, J_{i}^{\mathrm{conv}} L_{\eta,i}\,\Pi_a^{-1}) = -\tau(J_{i}^{\mathrm{conv}} L_{\eta,i}) \;. 
\end{split}
\end{equation*}
Now, we proceed by proving that $\tau(J_{1}^{\mathrm{conv}} L_{\eta,2})=-\tau(J_{2}^{\mathrm{conv}} L_{\eta,1})$. Observe that
\begin{equation*}
R_{2\pi/3} H R_{2\pi/3}^{-1} = H,
 \qquad R_{2\pi/3} S_{z} R_{2\pi/3}^{-1} = S_{z},
\end{equation*}
\begin{equation*}
R_{2\pi/3} \begin{pmatrix}
X_1 \\ X_2
\end{pmatrix} R_{2\pi/3}^{-1}= 
\begin{pmatrix}
-\frac{1}{2} & +\frac{\sqrt{3}}{2} \\
-\frac{\sqrt{3}}{2} & -\frac{1}{2}
\end{pmatrix} \begin{pmatrix}
X_1 \\ X_2
\end{pmatrix}.
\end{equation*}
Therefore, we get that
\begin{align*}
\tau(J_{1}^{\mathrm{conv}} L_{\eta,2})&=\frac{\sqrt{3}}{4}\tau(J_{1}^{\mathrm{conv}} L_{\eta,1})+\frac{1}{4}\tau(J_{1}^{\mathrm{conv}} L_{\eta,2})-\frac{3}{4}\tau(J_{2}^{\mathrm{conv}} L_{\eta,1})-\frac{\sqrt{3}}{4}\tau(J_{2}^{\mathrm{conv}} L_{\eta,2})\\
&=\frac{1}{4}\tau(J_{1}^{\mathrm{conv}} L_{\eta,2})-\frac{3}{4}\tau(J_{2}^{\mathrm{conv}} L_{\eta,1}),
\end{align*}
in the first equality we have used that $R_{2\pi/3} \chi_1 R_{2\pi/3}^{-1}=\widetilde{\chi}_1$ being the characteristic function of $\widetilde{C}_1=\mathrm{R}_{2\pi/3} \, C_1$ with $\mathrm{R}_{2\pi/3}$ defined in \eqref{eqn:r120} (in other words, $\widetilde{C}_1$ is the fundamental cell with reference to the basis $\{a_2, a_3\}$) and Lemma \ref{lem:tau}\ref{it:tau2}, and in the second equality we have exploited that $\tau(J_{i}^{\mathrm{conv}} L_{\eta,i}) = 0$.
\end{proof}

Let us now deal with the proof of Proposition \ref{thm:main}. We will use the following intermediate result, which establishes that the spin torque term gives a vanishing contribution. This is a particular case of \cite[Proposition A.3.1]{MaPaTe}, but we spell out the details since this specific case has a more transparent proof. 
\begin{proposition}\label{prop:spintorque}
Let $(H,\mu)$ be an insulator with $H \in \mathcal{P}_{0}^{\hexagon}(\Hi) $. Let $P = P(H,\mu)$ be its Fermi projector. Let $\mathscr{T}_z$ and $L_{\eta,j}$ be the operators defined respectively in \eqref{eq: spin-torque} and \eqref{eqn:kubo}. Then, we have that:
\begin{enumerate}[label=(\roman*), ref=(\roman*)]
\item \label{it:vanstorqueonP} 
\begin{equation}\label{eqn:spintorqueonpoints}
    \Tr(\mathscr{T}_z\, L_{\eta,j} \,\chi_{\mathcal{P}_{\gamma_i}})=0\quad\text{for every $i,j=1,2$},
\end{equation}
where the subset $\mathcal{P}_{\gamma_i}$ is defined in \eqref{eqn:decfundcell}.
\item \label{it:vansrotoncell} 
\begin{equation*}
\tau (X_i\, \mathscr{T}_z \,L_{\eta,j} )=0 \quad\text{for every $ i,j =1,2$}.
\end{equation*}
\end{enumerate}
\end{proposition}
\begin{proof}
{\it \ref{it:vanstorqueonP}} 
First of all, note that $\mathscr{T}_z,L_{\eta,j}\in\mathcal{P}(\Hi)$ by Lemmas \ref{lem:Pi0shortrange} and \ref{lemma:shortrange}, thus the operator $\mathscr{T}_z\, L_{\eta,j} \,\chi_{\mathcal{P}_{\gamma_i}}$ is trace class. Consider the operator $\chi_{\{T_\nu \mathcal{P}_{\gamma_i}:\, \nu\in\Gamma 
  \}}$ being in $\mathcal{P}(\Hi)$. By Lemma \ref{lem:tau}\ref{it:tau1} we have that 
\[
\abs{C_1}\tau(\chi_{\{T_\nu \mathcal{P}_{\gamma_i}:\, \nu\in\Gamma 
  \}}\mathscr{T}_z\, L_{\eta,j})=\Tr(\chi_1\,\chi_{\{T_\nu \mathcal{P}_{\gamma_i}:\, \nu\in\Gamma 
  \}}\,\mathscr{T}_z\, L_{\eta,j}\,\chi_1)=\Tr(\mathscr{T}_z\, L_{\eta,j}\,\chi_{\mathcal{P}_{\gamma_i}}).
\]  
So we conclude by showing that $\tau(\chi_{\{T_\nu \mathcal{P}_{\gamma_i}:\, \nu\in\Gamma 
  \}}\mathscr{T}_z\, L_{\eta,j})=0$. 
Introducing $\widetilde{\chi}_1$ as the characteristic function of $\widetilde{C}_1=\mathrm{R}_{2\pi/3} \, C_1$ with $\mathrm{R}_{2\pi/3}$ defined in \eqref{eqn:r120}, we observe that
\begin{equation*}
\begin{aligned}
\abs{C_1}\tau(\chi_{\{T_\nu \mathcal{P}_{\gamma_i}:\, \nu\in\Gamma 
  \}}\mathscr{T}_z\, L_{\eta,j})&=\Tr(\chi_1\,\chi_{\{T_\nu \mathcal{P}_{\gamma_i}:\, \nu\in\Gamma 
  \}}\,\mathscr{T}_z\, L_{\eta,j}\,\chi_1)\\
  &=\Tr(R_{2\pi/3}\,\chi_1\,\chi_{\{T_\nu \mathcal{P}_{\gamma_i}:\, \nu\in\Gamma 
  \}}\,\mathscr{T}_z\, L_{\eta,j}\,\chi_1\, R_{2\pi/3}^{-1})\\
  &=\Tr(\widetilde{\chi}_1\,\chi_{\{T_\nu \mathcal{P}_{\gamma_i}:\, \nu\in\Gamma 
  \}}\,\mathscr{T}_z\, R_{2\pi/3}\, L_{\eta,j}\,R_{2\pi/3}^{-1}\,\widetilde{\chi}_1)\\
  &=|\widetilde{C}_1|\tilde{\tau}(\chi_{\{T_\nu \mathcal{P}_{\gamma_i}:\, \nu\in\Gamma 
  \}}\mathscr{T}_z\, R_{2\pi/3}\, L_{\eta,j}\,R_{2\pi/3}^{-1})\\
  &=\abs{{C}_1}\tau(\chi_{\{T_\nu \mathcal{P}_{\gamma_i}:\, \nu\in\Gamma 
  \}}\mathscr{T}_z\, R_{2\pi/3}\, L_{\eta,j}\,R_{2\pi/3}^{-1}) \;,
\end{aligned}
\end{equation*}
where we used the invariance of the trace under unitary conjugation, the identities $R_{2\pi/3}\,\mathscr{T}_z\, R_{2\pi/3}^{-1}= \mathscr{T}_z$, $R_{2\pi/3}\,\chi_{\{T_\nu \mathcal{P}_{\gamma_i}:\, \nu\in\Gamma 
  \}}\, R_{2\pi/3}^{-1}=\chi_{\{T_\nu \mathcal{P}_{\gamma_i}:\, \nu\in\Gamma 
  \}}$, and Lemma \ref{lem:tau}\ref{it:tau2}. We proceed by noting that
\begin{align*}
 R_{2\pi/3}\, L_{\eta,1}\,R_{2\pi/3}^{-1}&=\iu\int_{-\infty}^{0}\di s\, \ex^{\eta s} \ex^{\iu s H}[R_{2\pi/3}\, X_1\,R_{2\pi/3}^{-1},P]\ex^{-\iu s H}\\
 &=\iu\int_{-\infty}^{0}\di s\, \ex^{\eta s} \ex^{\iu s H} \left[\frac{\sqrt{3}}{2}X_2-\frac{1}{2}X_1,P\right]\ex^{-\iu s H},
\end{align*}
and similarly
\[
R_{2\pi/3}\, L_{\eta,2}\,R_{2\pi/3}^{-1}=-\iu\int_{-\infty}^{0}\di s\, \ex^{\eta s} \ex^{\iu s H} \left[\frac{\sqrt{3}}{2}X_1+\frac{1}{2}X_2,P\right]\ex^{-\iu s H}.    
\]
Therefore, defining 
$
\alpha_j:=\tau(\chi_{\{T_\nu \mathcal{P}_{\gamma_i}:\, \nu\in\Gamma 
  \}}\mathscr{T}_z\, L_{\eta,j}),
$  
we obtain the following system in $\alpha_j$'s:
\begin{equation*}
\begin{cases}    
\alpha_1&=\frac{\sqrt{3}}{2}\alpha_2-\frac{1}{2}\alpha_1\\
\alpha_2&=-\frac{\sqrt{3}}{2}\alpha_1-\frac{1}{2}\alpha_2,
\end{cases}
\end{equation*}
from which it follows that $\alpha_1=0=\alpha_2$.\\
{\it \ref{it:vansrotoncell}} Lemma \ref{lem:tau}\ref{it:tau6} together with the previous point \ref{it:vanstorqueonP} implies the statement.
\end{proof}

\begin{remark}
\begin{enumerate}[label=(\roman*), ref=(\roman*)]
Let us comment on the vanishing of the above quantities.
\item Clearly, since $C_1=\mathcal{P}_{\gamma_1}\sqcup \mathcal{P}_{\gamma_2}$ by \eqref{eqn:decfundcell} we have that 
\begin{equation*}
\tau(\mathscr{T}_z\, L_{\eta,j})=0.
\end{equation*}
Thus, we can interpret this result as vanishing of the spin torque $\mathscr{T}_z$ in the linear response of the state $L_{\eta,j}$ at mesoscopic scale. This vanishing property is \emph{robust} since it does not depend on the choice of the fundamental cell, as a consequence of Lemma \ref{lem:tau}\ref{it:tau2} (since $\mathscr{T}_z\, L_{\eta,j}\in\mathcal{P}(\Hi)$). 
\item Observe that $\tau(X_i\, \mathscr{T}_z \,L_{\eta,j})=0$ is \emph{robust} since it does not depend on the choice of the exhaustion $C_L\nearrow \Cr$, see Remark \ref{rem:vanrob}.
\end{enumerate}
\end{remark}

\begin{proof}[Proof of Proposition \ref{thm:main}]
{\it \ref{it:main1} \& \ref{it:main2}}
Note that the operator $J_{i}^{\mathrm{conv}} \, L_{\eta,j}$ is periodic and bounded since both $J_{i}^{\mathrm{conv}}$ and $[P,X_j]$ are short-range by Lemmas \ref{lem:Pi0shortrange} and \ref{lemma:shortrange}. Thus, Lemma \ref{lem:tau}\ref{it:tau1} ensures that $\sigma_{ij}^{\mathrm{s},\mathrm{conv}}$ is well-defined. On the other hand, for the proper spin conductivity $\sigma_{ij}^{\mathrm{s},\mathrm{prop}}$ we have that $\tau(X_i \mathscr{T}_z\, L_{\eta,j} )=0$ in view of Proposition \ref{prop:spintorque}\ref{it:vansrotoncell}, thus only $\iu [H,X_i]S_z L_{\eta,j}\in\mathcal{P}(\Hi)$ contributes to its value. Specifically, we get that
\begin{equation}\label{eqn:sigma1}
\begin{aligned}
\re\tau(J_{i}^{\mathrm{prop}}\, L_{\eta,j})&=\re \tau(\iu [H,X_i]S_z \,L_{\eta,j})\\
&=\frac{1}{2}\(  \tau(\iu [H,X_i]S_z L_{\eta,j})+ \overline{\tau(\iu [H,X_i]S_z L_{\eta,j}) }   \)\\
&=\frac{1}{2}\(  \tau(\iu [H,X_i]S_z L_{\eta,j})+ \tau(S_z \iu [H,X_i]L_{\eta,j} )    \)\\
&=\tau(J_{i}^{\mathrm{conv}}\, L_{\eta,j})
\end{aligned}
\end{equation}
where we have used that $\overline{\Tr(A)}=\Tr(A^*)$ for every trace class operator $A$ and the cyclicity of the trace per unit volume, see Lemma \ref{lem:tau}\ref{it:tau5}. Therefore, we shall compute the adiabatic limit $\eta\to 0^+$ of $\tau(J_{i}^{\mathrm{conv}}\, L_{\eta,j})$.

To this end, it is convenient to rewrite the argument of the real part in \eqref{eqn:sigma1} by using that the operator $L_{\eta,j}$ is off-diagonal due to the off-diagonality of $[X_j,P]$. By exploiting the cyclicity of the trace per unit volume, we have that
\begin{equation*}
\begin{aligned}
 \tau(\iu [H,X_i]S_z \,L_{\eta,j})&=\tau\(\iu {([H,X_i]S_z)}\su{OD} \,{(L_{\eta,j})}\su{OD}\)\\
&=\tau\(\iu [H,X_i\su{OD}]S_z\su{D} \,{(L_{\eta,j})}\su{OD}\)+\tau\(\iu [H,X_i\su{D}]S_z\su{OD} \,{(L_{\eta,j})}\su{OD}\)\\
&=:T_1(\eta)+T_2(\eta),
\end{aligned}
\end{equation*}
where we have defined the two terms $T_1(\eta)$ and $T_2(\eta)$ according to the order of appearance. We further split $T_1(\eta)$ as sum of two terms:
\begin{align*}
T_1(\eta)&=-\int_{-\infty}^{0}\di s\, \ex^{\eta s}\tau\( [H,X_i\su{OD}]S_z  \ex^{\iu s H} [X_j,P]\ex^{-\iu s H}\)\\
&=-\int_{-\infty}^{0}\di s\, \ex^{\eta s}\left\{\tau\(X_i\su{OD}S_z  \ex^{\iu s H} [X_j,P]\ex^{-\iu s H}H\)\right.\\
&\qquad\left.+\tau\( X_i\su{OD} [H,S_z]  \ex^{\iu s H}[X_j,P]\ex^{-\iu s H}\)+\tau\( X_i\su{OD} S_zH \ex^{\iu s H}[X_j,P]\ex^{-\iu s H}\)\right\}\\
&=\int_{-\infty}^{0}\di s\, \ex^{\eta s}\left\{\tau\( X_i\su{OD}S_z  [H,\ex^{\iu s H}[X_j,P]\ex^{-\iu s H}]\)+\tau\( X_i\su{OD} [H,S_z]  \ex^{\iu s H}[X_j,P]\ex^{-\iu s H}\)\right\}\\
&=:T_{11}(\eta)+T_{12}(\eta),
\end{align*}
by using that $X_i\su{OD} =[[X_i,P],P]\in \mathcal{P}(\Hi)$ in view of \eqref{eqn:diagoffdiag}, Lemma \ref{lem:Pi0shortrange} and \ref{lemma:shortrange}, and repeatedly the cyclicity of the trace per unit volume. For the term $T_{11}(\eta)$ we perform an integration by parts:
\begin{equation}
\label{eqn:eqn:sigma3}
\begin{aligned}
T_{11}(\eta) &=-\iu\int_{-\infty}^{0}\di s\, \ex^{\eta s}\tau\( X_i\su{OD}S_z   \frac{\di }{\di s}\ex^{\iu s H}[X_j,P]\ex^{-\iu s H}\)\\
&=-\iu \tau\( X_i\su{OD}S_z  [X_j,P]\)+\iu\int_{-\infty}^{0}\di s\, \eta\ex^{\eta s}\tau\( X_i\su{OD}S_z  \ex^{\iu s H}[X_j,P]\ex^{-\iu s H}\).
\end{aligned}
\end{equation}
 The second term on the right-hand side of \eqref{eqn:eqn:sigma3} does not contribute to the spin conductivity by Lemma \ref{lem:I}\ref{it:etaI}. On the other hand, in view of \eqref{eqn:diagoffdiag}, the first term on the right-hand side of \eqref{eqn:eqn:sigma3} can be rewritten as: 
\begin{equation}
\label{eqn:eqn:sigma4}
\begin{aligned}
-\iu \tau\( X_i\su{OD}S_z\su{D}   [X_j,P]\)&=\iu\tau\( P[X_i,P]S_z
   [X_j,P]\)-\iu \tau\( [X_i,P]P S_z  [X_j,P]\)\\
&=\iu\tau\( P[X_i,P][S_z,  [X_j,P]]\)+\iu\tau\( S_z P[[X_i,P],  [X_j,P]]\)\\
&=\iu\tau\( P[X_i,P][[P,S_z],X_j]\)+\iu\tau\( S_z P[[X_i,P],  [X_j,P]]\)
\end{aligned}
\end{equation}
where we have used repeatedly the cyclicity of the trace per unit volume, and in the last equality we have exploited that $[S_z,  [X_j,P]]=-[P,  [S_z,X_j]]-[X_j,[P,S_z]]=[[P,S_z],X_j]$.
By taking the real part of the second term on the last line of \eqref{eqn:eqn:sigma4}, we obtain the term in \eqref{eqn:sigmasc}; while in the last line of \eqref{eqn:eqn:sigma4} the real part of the first summand gives the third term in the expression for $\Sigma\su{snc}_{ij}$, see \eqref{eqn:sigmasnc}. Next, we collect the remaining terms contributing to the spin conductivity:
\begin{equation}
\label{eqn:firsttwoinsnc}
\lim_{\eta\to 0^+}\re \left\{T_2(\eta)+T_{12}(\eta)\right\}=\re \tau\(\iu [H,X_i] [[S_z,P],P] P_{1}\)+ \re \tau\( \iu [[X_i,P],P][S_z,H]P_{1}  \)
\end{equation}
where we have used that we can omit the diagonal part of $X_i$ due to the cyclicity of the trace per unit volume and definition of the operator $P_{1}$ in \eqref{eqn:P1}.
The right-hand side of \eqref{eqn:firsttwoinsnc} gives the remaining first two terms in the expression for $\Sigma\su{snc}_{ij}$.

{\it \ref{it:main3}} follows immediately from {\it \ref{it:main2}} by using that $\Sigma\su{snc}_{ij}=0$ due to $[H,S_z]=0=[P,S_z]$.
\end{proof}

\begin{remark}\label{rmk: NEASS}
Since Remark \ref{rmk:well-posed}\ref{it:P12} ensures that the linear response operator $P_1$ computed via Kubo's formula agrees with $\Pi_1$ given by the NEASS approach in \cite{MaPaTe}, it is clear that a priori these two methods yield the same formulas for the spin conductivies. Here, we have preferred to rewrite the spin-commuting term $\Sigma\su{sc}_{ij}$ in \eqref{eqn:sigmasc} by having the operator $S_z$ at the first place in the argument of the trace per unit volume, while in \cite[Theorem 5.6]{MaPaTe} $S_z$ appears nested in a commutator involving $[X_i,P]$ and $[X_j,P]$.
\end{remark}

\subsection{Spin-non-conserving contributions}
\label{sec: spin-non-conserving}

We shall now grasp a better understanding of the contribution to the spin conductivity $\sigma_{ij}^{\mathrm{s}}$ due to the spin-non-conserving terms in the Hamiltonian. This analysis will turn out to be useful in connection with the non-universality result presented in Section \ref{subsec: non-univers}.

Let us first recall what happens when the spin $S_{z}$ is conserved, that is, when $[H,S_z]=0$. For any orthogonal projection $P \in \mathcal{P}(\Hi)$, one defines its \emph{(first) Chern number} by 
\begin{equation*}
\mbox{Chern}(P)_{ij}:=\frac{\iu}{2\pi}\int_{\T^2_*}\di k\,\Tr_{\C^4}\big( P(k)[\partial_{k_j} P (k), \partial_{k_i} P (k)]  \big) \in \Z \;,
\end{equation*}
where $P(k)$ is the fiber operator of $P$, see \eqref{eq:BF-fibration}, for any choice of the dimerisation.

Denoting by $p^\uparrow$ and $p^\downarrow$ the $S_z$-eigenprojections, one also introduces the so-called \emph{spin Chern number} associated with $P$: \cite{ShengWengShengHaldane2006, Prodan2009}
\begin{equation}\label{eq:spin-Chern-number}
\mbox{ $S$-Chern$(P)_{ij}$} :=\frac{1}{2}\mbox{Chern}(p^\uparrow P p^\uparrow)_{ij} - \frac{1}{2} \mbox{Chern}(p^\downarrow P p^\downarrow)_{ij}  \in \frac{1}{2}\Z\;.
\end{equation}

In the spin commuting case, as is well-known the following holds true:
\begin{lemma}\label{lem:sigmasc}
Let $(H,\mu)$ be an insulator with $H \in \mathcal{P}_{0}^{\hexagon}(\Hi)$ such that $[H,S_{z}] = 0$. Then,
\begin{equation*}
\sigma_{ij}^{\mathrm{s}} = \frac{1}{2\pi}\mbox{ $S$-Chern$(P)_{ij}$} \;.
\end{equation*}
Furthermore, if $H$ is time-reversal symmetric, see \eqref{eq:time-rev-def} and below, then
\begin{equation}\label{eq:quantisation-spin-cons}
\sigma_{ij}^{\mathrm{s}} = \frac{1}{2\pi}\mbox{ $S$-Chern$(P)_{ij}$} \in \frac{1}{2\pi}\Z \;.
\end{equation}

\end{lemma}
In other words, whenever the spin $S_z$ is conserved and $H$ is time-reversal symmetric, the spin conductivity reduces to two copies of the charge conductivity, associated with the different spin sectors, and so its quantisation in fundamental units $\frac{e}{h}=\frac{1}{2\pi}$ is obviously established.
\begin{proof}
For brevity, we write $P^{\upuparrows}:= p^\uparrow P p^\uparrow $ and $P^{\downdownarrows}:= p^\downarrow P p^\downarrow $, so that $[H,S_{z}] = 0$ implies
$P= P^{\upuparrows}+ P^{\downdownarrows}$. 
Thus, by using Proposition \ref{thm:main}\ref{it:main3}, Lemma \ref{lem:tau}\ref{it:tau1} and \ref{it:tau3} one obtains that
\begin{equation*}\label{eqn:sigmaspinchern}
\begin{split}
\sigma_{ij}^{\mathrm{s}}&=\re \tau\big( \iu S_z P[[X_i,P],  [X_j,P]]\big)\\
&=\frac{1}{2}\tau\big( \iu P^{\upuparrows}[[X_i,P^{\upuparrows}],  [X_j,P^{\upuparrows}]]\big)-\frac{1}{2}\tau\big( \iu P^{\downdownarrows}[[X_i,P^{\downdownarrows}],  [X_j,P^{\downdownarrows}]]\big)\\
&=\frac{1}{2\pi}\left(\frac{1}{2} \mbox{Chern}(P^{\upuparrows})_{ij} -\frac{1}{2} \mbox{Chern}(P^{\downdownarrows})_{ij}\right)=\frac{1}{2\pi}\mbox{ $S$-Chern$(P)_{ij}$} \in \frac{1}{4\pi}\Z\;,
\end{split}
\end{equation*}
where we used that the fibre of $[X_i,P]$ is $\iu \partial_{k_{i}}P(k)$.
Moreover, if $H$ is time-reversal symmetric then we note the identity $\mbox{Chern}(P^{\downdownarrows})=-\mbox{Chern}(P^{\upuparrows})$  and conclude that $\sigma_{ij}^{\mathrm{s}}=\frac{1}{2\pi} \mbox{Chern}(P^{\upuparrows})_{ij}\in \frac{1}{2\pi}\Z$.
\end{proof}
In the general case in which the spin $S_{z}$ is not conserved, that is, when $[H,S_{z}] \neq 0$, as anticipated in the Introduction, we decompose the Hamiltonian into its \emph{spin-commuting} term $H\su{sc}$ and \emph{spin-non-commuting} term $H\su{snc}$, which we here recall:
\begin{equation}\label{eqn:decscsnc}
H\su{sc}:=H+2 [S_z,H]S_z,\qquad H\su{snc}:=2 [H,S_z]S_z \qquad \text{ such that } H=H\su{sc}+H\su{snc}\;,
\end{equation}
following \cite{SchulzBaldes,MaPaTe}.
To see that $H\su{sc}$ indeed commutes with $S_{z}$, 
we note the following simple identity, which will be exploited several times below: for any operator $A \in \Bs(\Hi)$, by recalling \eqref{eqn:spin} one has that
\begin{align*}
0=\frac{1}{4}[\Id_{\mathcal{H}}, A]=[S_z^2,A] = S_z[S_z,A]+[S_z,A]S_z,
\end{align*}
equivalently, since $S^{2}_{z}=\Id_{\mathcal{H}}/4 $,
\begin{equation}\label{eqn:[spin,O]}
[S_z,A] = -4 S_z[S_z,A]S_z.
\end{equation} 
Accordingly, we have that
\begin{equation*} 
[H\su{sc},S_z] =[H, S_z]+2\big([S_z,H]S_z^2- S_z[S_z,H]S_z\big) = 0 \;,
\end{equation*}
where we used \eqref{eqn:[spin,O]} and that $S^{2}_{z}=\Id_{\mathcal{H}}/4 $. 

It is convenient to introduce the following norm on $\mathcal{B}(\Hi)$
\begin{equation}
\label{eqn:verti}
\verti{A}:=\norm{A}+\sum_{j=1}^2\norm{[A,X_j]},
\end{equation}
whenever the commutator $[A,X_j]$ (is densely defined and) extends to a bounded operator.
\begin{definition}
\label{def:almspincon}
We say that an insulator $(H,\mu)$ \emph{almost conserves the spin} (in the $z$-direction) if and only if there exists a  constant $C_{s}<\frac{1}{\norm{{(H-\mu)}^{-1}}}$ such that $\verti{[H,S_z]}\leq C_s$, or equivalently $\verti{H\su{snc}}\leq C_s$. 
\end{definition}
\begin{remark}
\label{rem:PofHsc}
Here follow some observations which will be useful for our subsequent analysis.
\begin{enumerate}
[label=(\roman*), ref=(\roman*)]
\item \label{it:Cs} Note that requiring that $C_s$ is smaller than the size of the gap associated with the insulator $(H,\mu)$ implies that $(H\su{sc},\mu)$ is an insulator as well.
Indeed, let $\mathcal{C}$ be the  complex contour in \eqref{eqn:riesz}, fix $z\in\mathcal{C}$, consider 
\begin{align*}
H\su{sc}-z&=((H\su{sc}-H){(H-z)}^{-1}+\Id)(H-z)=(\Id-H\su{snc}{(H-z)}^{-1})(H-z).  
\end{align*}
Since $\norm{H\su{snc}}\norm{{(H-z)}^{-1}}\leq\frac{C_s}{\dist(z,\Sp(H))}\leq \frac{C_s}{\dist(\mu,\Sp(H))} <1$, the operator on the right-hand side is invertible. 
Thus, we can define the Fermi projector $P\su{sc}:=P(H\su{sc},\mu)$ associated with $(H\su{sc},\mu)$ by the Riesz formula:
\begin{equation*}
\label{eqn:rieszPsc}
P\su{sc}=\frac{\iu}{2\pi} \oint_{\mathcal{C}} \di z\, (H\su{sc} - z )^{-1},
\end{equation*}
where $\mathcal{C}$ is the same complex contour chosen in \eqref{eqn:riesz} to determine the Fermi projection $P$. 
\item \label{it:diffPPsc}
Observe that the operator norm of the difference between $P$ and $P\su{sc}$ is controlled, up to a constant, by $\norm{[H,S_z]}$. Indeed, we have that
\begin{align*}
P-P\su{sc}&=\frac{\iu}{2\pi} \oint_{\mathcal{C}} \di z \((H - z )^{-1}-(H\su{sc} - z )^{-1}\)=-\frac{\iu}{2\pi} \oint_{\mathcal{C}} \di z\,(H - z )^{-1}H\su{snc}  (H\su{sc} - z )^{-1},
\end{align*}
thus $\norm{P-P\su{sc}}\leq c \norm{[H,S_z]}\leq c\cdot C_s$, with $C_s$ in Definition \ref{def:almspincon}.
\end{enumerate}
\end{remark}

For insulators in $\mathcal{P}_{0}^{\hexagon}(\Hi)$ that almost conserve the spin, we derive approximation results for the spin conductivity $\sigma^{\mathrm{s}}_{ij}$ using an approximation procedure \cite{Prodan2009, SchulzBaldes, MaPaTe}.  
In the following theorem, which is the main result of this section, we will show that $\sigma^{\mathrm{s}}_{ij}$ equals, up to a remainder of order $\verti{H\su{snc}}^2$, both the spin-commuting term $\Sigma\su{sc}_{ij}$ in \eqref{eqn:sigmasc} and the $\frac{1}{2\pi}$-spin Chern number $\frac{1}{2\pi}\mbox{$S$-Chern}(P\su{sc})_{ij}$ associated with the spin-commuting Fermi projection, see \eqref{eq:spin-Chern-number} and Remark \ref{rem:PofHsc}. These results refine those obtained in \cite[Eq.~(8)]{SchulzBaldes} and \cite[Proposition 5.13]{MaPaTe}, improving them by one order in $\verti{[H,S_z]}$. 

\begin{theorem}\label{cor:main1}
Let $(H,\mu)$ be an insulator with $H \in \mathcal{P}_{0}^{\hexagon}(\Hi)$ almost conserving the spin in the sense of Definition \ref{def:almspincon}. Let $P= P(H,\mu)$ be its Fermi projection. Then, there exists a constant $C$ independent of $\verti{[H,S_z]}$ such that
\begin{equation}\label{eq: sigma-quasi-1}
\big|\sigma^{\mathrm{s}}_{ij} - \Sigma\su{sc}_{ij} \big| \leq  C \verti{[H,S_z]}^2
\end{equation}
and
\begin{equation}\label{eq: sigma-quasi-2}
\big|\sigma^{\mathrm{s}}_{ij} - \frac{1}{2\pi}\mbox{$S$-Chern}(P\su{sc})_{ij}\big| \leq C \verti{[H,S_z]}^2\;,
\end{equation}
where $P\su{sc}:=P(H\su{sc},\mu)$ is the projection associated with the insulator $(H\su{sc},\mu)$.
\end{theorem}

Under the hypotheses of Theorem \ref{cor:main1} and assuming that $H$ is time-reversal symmetric, then Theorem \ref{thm:intro}\ref{it:ii} follows.

Note that the deviation \eqref{eq: sigma-quasi-1} was studied numerically in \cite{MoUl} for the Kane–Mele and the Bernevig–Hughes–Zhang \cite{Bernevig} models, and later for a more general class of models \cite[Supplementary Notes]{LPGetal}. In all cases, it was found to be quadratic in $\verti{[H,S_z]}$, in agreement with our analytic result.

To prove Theorem \ref{cor:main1} we present two technical lemmas and introduce the following space of operators:
\begin{equation}
\label{eqn:Z(S_z)}
\cZ(S_z):=\set{A\in \mathcal{P}(\Hi):\;\exists\; C>0\text{ such that } \norm{[A,S_z]}\leq C\verti{[H,S_z]}  },
\end{equation}
bearing in mind that we are interested in models, like the extended Kane--Mele model in Subsection \ref{sec:extended Kane--Mele}, where $\verti{[H,S_z]}$ is a small quantity, in the sense that $H\su{snc}$ is a small perturbation of $H\su{sc}$, see \eqref{eqn:decscsnc}. In other words, $\cZ(S_z)$ consists of all periodic bounded operators which \virg{nearly commute} with the spin operator. 

\begin{remark}\label{rmk-products-Z}
By Leibniz's rule we have that if $A,B\in \cZ(S_z)$ then $AB\in \cZ(S_z)$.
\end{remark}

\begin{lemma}\label{lem:lambdasquared}
Let $(H,\mu)$ be an insulator that almost conserves the spin.
If $A,B,F \in \cZ(S_z)$, then there exists a constant $C$ independent of $\verti{[H,S_z]}$ such that
\begin{equation}\label{eqn:lambdasquared1}
\abs{\tau\left(A[S_z, F] B\right)}\leq C  \verti{[H,S_z]}^2 \;.
\end{equation}
\end{lemma}

\begin{proof}
We note that $\tau (A[S_z, F] B )$ 
is well-defined by Lemma \ref{lem:tau}\ref{it:tau1}.
By using identity \eqref{eqn:[spin,O]} for the operator $F$, we obtain that
\begin{align*}
\frac{1}{4}\tau(A[S_z, F] B)&=-\tau(A S_z[S_z, F] S_z B)\\
&=-\tau([A, S_z][S_z, F] S_z B)-\tau( S_z A [S_z, F] S_z B)\\
&=-\tau([A, S_z][S_z, F] [S_z, B])-\tau([A, S_z][S_z, F] B S_z)\\
&\phantom{=}\quad-\tau( S_z A [S_z, F] [S_z, B])-\tau( S_z A [S_z, F] B S_z ).
\end{align*}
Therefore, by using the cyclicity of $\tau(\cdot)$ and the identity $S_z^2=\frac{1}{4}\Id_{\mathcal{H}}$, we obtain that
\begin{align}\label{eqn:terms-tauAB}
-\frac{1}{2}\tau(A[S_z, F] B)&=\tau([A, S_z][S_z, F] [S_z, B])+\tau([A, S_z][S_z, F] B S_z) \nonumber \\
&\phantom{=}\quad+\tau( S_z A [S_z, F] [S_z, B]) \;.
\end{align}
Thus, by using \eqref{eq: tau-state} and the hypothesis that $A,B,F\in \cZ(S_z)$, inequality \eqref{eqn:lambdasquared1} follows.
\end{proof}

For proving Theorem \ref{cor:main1}, in the next lemma we show that certain operators are in $\cZ(S_z)$, so that we can apply Lemma \ref{lem:lambdasquared} to the difference between the spin conductivity $\sigma_{ij}^{\mathrm{s}}$ and the spin-commuting term $\Sigma\su{sc}_{ij}$, or the  $\frac{1}{2\pi}\mbox{$S$-Chern}(P\su{sc})_{ij}$ associated with the spin-commuting Fermi projection $P\su{sc}$.

\begin{lemma}\label{lem:Zspace}
Let $(H,\mu)$ be an insulator which almost conserves the spin, with $H \in \mathcal{P}_{0}(\Hi)$. Let $P = P(H,\mu)$ be its Fermi projector. Then we have that
\begin{equation*}
H, \, {(H-z)}^{-1},\, P,\,  [H,X_j],\, [{(H-z)}^{-1},X_j],\,[P,X_j],\, P_1\in\mathcal{Z}(S_z),
\end{equation*}
for any $ j=1, 2$, and for every $z$ in the complex contour $\mathcal{C}$ in \eqref{eqn:riesz}.
\end{lemma}
\begin{proof}
First of all, we note that $P, {(H-z)}^{-1} ,[H,X_j], [P,X_j], P_1\in \mathcal{P}(\Hi)$ by Lemmas \ref{lem:Pi0shortrange}, \ref{lemma:shortrange}, and Corollary \ref{cor:Pi_1}. Notice that 
\begin{equation}
\label{eqn:commresxj}    
[{(H-z)}^{-1},X_j]={(H-z)}^{-1}[X_j,H]{(H-z)}^{-1}
\end{equation}
which is in $\cZ(S_z)$ as well.
Observe that obviously $H\in \cZ(S_z)$ by the very definition of the space $\cZ(S_z)$. Note that 
\begin{equation}
\label{eqn:resolventinZ}
[(H - z )^{-1},S_z]=(H - z )^{-1}[S_z,H] (H - z )^{-1}, 
\end{equation}
thus by the boundedness of the resolvent operators we get that $(H - z )^{-1}\in \cZ(S_z)$ uniformly in $z$ which varies in the compact set $\mathcal{C}$.
Therefore, by using \eqref{eqn:riesz}, we have that
\begin{equation*}
\label{eqn:[P,Sz]}
[P,S_z]=\frac{\iu}{2\pi} \oint_{\mathcal{C}} \di z\, [(H - z )^{-1},S_z] \in \cZ(S_z).   
\end{equation*}
By Jacobi's identity we note that
\[
[[H,X_j],S_z]=[[H,S_z],X_j],
\]
so we obtain that $[H,X_j]\in\cZ(S_z)$ as well. 
By Jacobi's identity, \eqref{eqn:resolventinZ} and \eqref{eqn:commresxj}, we have that
\begin{align*}
[[(H - z )^{-1},X_j],S_z]&=[[(H - z )^{-1},S_z],X_j]= [(H - z )^{-1}[S_z,H] (H - z )^{-1},X_j]\\
&=[(H - z )^{-1},X_j][S_z,H] (H - z )^{-1} +(H - z )^{-1}[[S_z,H],X_j] (H - z )^{-1}\\
&\qquad+(H - z )^{-1}[S_z,H] [(H - z )^{-1},X_j]\\
&= (H - z )^{-1} [X_j,H](H - z )^{-1}[S_z,H] (H - z )^{-1}\\
&\qquad+(H - z )^{-1}[[S_z,H],X_j] (H - z )^{-1}\\
&\qquad+(H - z )^{-1}[S_z,H](H - z )^{-1} [X_j,H](H - z )^{-1},
\end{align*}
thus one estimates the right-hand side terms by $\verti{[H,S_z]}$ up to a constant, being uniform in $z\in\mathcal{C}$. Thus, again in view of \eqref{eqn:riesz}, we get that 
\[
[[P,X_j],S_z]=\frac{\iu}{2\pi} \oint_{\mathcal{C}} \di z\, [[(H - z )^{-1},X_j],S_z]\in\cZ(S_z).
\]
We are left to show that $P_1\in \cZ(s_z)$. By employing identity \eqref{eqn:2expP_1}, we get that
\[
[P_1,S_z]=\frac{\iu}{2\pi} \oint_{\mathcal{C}} \di z\, [(H - z )^{-1}[[X_j,P],P] (H - z )^{-1},S_z], 
\]
thus Leibniz's rule and ${(H-z)}^{-1},P,[P,X_j]\in \cZ(S_z)$ imply that $P_1\in \cZ(s_z)$.
\end{proof}

Now, we are ready to show the main result of this Section.

\begin{proof}[Proof of Theorem \ref{cor:main1}]
Let us first prove \eqref{eq: sigma-quasi-1}.
By Proposition \ref{thm:main}, we have
\begin{equation}\label{eqn:spuriousterms}
\sigma^{\mathrm{s}}_{ij}-\Sigma\su{sc}_{ij}=\re \tau\(\iu [H,X_i] [[S_z,P],P] P_{1}  +\iu [[X_i,P],P][S_z,H]P_{1}  + \iu P[X_i,P][[P,S_z],X_j] \).
\end{equation}
The conclusion follows by applying Lemma \ref{lem:lambdasquared} on each of the terms on the r.h.s.~ of \eqref{eqn:spuriousterms}. Indeed, considering the first term, we write
\begin{equation*}
\re\tau\left(\iu [H,X_i] [[S_z,P],P] P_{1}\right) =\re \tau\left(\iu [H,X_i] [S_z,P]PP_{1}\right)-\re \tau\left(\iu [H,X_i] P[S_z,P] P_{1}\right) \;.
\end{equation*}
By Lemma \ref{lem:Zspace} and Remark \ref{rmk-products-Z} $[H,X_i], P, P P_1, [H,X_i]P, P_1  \in \mathcal{Z}(S_z)$, so Lemma \ref{lem:lambdasquared} implies that
\begin{equation*}
\big| \tau\left(\iu [H,X_i] [[S_z,P],P] P_{1}\right) \big| \leq C \verti{[H,S_z]}^{2} \;.
\end{equation*}
The bound for the second term in \eqref{eqn:spuriousterms} follows similarly.
For the third term we notice that
\begin{align*}
\tau\( P[X_i,P][[P,S_z],X_j]\)&= \tau\( P[X_i,P][S_z,[X_j,P]]P\),    
\end{align*}
where we have used the cyclicity of the trace per unit volume and the Jacobi's identity; thus, by noticing that $P[X_i,P], P, [X_j,P]\in \cZ(S_z)$ thanks to Lemma \ref{lem:Zspace} and Remark \ref{rmk-products-Z}, and employing again Lemma \ref{lem:lambdasquared} estimate \eqref{eq: sigma-quasi-1} follows. 

We are left with proving \eqref{eq: sigma-quasi-2}. By Proposition \ref{thm:main} and by Lemma \ref{lem:sigmasc}, we have that
\begin{equation}
\label{eqn:diffsigmaschern}
\begin{aligned}
\Sigma\su{sc}_{ij} -\frac{1}{2\pi}\mbox{Chern}(P\su{sc})_{ij} 
& = \re \tau\left( \iu S_z P[[X_i,P],  [X_j,P]]\right) - \re \tau\left( \iu S_z P\su{sc}[[X_i,P\su{sc}],  [X_j,P\su{sc}]]\right)
\\    
& = \re \tau\left( \iu S_z (P-P\su{sc})[[X_i,P],  [X_j,P]]\right)\\
&\phantom{=}+\re \tau\left( \iu S_z P\su{sc}[[X_i,(P-P\su{sc})],  [X_j,P]]\right) \\
&\phantom{=}+\re \tau\left( \iu S_z P\su{sc}[[X_i,P\su{sc}],  [X_j,(P-P\su{sc})]]\right).
\end{aligned}
\end{equation}
By Remark \ref{rem:PofHsc}\ref{it:diffPPsc} and \eqref{eqn:decscsnc} it holds true that
\begin{equation}
\label{eqn:P-Psc}
P-P\su{sc}=\frac{\iu}{\pi} \oint_{\mathcal{C}} \di z\,(H - z )^{-1}[S_z,H]S_z  (H\su{sc} - z )^{-1}.
\end{equation}
Therefore, to estimate the first summand on the r.h.s.~of \eqref{eqn:diffsigmaschern} we notice that
\begin{align*}
 \tau\left(  S_z (P-P\su{sc})[[X_i,P],  [X_j,P]]\right)&= \frac{\iu}{\pi}\oint_{\mathcal{C}} \di z\,\tau\left(  S_z  (H - z )^{-1}[S_z,H]S_z  (H\su{sc} - z )^{-1}[[X_i,P],  [X_j,P]]\right).  
\end{align*}
Observe that the operator $S_z  (H\su{sc} - z )^{-1}[[X_i,P],  [X_j,P]]$ is in $\cZ(S_z)$ in view of Lemma \ref{lem:Zspace} and Remark \ref{rmk-products-Z}. Also the operators $H, S_z  (H - z )^{-1}\in \cZ(S_z)$ again by Lemma \ref{lem:Zspace}, thus Lemma \ref{lem:lambdasquared} implies that the first summand on the r.h.s.~of \eqref{eqn:diffsigmaschern} is bounded, up to a constant, by ${\verti{[H,S_z]}}^2$.
For the analysis of the second term on the r.h.s.~of \eqref{eqn:diffsigmaschern}, we note that by Leibniz's rule and \eqref{eqn:P-Psc}, we get that
\begin{align*}
[X_i,P-P\su{sc}]&= \frac{\iu}{\pi} \oint_{\mathcal{C}} \di z\Bigl([X_i,(H - z )^{-1}][S_z,H]S_z  (H\su{sc} - z )^{-1} \Bigr.\\
&\phantom{=}\Bigl.+(H - z )^{-1}][[H,X_i],S_z]S_z  (H\su{sc} - z )^{-1} +(H - z )^{-1}][S_z,H]S_z[X_i,  (H\su{sc} - z )^{-1}]\Bigr)\\
&=:T_1+T_2+T_3,  
\end{align*}
where the definitions of the three terms $T_1,T_2,T_3$ is understood with respect to the order of the appearance. Let us proceed by analyzing the first contribution for the second term on the r.h.s.~of \eqref{eqn:diffsigmaschern}:
\begin{equation}
\label{eqn:T1}
\begin{aligned}
\tau\left( \iu S_z P\su{sc}[T_1,  [X_j,P]]\right)&=\frac{\iu}{\pi} \oint_{\mathcal{C}} \di z\,\Bigl(\tau\left( \iu S_z P\su{sc} [X_i,(H - z )^{-1}][S_z,H]S_z  (H\su{sc} - z )^{-1} [X_j,P]\right)\Bigr.\\
&\phantom{=}\Bigl.-\tau\left( \iu S_z P\su{sc}[X_j,P] [X_i,(H - z )^{-1}][S_z,H]S_z  (H\su{sc} - z )^{-1}  \right)\Bigr).
\end{aligned}
\end{equation}
Since on the r.h.s~of the last equality the operators $S_z P\su{sc} [X_i,(H - z )^{-1}],\,H,\,S_z  (H\su{sc} - z )^{-1} [X_j,P]\in\cZ(S_z)$ and also the operators $S_z P\su{sc}[X_j,P] [X_i,(H - z )^{-1}],\,S_z  (H\su{sc} - z )^{-1}\in\cZ(S_z)$ by using again  Lemma \ref{lem:Zspace} and Remark \ref{rmk-products-Z}, then Lemma \ref{lem:lambdasquared} implies that the l.h.s.~of \eqref{eqn:T1} is bounded up to a constant by $\verti{[H,S_z]}^2$. The analysis of the contributions for the second term on the r.h.s.~of \eqref{eqn:diffsigmaschern} coming from $T_2$ and $T_3$ is analogous. The last summand on the r.h.s.~of \eqref{eqn:diffsigmaschern} can be estimated similarly.
\end{proof}

\section{Phase transition and lack of quantisation}
\label{sec:phase-transition}

\subsection{The extended Kane--Mele model}
\label{sec:extended Kane--Mele}

We consider a simple generalisation of the Kane–Mele model obtained by adding a next-to-nearest-neighbour Rashba interaction compatible with its symmetries, denoted by $H\sub{R2}$. The resulting \textit{extended Kane–Mele model} is characterised by the Hamiltonian $H\sub{KM} = H\sub{KM}(t,\lambda\sub{SO},w,\lambda\sub{R},r)$
\begin{equation}\label{eqn:hKM}
H\sub{KM} := t H\sub{NN}  + \lambda\sub{SO} H\sub{SO}+w H_{W}+\lambda\sub{R}(H\sub{R1} +r H\sub{R2}) \;,
\end{equation}
where $t>0,\lambda_\mathrm{SO} \in \mathbb{R} \setminus \{ 0\}$, $w,\lambda\sub{R},r \in \mathbb{R}$ are strength parameters\footnote{Note that both $t$ and $\lambda\sub{SO}$ are assumed to be non-zero otherwise the physics of the model is completely different.  Moreover, observe that our choice of parameters reflects the fact that $H\sub{R2}$ is  a second-order effect with respect to $H\sub{R1}$, so that if $H\sub{R1}$ is absent, so is $H\sub{R2}$.}, and where $H\sub{\sharp}$ for $\sharp \in \{ \mathrm{NN}, \mathrm{SO},\mathrm{W}, \mathrm{R1}, \mathrm{R2} \}$
are specified below, as operators acting on $\ell^{2}(\mathcal{C}) \otimes \mathbb{C}^{2}$:

\begin{itemize}
\item The operator $H_\mathrm{NN}$ is the standard nearest-neighbour hopping term representing the kinetic energy of the particles
$$
H_\mathrm{NN} := \sum_{j=1}^3 \left( T_{d_j} + T_{- d_j} \right) \otimes \Id_{\mathbb C^2}.
$$

\item The operator $H\sub{SO}$ describes a spin-orbit interaction, corresponding to an effective spin-dependent magnetic field due to an in-plane electric field, see \cite{Haldane88}. It is a next-to-nearest-neighbour hopping term of the form
$$
H\sub{SO} := - \iu \frac{\chi_{A} - \chi_{B}}{3\sqrt{3}}  \sum_{j=1}^3 \left( T_{a_j} - T_{-a_j} \right) \otimes \s_{z} \;,
$$
where $\chi_\sharp$, by abuse of notation, is the characteristic function on the lattice $\Gamma_\sharp$, $\sharp\in\{A,B \}$.

\item The operator $H\sub{W}$ is a staggered sub-lattice potential that distinguishes sites $A$ and $B$
$$
H_W := (\chi_{A} - \chi_{B})\otimes \Id_{\mathbb C^2} \;.
$$

\item Finally, the last contributions to the Hamiltonian are \emph{nearest-neighbour} and \emph{next-to-nearest-neighbour Rashba} terms, which describe a spin-orbit interaction due to an electric field orthogonal to the two-dimensional crystal (for example in a heterostructure): 
\begin{equation*}
\begin{split}
 H\sub{R1} & := \frac{\iu}{3} \sum_{j=1}^3\big( T_{ d_j} -  T_{-d_j} \big) \otimes (d_{j} \wedge \s)_{z} \;,
\\
H\sub{R2} & := \frac{\iu}{3} \sum_{j=1}^3\big( T_{ a_j} -  T_{-a_j} \big) \otimes (a_{j} \wedge \s )_{z} \;.
\end{split}
\end{equation*}
\end{itemize} 
We collect the Rashba terms into the \emph{Rashba Hamiltonian}
\begin{equation*}
H\sub{R}:=  H\sub{R1} + r H\sub{R2} \;.
\end{equation*}
\begin{remark}
The numerical prefactors $\frac{1}{3\sqrt{3}}$ and $\frac{1}{3}$ that appear in the definition of $H\sub{\sharp}$, for $\sharp \in \{\mathrm{SO}, \mathrm{R1}, \mathrm{R2} \}$ are there only for convenience of the subsequent analysis.
\end{remark}
\begin{remark}\label{rmk:H-non-spin-conserving}
Note that $H\sub{KM}$ is \emph{not spin-conserving}, that is $[H\sub{KM},S_{z}] \neq 0$ because of the Rashba Hamiltonian, indeed:
\begin{equation*}
[H-H\sub{R},S_z]=0\quad\text{and}\quad[H\sub{R},S_z]\neq 0 \; .
\end{equation*}
\end{remark}

\begin{lemma}\label{lem:KM} 
For any $t,\lambda\sub{SO},w,\lambda\sub{R},r \in \R$, the Hamiltonian $H\sub{KM}$ belongs to $\mathcal{P}_{0}^{\hexagon}(\Hi)$ and is inversion and time-reversal symmetric, compare with \eqref{eqn:Pi} and \eqref{eq:time-rev-def}.
\end{lemma}
\begin{proof}
Note that the operators $\chi_{A}$, $\chi_{B}$, and $T_{v}$ for any fixed $v \in \R^{2}$, are short-range, so is $H\sub{KM}$. 
Furthermore, $H\sub{KM}$ is periodic since $[T_{ v},T_{u}] =0$ for any vectors $v,u\in \R^2$ and $[\chi_\sharp, T_\gamma]=0$ for every $\sharp\in\{A,B\}$.

To verify that $H\sub{KM}$ is $2\pi/3$-rotationally symmetric, one notes that 
\begin{equation}\label{eqn:actionexp1}
\big[S_{j},S_{k}\big] = \ii \, \sum_{\ell} \eps_{j k \ell} S_{\ell} \;,
\end{equation}
$\eps_{j k \ell}$ being the totally antisymmetric tensor, so that
\begin{equation*}
\ex^{-\iu \frac{2\pi}{3}S_z}S_x =\(\frac{\sqrt{3}}{2}S_y-\frac{1}{2}S_x\)\ex^{-\iu \frac{2\pi}{3}S_z} \;,
\qquad
\ex^{-\iu \frac{2\pi}{3}S_z}S_y= -\(\frac{\sqrt{3}}{2}S_x+\frac{1}{2}S_y \)\ex^{-\iu \frac{2\pi}{3}S_z} \;.
\end{equation*}
We also observe that a translation operator $T_v$ and a rotation $R_{2 \pi /3}$ are intertwined by the following relation:
\begin{equation}\label{eqn:transrot}
R_{2 \pi /3}\,T_v=T_{\mathrm{R}_{2 \pi /3} v}\,R_{2 \pi /3}.
\end{equation}
Accordingly, setting for brevity $\sigma_{\pm} := (\sqrt{3} \sigma_{x} \pm \sigma_{y})/2$ we have that 
\begin{align*}
3R_{2\pi/3}H\sub{R1} &= \iu R_{2\pi/3}\Big[- \big( T_{ d_1} -  T_{-d_1} \big) \otimes \sigma_{+} +\big( T_{ d_2} -  T_{-d_2} \big) \otimes \sigma_{-} + \big( T_{ d_3} -  T_{-d_3} \big) \otimes \sigma_y \Big]\\
&=\iu \Big[\big( T_{ d_2} -  T_{-d_2} \big) \otimes \sigma_{-} + \big( T_{ d_3} -  T_{-d_3} \big) \otimes \sigma_y - \big( T_{ d_1} -  T_{-d_1} \big) \otimes \sigma_{+} \Big] R_{2\pi/3}\\
&=3H\sub{R1}R_{2\pi/3}.
\end{align*}
Analogously, one checks that $R_{2\pi/3}H\sub{R2}=H\sub{R2}R_{2\pi/3}$. Using again \eqref{eqn:transrot}, we obtain $[R_{2\pi/3},H\sub{NN}]=[R_{2\pi/3},H\sub{SO}]=0$, and obviously $[R_{2\pi/3}, H_v]=0$ as well.  
Hence, $H\sub{KM}$ is $2\pi/3$-rotation symmetric. 

Let us proceed to check that $H\sub{KM}$ is inversion symmetric. The relevant operator in this case is
\begin{equation*}
(\Pi_{1}\psi)(x):= \ee^{-\ii \pi S_{y}}\psi(\mathrm{R}_{v}^{-1} x) \;, \qquad \forall \psi \in \mathcal{H} \;,
\end{equation*}
where $\mathrm{R}_{v}$ is the restriction to $\mathcal{C}$ of the vertical reflection on $\mathbb{R}^{2}$, namely $\mathrm{R}_{v}(x_{1},x_{x})= (x_{1},-x_{2})$. It is straightforward to check that $\Pi_{1}$ satisfies \eqref{eqn:Pi} with $a=1$.
By \eqref{eqn:actionexp1}, we have
\begin{equation}\label{eqn:acthref}
\ex^{-\iu \pi S_y}S_x=-S_x\ex^{-\iu \pi S_y},\quad\ex^{-\iu \pi S_y}S_y=S_y\ex^{-\iu \pi S_y},\quad \ex^{-\iu \pi S_y}S_z=-S_z\ex^{-\iu \pi S_y}.
\end{equation}
Furthermore, we observe that the $\Pi_{1}$ and translation operators are intertwined by the following identity:
\begin{equation}\label{eqn:transref}
\Pi_{1}\,T_v=T_{\mathrm{R}\sub{v} v}\,\Pi_{1}.
\end{equation}
By identity \eqref{eqn:transref}, one easily verifies that $[H\sub{NN}, \Pi_{1}]=[H_v,\Pi_{1}]=0$. 
By using the third equality in \eqref{eqn:acthref} and \eqref{eqn:transref}, one has that $[H\sub{SO},\Pi_{1}]=0$.
By employing the first two equalities in \eqref{eqn:acthref} and again \eqref{eqn:transref}, one also has that 
\begin{align*}
3\Pi_{1}H\sub{R1} &= \iu \Pi_{1} \Big[ -\big( T_{ d_1} -  T_{-d_1} \big) \otimes \sigma_{+} + \big( T_{ d_2} -  T_{-d_2} \big) \otimes \sigma_{-} +\big( T_{ d_3} -  T_{-d_3} \big) \otimes \sigma_y \Big]\\
&=\iu \Big [ \big( T_{ d_2} -  T_{-d_2} \big) \otimes \sigma_{-} +\big( T_{ d_1} -  T_{-d_1} \big) \otimes\(-\sigma_{+}\) + \big( T_{ d_3} -  T_{-d_3} \big) \otimes \sigma_y  \Big] \Pi_{1}\\
&=3H\sub{R1}\Pi_{1} . 
\end{align*}
Similarly, one checks that $[H\sub{R2},\Pi_{1}]=0$.
Finally, analogous computations show that $[H\sub{KM},\Theta]=0$.
\end{proof}

\subsection{Phase diagram of the extended Kane--Mele model}
\label{sec:spectral_prop}

We begin by establishing some spectral properties that allow us to identify its insulator phases for suitable values of the parameters. 

Because of the $\Gamma$-periodicity, see Lemma \ref{lem:KM}, we switch to its Bloch--Floquet fibration
\begin{equation*}
 \F_i  \, H\sub{KM}\, \F_i ^{*}  = \int_{\T^2_*} ^{\bigoplus} \dd k \,H\sub{KM}(k) ,
\end{equation*}
where $H\sub{KM}(k)\equiv\left(\F_i  \, H\sub{KM} \, \F_i ^{*}\right)(k)  \in \mathbb{C}^{4 \times 4}$ for any $k \in \T^2_*$ with the Bloch--Floquet transform specified in Definition \ref{def:BF} and where, by abuse of notation, we fixed once for all the dimerisation in direction $i=3$, see \eqref{eq:Fourier-transf}. An explicit expression of $H\sub{KM}(k)$ is provided in the following lemma.

\begin{lemma}\label{lemma:Hk}
Let $H\sub{KM}$ be defined as in \eqref{eqn:hKM}. Then, for any $k \in \T^2_*$ we have 
\begin{equation}\label{eq:bloch-fibre-ham}
H\sub{KM}(k)= \left(
\begin{array}{cccc}
w +\alpha\sub{SO}(k) & \Omega(k) &\alpha\sub{R}(k) & \Omega_R(k) \\
\overline{\Omega(k)} & -w -\alpha\sub{SO}(k) & -\Omega_R(-k) & \alpha\sub{R}(k) \\
 \overline{\alpha\sub{R}(k)} & -\overline{\Omega_{R}(-k)} & w -\alpha\sub{SO}(k) & \Omega(k) \\
 \overline{\Omega_R(k)} & \overline{\alpha\sub{R}(k)} & \overline{\Omega(k)} & -w + \alpha\sub{SO}(k)
\end{array}
\right)\;,
\end{equation}
having set:
\begin{align*}\label{eq: functions_bloch_ham_1}
 \Omega(k) &:= t \sum_{j=1}^{3}\ee^{\ii k\cdot (d_{j} - d_{3})}\equiv t(\ee^{-\ii k\cdot a_2}+\ee^{\ii k\cdot a_1}+1)\;, \;\;   &
\Omega_R(k) & := \frac{\lambda_{\mathrm{R}}}{3} \sum_{j=1}^{3}\ee^{\ii\frac{2\pi}{3}(3-j)} \ee^{\ii k\cdot (d_{j}-d_{3})} \;,
\\
 \alpha\sub{SO}(k) &:= -
 \frac{2\lambda\sub{SO} }{3\sqrt{3}}\sum_{j=1}^{3}\,\sin (k\cdot a_{j})\;,
 &
\alpha\sub{R}(k) &:=  \frac{2 \lambda\sub{R} r}{\sqrt{3}} \sum_{j=1}^{3}\,\ee^{\ii\frac{2\pi}{3}(3-j)}\sin (k\cdot a_{j}) \;.
\end{align*}
\end{lemma}
\begin{proof}
First of all, we provide the dimerisation $[H\sub{KM}]_{{\mathcal{D}}_i}$. With abuse of notation, we let $T_{v}$ denote the translation operator on $\ell^{2}(\Gamma)$, whose definition is as in \eqref{eq:transl} with $\Gamma$ in place of $\mathcal{C}$. 
By straightforward computations, we obtain:
\begin{equation*}
\begin{split}
[H\sub{NN}]_{\mathcal{D}_i} & = \sum_{j=1}^{3} \begin{pmatrix}
 0 & T_{d_{i}-d_{j}}\\ T_{d_{j}-d_{i}} 
 & 0\end{pmatrix} \otimes \Id_{\mathbb{C}^{2}} \;,
 \\
[H\sub{SO}]_{\mathcal{D}_i}& = - \frac{\iu}{3 \sqrt{3}} \sum_{j=1}^{3} \begin{pmatrix}
T_{a_{j}} -T_{-a_{j}} & 0 \\
0 & T_{-a_{j}} -T_{a_{j}} 
\end{pmatrix} \otimes \s_{z} \;,
\\
[H\sub{W}]_{\mathcal{D}_i} & = \begin{pmatrix}
\Id_{\ell^{2}(\mathcal{C})} & 0\\
0 & -\Id_{\ell^{2}(\mathcal{C})}
\end{pmatrix} \otimes \Id_{\mathbb{C}^{2}} \;,
\end{split}
\end{equation*}
together with
\begin{equation*}
\begin{split}
[H\sub{R1}]_{\mathcal{D}_i} & =
\frac{\iu}{3} \sum_{j=1}^{3} \begin{pmatrix}
 0& -T_{d_{i}-d_{j}} \\ T_{d_{j}-d_{i}} 
 & 0\end{pmatrix} \otimes ( d_{j} \wedge \s)_{z} \;,
\\
[H\sub{R2}]_{\mathcal{D}_i} & =
 \frac{\iu}{3} \sum_{j=1}^{3} \begin{pmatrix}
T_{a_{j}} -T_{-a_{j}} & 0 \\
0 & T_{a_{j}} - T_{-a_{j}}  
\end{pmatrix} \otimes ( a_{j} \wedge \s)_{z} \;.
\end{split}
\end{equation*}
$\sigma= (\sigma_{x},\sigma_{y},\sigma_{z})$ being the vector of Pauli matrices, and where $\{ d_{j}\}_{j}$ and $\{ a_{j}\}_{j}$ are the vectors introduced in \eqref{eq:d-vectors} and \eqref{eq:a-vectors}.
Let $H\sub{\sharp}(k)$ denote the fibration of $H\sub{\sharp}$ via the Bloch--Floquet transform at the momentum $k\in\T^2_*$, with $\sharp \in \{\text{NN}, \text{SO},\text{W},\text{R1}, \text{R2} \}$, having fixed the dimerisation in with $i=3$. Because $\big(\F \,T_{v}\psi\big)(k) = \ee^{-\ii k \cdot v}\, \F\psi(k) $ for any $v\in\Gamma$, it follows that
\begin{equation*}
tH\sub{NN}(k) = \begin{pmatrix}
0 & \Omega(k) \\
\overline{\Omega(k)} & 0
\end{pmatrix} \otimes \mathbb{1}_{\C^{2}} \;,
\quad
\lambda\sub{SO}H\sub{SO}(k) = \begin{pmatrix}
\alpha\sub{SO}(k) & 0 \\
  0 & - \alpha\sub{SO}(k)
\end{pmatrix} \otimes \sigma_{z} \;.
\end{equation*}
The term $H\sub{W}(k)$ is trivial, so that we are left with determining $H\sub{R2}(k)$ and $H\sub{R1}(k)$. By simple computations, we obtain the identity
\begin{equation*}
\ii (\sigma \wedge d_{j})_{z} =  \begin{pmatrix}
0 & \ee^{\ii \frac{2 \pi}{3}(3-j)} \\
-\ee^{-\ii \frac{2 \pi}{3}(3-j)} & 0
\end{pmatrix} \qquad \quad j =1,2,3 \;,
\end{equation*}
which allows us to check the correctness of the formula for $\Omega\sub{R}(k)$. Similarly, we obtain
\begin{equation*}
{(\sigma \wedge a_{j})}_z = -\sqrt{3}\begin{pmatrix}
0 & \ee^{\ii \frac{2 \pi}{3}(3-j)} \\
\ee^{-\ii \frac{2 \pi}{3}(3-j)} & 0
\end{pmatrix} \qquad \quad j =1,2,3 \;,
\end{equation*}
and the formula for $\alpha\sub{R}$ ensues by simple manipulations.
\end{proof}
Before moving on, we note that after taking Bloch--Floquet transform, time-reversal symmetry acts as the following anti-unitary
\begin{equation}\label{eq:time-reversal-Floquet}
\big(\mathcal{F}_{i}\Theta \psi\big) (k) =\begin{pmatrix}
\mathbb{0}_2 & \mathbb{1}_{2}
\\
-\mathbb{1}_{2} & \mathbb{0}_2
\end{pmatrix} K \big(\mathcal{F}_{i} \psi\big) (-k) \;, \qquad k \in \mathbb{T}^{2}_{*} \;,
\end{equation}
where $K$ is complex conjugation. An important and straightforward consequence of time-reversal symmetry is the following lemma.
\begin{lemma}\label{lemma:time-reversal-symmetry}
We have $\mathrm{Spectrum}\big( H\sub{KM}(k)\big) = \mathrm{Spectrum}\big( H\sub{KM}( -k)\big) $.
\end{lemma}
\begin{proof}
By time-reversal symmetry, we have
\begin{equation*}
H\sub{KM}(k) = \begin{pmatrix}
\mathbb{0}_2 &\mathbb{1}_{2}
\\
-\mathbb{1}_{2} & \mathbb{0}_2
\end{pmatrix}  \overline{H\sub{KM}(-k)} \begin{pmatrix}
\mathbb{0}_2 & \mathbb{1}_{2}
\\
-\mathbb{1}_{2} & \mathbb{0}_2
\end{pmatrix}^{-1}
 \;.
\end{equation*}
Since $H\sub{KM}(k) = H\sub{KM}(k)^{*}$ and since $\begin{pmatrix}
\mathbb{0}_2 &\mathbb{1}_{2}
\\
-\mathbb{1}_{2} & \mathbb{0}_2
\end{pmatrix}$ is unitary, $H\sub{KM}(-k)^{\mathrm{T}}$ is unitarily equivalent to $H\sub{KM}(k)$, hence the claim.
\end{proof}

The analytical expression of the eigenvalues of the diagonalized Bloch--Floquet Hamiltonian $H\sub{KM}$ is cumbersome, and provides limited insight. However, $H\sub{KM}$ is block diagonal when $\lambda\sub{R}=0$, allowing for a simple analysis. Our strategy is to first study the model in detail at $\lambda\sub{R} = 0$ and 
then to show that the spectral properties remain stable as the Rashba coupling is introduced but remains sufficiently small.

When $\lambda\sub{R} = 0$, the model consists of two copies of the Haldane model \cite{Haldane88} with magnetic flux $\phi = \pm \frac{\pi}{2}$, related by the time-reversal transformation. In fact, introducing 
\begin{equation*}
H_{\sigma}(k) := \begin{pmatrix}
w+\sigma \alpha\sub{SO}(k) & \Omega(k) \\ 
\overline{\Omega(k)} & - w- \sigma \alpha\sub{SO}(k)
\end{pmatrix} \qquad \qquad \sigma = \pm \;,
\end{equation*}
for $\lambda\sub{R}=0$ we can write the fibered Hamiltonian $H\sub{KM}(k)$ in block-diagonal form
\begin{equation}\label{eq: spin-commuting-HKM}
H\sub{KM}(k) = \begin{pmatrix}
H_{+}(k) & \mathbb{0}_2
\\
\mathbb{0}_2 & H_{-}(k)
\end{pmatrix} \;.
\end{equation}
 We denote the energy bands of \eqref{eq:bloch-fibre-ham} at $\lambda\sub{R} = 0$ by $(\mathcal{E}^{w}_{\sigma,l})_{\sigma,l  = \pm}$, with $\mathcal{E}^{w}_{\sigma,l} : \mathbb{T}^2_{*} \to \R$ having the simple expression
\begin{equation*}
\mathcal{E}^{w}_{\sigma,l}(k)= l \sqrt{(w+\sigma\alpha\sub{SO}(k))^2 +|\Omega(k)|^2},
\end{equation*}
where $\sigma = \pm$ refers to the spin degrees of freedom, while the label $l =\pm$ refers to the sublattice ones, and where we made the dependence on $w \in \R$ explicit. The bands of each spin component have a degeneracy at zero energy only when both $\Omega(k)$ and $(w+\sigma \alpha\sub{SO}(k))$ are identically null. This happens at the so-called Dirac points 
\begin{equation}
\label{eqn:kF}
k_{F}^{\varepsilon}:= \frac{2 \pi}{3} \Big(1,  \frac{\varepsilon}{\sqrt{3}}\Big) \qquad \varepsilon  =\pm \;,
\end{equation}
and when $w = -\sigma \varepsilon \lambda\sub{SO}$. Note that on the Brillouin torus $\mathbb{T}^{2}_{*}$, we have $-k_{F}^{+}=k_{F}^{-}$.
\begin{figure}[t]
\includegraphics[scale=.56]{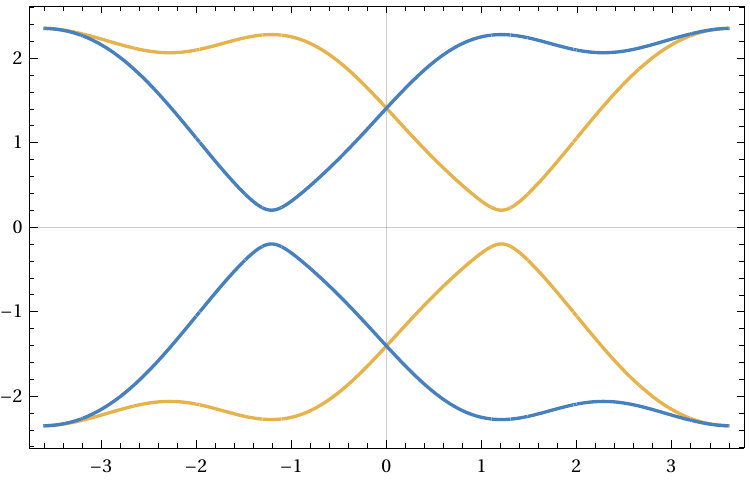}
$\quad$
\includegraphics[scale=.56]{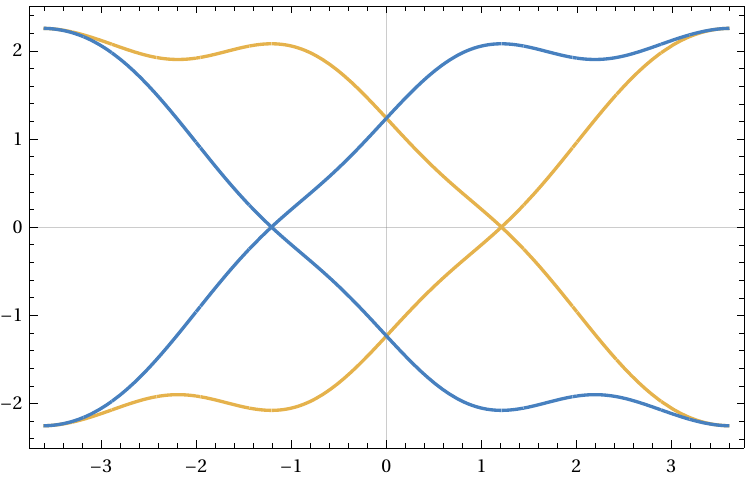}
\caption{Plot of the energy bands $\mathcal{E}^{w}_{l,\sigma}(k_1,k_2)$ along the line $(\frac{2 \pi}{3}, k_2)$ at $t= 2/3$, $\lambda\sub{SO} = 3\sqrt{3}/5 $, $\lambda\sub{R}  =  0$ and for $w>\lambda\sub{SO}$ and $w = \lambda\sub{SO}$, respectively on the left and on the right.
The blue and the yellow bands correspond to $\sigma = +$ and $\sigma= -$ respectively.
\label{fig:haldane}}
\end{figure}
Around the Dirac points, the spectrum is conical, in particular for $q\in\R^2$ with $\abs{q}$ sufficiently small
\begin{equation}\label{eq:conical-spectrum}
\mathcal{E}^{\mp\sigma \lambda\sub{SO} }_{\sigma,l}(k^{\pm}_{F}+q) = l \frac{3}{2}t |q| + O(|q|^{2}) \;,\qquad \sigma, l  =\pm 
\end{equation}
with the so-called Fermi velocity $\frac{3}{2}t$. To check this, observe that
\begin{equation}\label{eq:prod-kf-vectors}
k_{F}^{\pm} \cdot (d_{j} - d_{3}) = \pm \frac{2 \pi}{3} j \, ( \mathrm{mod}\; 2\pi)\;, \qquad k_{F}^{\pm} \cdot a_{j} = \mp \frac{2 \pi}{3}  \, ( \mathrm{mod}\; 2\pi) \;,
\end{equation}
from which one obtains
\begin{equation}\label{eq:expansion-1}
\alpha\sub{SO}(k_{F}^{\pm}+q) = \pm \lambda_{\mathrm{SO}} + O(|q|^{2}) \;, \qquad
\Omega(k_{F}^{\pm}+q) = \frac{3}{2}t (-\ii q_1 \pm q_2)+O(|q|^{2}) \;.
\end{equation}
In passing, as a consequence of \eqref{eq:prod-kf-vectors}, note the following expansions,
\begin{equation}\label{expansion-2}
\begin{split}
\alpha\sub{R}(k_{F}^{+}+q) & =  \frac{3}{2}  \lambda_{\mathrm{R}} r (\ii q_1 + q_2) + O(|q|^{2}) \;, \\
\Omega\sub{R}(k_{F}^{+}+q) &= \lambda_{\mathrm{R}} ( 1+\ii q_1 -\frac{3}{4}q_1^2- \frac{1}{4}q_2^2) + O(|q|^{3}) \;,\\
\Omega\sub{R}(k_{F}^{-}+q) & = \frac{1}{2}\lambda_{\mathrm{R}}(-\ii q_1+ q_2)+O(|q|^{2}) \;,
\end{split}
\end{equation}
which will be used in the following.
\begin{remark}\label{rmk:invert-no-rashba}
In particular, we stress that $ H\sub{KM}(k)$ is always invertible at $\lambda\sub{R} =0$ unless $k = k_{F}^{\pm}$. Furthermore, as a  straightforward consequence of \eqref{eq:conical-spectrum}, we know that for $q$ sufficiently small with $|q|\geq \rho>0$ it holds true that
\begin{equation*}
\sup _{q : |q|\geq \rho}
\Big\| \big(H\sub{KM}(k_{F}^{\pm}+q)\big|_{\lambda\sub{R} = 0} \big)^{-1}\Big\| \lesssim \rho^{-1} \;,
\end{equation*}
where the symbol $\lesssim$ means that the inequality holds up to a universal constant.
\end{remark}

Let us now consider the spectral properties of the Bloch--Floquet Hamiltonian at the Dirac points when $\lambda\sub{R} \neq 0$, without requiring it to be small.
\begin{lemma}\label{lem:gap}
Let $\mathcal{E}_1(k)\leq \cdots \leq \mathcal{E}_4(k)$ be the ordered eigenvalues of $H\sub{KM}(k)$ as in \eqref{eq:bloch-fibre-ham} and define the local internal spectral gap as
\begin{equation*}
\Delta(k):= \mathrm{dist}(\mathcal{E}_2(k),\mathcal{E}_3(k)) \;.
\end{equation*}
Then,
\begin{equation*}
\Delta (k_{F}^{+}) =
\begin{cases}
\abs{2|\lambda_{\mathrm{SO}}|- |w| -\sqrt{w^2 + \lambda_{\mathrm{R}}^{2}}} \qquad \quad  & \text{if}\quad|\lambda_{\mathrm{R}}|\leq 2\sqrt{\lambda_{\mathrm{SO}}^{2} + |w \lambda\sub{SO}|}, \\
2 |w| \qquad \quad  & \text{if}\quad |\lambda_{\mathrm{R}}|\geq 2\sqrt{\lambda_{\mathrm{SO}}^{2} + |w \lambda\sub{SO}|}.
\end{cases}
\end{equation*}
The condition $ \Delta(k_{F}^{+}) =0$ is satisfied only on the locus $w= w_{c}^{\pm}(\lambda\sub{SO},\lambda\sub{R})$, where the critical curves $w_{c}^{\pm}$ are defined by
\begin{equation}
\label{eqn:wc}
w_{c}^{\pm}(\lambda\sub{SO},\lambda\sub{R}) := \pm \bigg[ |\lambda\sub{SO}| - \frac{\lambda\sub{R}^{2}}{4|\lambda\sub{SO}|}\bigg]_{+}
\end{equation}
with $[\, \cdot \, ]_{+}:= \max\{\,\cdot\, ,0 \}$. Finally, defining the critical energy
\begin{equation*}
\mu_{c}:= \mathcal{E}_{2}(k_{F}^{+})|_{\Delta(k_{F}^{+})=0}  \;,
\end{equation*}
we have
\begin{equation}\label{eq:critical-energy}
\mu_{c} = 
\begin{cases}
- \frac{\lambda_{\mathrm{R}}^{2}}{4 \lambda_{\mathrm{SO}}} \qquad \quad  & \text{if}\quad|\lambda_{\mathrm{R}}|\leq 2|\lambda_{\mathrm{SO}}|, \\
- \lambda_{\mathrm{SO}} \qquad \quad  & \text{if}\quad|\lambda_{\mathrm{R}}|\geq 2|\lambda_{\mathrm{SO}}| \;.
\end{cases}
\end{equation}
\end{lemma}
\begin{figure}[t]
\includegraphics[scale=.6]{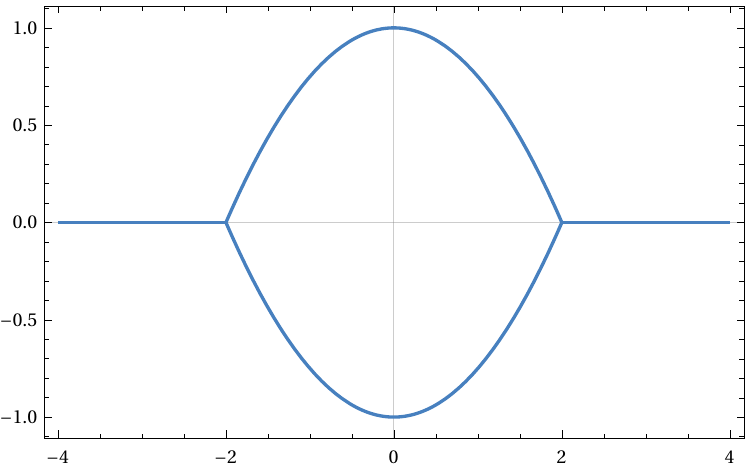}
\caption{Plot of $w_c^{\pm}/|\lambda\sub{SO}| $ as a function of $\lambda\sub{R}/|\lambda\sub{SO}|$.\label{fig:critical-curve}}
\end{figure}
\begin{remark}
Note that by Lemma \ref{lemma:time-reversal-symmetry} we have $\Delta(k_{F}^{+}) =\Delta(k_{F}^{-})$ and $\mu_{c} =  \mathcal{E}_{2}(k_{F}^{-})|_{\Delta(k_{F}^{-})=0} $.
Moreover, note that the spectral properties described in Lemma \ref{lem:gap} do not depend on $t$ and $r$.
\end{remark}
\begin{proof}
In view of Lemma \ref{lemma:Hk}, by using \eqref{eq:expansion-1} and \eqref{expansion-2} we obtain
\begin{equation*}
H\sub{KM}(k_F^+)= 
\begin{pmatrix}
 (w+ \lambda\sub{SO})\,\sigma_{z} &   \lambda\sub{R} \,\sigma_+ \\
 \, \lambda\sub{R} \,\sigma_- & (w - \lambda\sub{SO})\,\sigma_{z} 
\end{pmatrix}
\end{equation*}
where $ \sigma_\pm = (\sigma_{x} \pm \ii  \sigma_{y})/2$. By computing the characteristic polynomial, we obtain the following expressions for the energy bands, denoted by $\big(\mathcal{E}_{i,l}\big)_{i=1,2, l = \pm}$, at $k_{F}^{+}$:
\begin{equation*}
\begin{split}
\mathcal{E}_{1,\pm}(k_{F}^{+})&= \lambda\sub{SO}\pm \sqrt{w^2+\lambda\sub{R}^{2}}  \;, \qquad
\mathcal{E}_{2,\pm}(k_{F}^{+})= -\lambda\sub{SO} \pm |w|\;.
\end{split}
\end{equation*}
Note that the spin is not conserved and thus is not a good label for the energy bands. To compute the local internal gap $\Delta(k_{F}^{+})$, it is actually easier to first determine which bands are the external ones. Without loss of generality, we assume $\lambda\sub{SO} >0$: otherwise, we could flip the sign of $w$ and $\lambda\sub{R}$ to flip the sign of $H\sub{KM}(k)$, the latter operation leaving the set of internal/external bands invariant.
Note that $\mathcal{E}_{1,+}(k_{F}^{+})$ is always the positive external band. When $\mathcal{E}_{1,-}(k_{F}^{+})\geq\mathcal{E}_{2,-}(k_{F}^{+})$, the negative external band is $\mathcal{E}_{2,-}(k_{F}^{+})$, so that the local internal spectral gap is
\begin{equation*}
\Delta(k_{F}^{+})= |\mathcal{E}_{1,-}(k_{F}^{+})-\mathcal{E}_{2,+}(k_{F}^{+}) |  \;.
\end{equation*}
The condition $\mathcal{E}_{1,-}(k_{F}^{+})\geq\mathcal{E}_{2,-}(k_{F}^{+})$ is equivalent to $|\lambda_{\mathrm{R}}|\geq 2\sqrt{\lambda_{\mathrm{SO}}^{2} + |w |\lambda\sub{SO}}$.
In contrast, if $\mathcal{E}_{1,-}(k_{F}^{+})<\mathcal{E}_{2,-}(k_{F}^{+})$, $\mathcal{E}_{2,\pm}(k_{F}^{+})$ become the internal bands, and therefore $\Delta(k_{F}^{+}) =2|w|$. The expression for $w_{c}$ and $\mu _{c}$ is a straightforward computation from $\Delta(k_{F}^{+})$, noticing that $\mathcal{E}_{2,+}(k_{F}^{+})$ is always an internal band.
\end{proof}

The plot in Fig.~\ref{fig:critical-curve} does not generally say whether $(H\sub{KM},\mu _{c})$ is an insulator, but only when the gap between the internal bands closes at the Dirac points. However, if $|\lambda\sub{R}|$ is small enough, $(H\sub{KM}, \mu _{c})$ is indeed an insulator unless $w = w^{\pm}_{c}$, as shown in the following proposition.
\begin{proposition}\label{prop:robust-dirac-points}
Let $H\sub{KM} = H\sub{KM}(t,\lambda\sub{SO},w,\lambda\sub{R},r)$ be as in \eqref{eqn:hKM} and $\mu_c$, $w_c^{\pm}$ as in Lemma \ref{lem:gap}.
Let $t>0, \lambda\sub{SO} \neq 0, w, r \in \mathbb{R}$. If $|\lambda\sub{R}|$ is sufficiently small, depending on $t,\lambda\sub{SO},r$, then $(H\sub{KM}, \mu _{c})$ is an insulator unless $w = w_{c}^{\pm}$, in which case the spectral gap closes at the Dirac points only.
\end{proposition}
\begin{proof}
Equivalently, we prove that
the operator $H\sub{KM}(k) - \mu _{c} \Id_{4}$ is invertible for any $k\neq k_{F}^{\pm}$, and for $k = k_{F}^{\pm} $ but $w \neq w_{c}^{\pm}$.
Because of the time-reversal invariance, namely Lemma \ref{lemma:time-reversal-symmetry}, we can restrict ourselves to the upper-half Brillouin torus, $\mathbb{T}^{2}_{+}:=\{(k_{1},k_{2}) \in \T^2_*\,  | \,  k_{2}\geq 0\}$. We let $B_{\rho}(k):=\{q \in \mathbb{R}^{2} \,| \,|q-k|<\rho\}$ and $\widetilde{B}_{\rho}(k):=B_{\rho}(k) \setminus \{ k\}$.  By using the explicit expression for the energy bands at $\lambda\sub{R} =0$, see Remark~\ref{rmk:invert-no-rashba}, we know that for $\rho$ small
\begin{equation*}
\sup _{q \in \mathbb{T}^{2}_{+} \setminus B_{\rho}(k_{F}^{+})}
\Big\| \big(H\sub{KM}(k)\big|_{\lambda\sub{R} = 0} - \mu _{c} \Id_{4}\big)^{-1}\Big\| \lesssim \rho^{-1} \;.
\end{equation*}
We use the identity 
\begin{equation}
\label{eqn:vonneumann}
{(A+B)}^{-1}=A^{-1}{(\Id+BA^{-1})}^{-1}    
\end{equation}
with $A = H\sub{KM}(k)\big|_{\lambda\sub{R} = 0} - \mu _{c} \Id_{4}$, $B = \lambda\sub{R} H\sub{R}$ and the smallness of $BA^{-1}$, to obtain
\begin{equation*}
\begin{split}
\sup _{k \in \mathbb{T}^{2}_{+} \setminus B_{\rho}(k_{F}^{+})}
\Big\| \big(H\sub{KM}(k) - \mu _{c} \Id_{4}\big)^{-1}\Big\| 
&
\leq  \sup _{k \in \mathbb{T}^{2}_{+}\setminus B_{\rho}(k_{F}^{+})}\frac{\Big\| \big(H\sub{KM}(k)\big|_{\lambda\sub{R} = 0} - \mu _{c} \Id_{4}\big)^{-1}\Big\|}{1 -|\lambda\sub{R}|\| H\sub{R}\| \Big\| \big(H(k)\big|_{\lambda\sub{R} = 0} - \mu _{c} \Id_{4}\big)^{-1} \Big\|}
\\
&\lesssim \rho^{-1}
\end{split}
\end{equation*}
provided that $|\lambda\sub{R}|$ is small enough depending on $\rho$, where we used that $H_{R}$ is a bounded operator. Now, we proceed by investigating the invertibility of $H\sub{KM}(k)-\mu_c\Id_{4}$ inside the pierced ball $\widetilde{B}_{\rho}(k_F^+)$. 
We note that if $|w|$ is large enough the Hamiltonian $H\sub{KM}(k) - \mu _{c}\Id_{4} $ is invertible for any $k \in \T^2_*$. This follows by using equality \eqref{eqn:vonneumann} again, where $A=w H\sub{W}$ and $B=H\sub{KM}-w H\sub{W}$.
Therefore, the claim of the proposition follows if we can prove that $H\sub{KM}(k_{F}^{+}+q) - \mu _{c}\Id_{4} $ is invertible for $q \in \widetilde{B}_{\rho}(0)$ with $\rho>0$ small enough, for $w$ in a compact set and for $|\lambda\sub{R}|$ small enough, depending only on $t$, $\lambda\sub{SO}$ and $r$. To this end, we compute the determinant and obtain:
\begin{equation}
\label{eqn:detHKM}
\begin{split}
&\det\Big( H\sub{KM}(k_{F}^{+}+q) - \mu _{c} \Big) 
\\
&\,= \Big(w^2 - \Big(\lambda\sub{SO} - \frac{\lambda\sub{R}^{2}}{4\lambda\sub{SO}} \Big)^{2}\Big)^{2}+2 |q|^{2}\Big(w^2 + \lambda\sub{SO}^{2}+O_{w,\lambda\sub{SO},r}(\lambda\sub{R}^{2})\Big) + o_{w,\lambda\sub{SO},\lambda\sub{R},r}(|q|^2),
\end{split}
\end{equation}
where, without loss of generality we set $t = 2/3$ (compare with \eqref{eq:conical-spectrum}).
This expression was obtained by expanding the matrix elements of $H\sub{KM}$ see \eqref{eq:bloch-fibre-ham} as in \eqref{eq:expansion-1} and \eqref{expansion-2}.
If $w\neq w_c^{\pm}$ then the determinant is obviously non-zero since the first term in \eqref{eqn:detHKM} does not vanish. Otherwise, we need to analyze the second-order term in $\abs{q}$. By inspection of $H\sub{KM}(k^{+}_{F}+q)$, we have that $O_{w,\lambda\sub{SO},r}(\,\cdot \,)$ and $o_{w,\lambda\sub{SO},\lambda\sub{R},r}(\,\cdot \,)$ depend continuously on the parameters provided that $\lambda\sub{SO} \neq 0$. Therefore, since $\lambda\sub{SO}$ is fixed and $w$ is in a compact set, we can choose $|\lambda\sub{R}|$ small enough so that the term proportional to $|q|^{2}$ is non-zero and, if $\rho$ is small enough (recall $q \in \widetilde{B}_{\rho}(0)$), dominates the term $o_{w,\lambda\sub{SO},\lambda\sub{R},r}(|q|^2)$. This implies that $\det\Big( H\sub{KM}(k_{F}^{+}+q) - \mu _{c}\Id_{4} \Big) \neq 0$ unless $q=0$ and the claim is proven.
\end{proof}

We have thus shown that for $|\lambda\sub{R}|$ small, the plot in Fig.~\ref{fig:critical-curve} describes the quantum phase diagram of the model, consisting of three insulating phases separated by a semi-metallic one along the critical curves $w_{c}^{\pm}$. This phase diagram was already obtained for the Kane--Mele model ($r=0$) in the seminal work \cite[Fig. 1, inset]{KaneMele2005}. In particular, the central phase is a topological insulator with spin Chern number equal to $1$, whereas the outer regions have spin Chern number equal to $0$. This can be directly inferred by continuity of the spin Chern number for insulators with respect to $\lambda\sub{R}$, and thus by considering the case $\lambda\sub{R} = 0$, see \eqref{eq: spin-commuting-HKM}, for which the spin Chern number is equal to the Chern number of $H_{+}$ in \eqref{eq: spin-commuting-HKM}, which was computed in \cite[Fig. 2 at $\phi=\pi/2$]{Haldane88}.

\subsection{Lack of quantisation}
\label{subsec: non-univers}

Unlike the case of charge transport, the spin conductivity as introduced in Definition \ref{def-spin-cond} is not universally quantised for time-reversal symmetric insulators. 
\begin{theorem}\label{thm:main2}
There exist insulators $(H,\mu)$ with $H$ belonging to $\mathcal{P}_{0}^{\hexagon}(\Hi)$, time-reversal symmetric and almost conserving the spin (in the sense of Definition \ref{def:almspincon}) such that $\sigma^{\mathrm{s}}_{12}(H,\mu) \notin \frac{1}{2 \pi} \mathbb{Z}$. In particular, for such insulators it holds true that
\begin{equation}\label{eq:quadratic-corrections}
c \verti{[H,S_z]}^{2} \leq 
\Big| \sigma^{\mathrm{s}}_{12} - \frac{1}{2\pi}\mbox{$S$-Chern}(P\su{sc})_{12}\Big| \leq C \verti{[H,S_z]}^{2} \;,
\end{equation}
for some constants $0<c\leq C$ independent of $\verti{[H,S_z]}$, where the norm $\verti{\,\cdot\,}$ is introduced in \eqref{eqn:verti}.
\end{theorem}
\begin{remark}\label{rmk:thm2}
Consider time-reversal symmetric insulators in $\mathcal{P}_{0}^{\hexagon}(\Hi)$ class. Recall from Section \ref{sec: spin-non-conserving} that for spin-conserving insulators, one has $\sigma^{\mathrm{s}}_{12} = \frac{1}{2\pi}\mbox{$S$-Chern}(P\su{sc})_{12} \in \frac{1}{2 \pi} \mathbb{Z}$, see \eqref{eq:quantisation-spin-cons}. Theorem \ref{cor:main1} affirms that, for insulators almost conserving the spin, the spin conductivity $\sigma^{\mathrm{s}}_{12}$ may deviate from the quantised value $\frac{1}{2 \pi} \mathbb{Z}$
by correction of order $\verti{[H,S_z]}^{2}$. Theorem \ref{thm:main2} states that there are time-revarsal symmetric insulators in $\mathcal{P}_{0}^{\hexagon}(\Hi)$, for which no cancellation takes place, so that such corrections are exactly of order $\verti{[H,S_z]}^{2}$. Since $\mbox{$S$-Chern}(P\su{sc})_{ij} \in \mathbb{Z}$ and $\verti{[H,S_z]} $ is small, Theorem \ref{thm:intro}\ref{it:iii} then follows.
\end{remark}
To prove Theorem \ref{thm:main2}, we consider insulators of the form $(H\sub{KM},\mu_{c})$, where $H\sub{KM}$ is the Hamiltonian of the extended Kane--Mele model defined in \eqref{eqn:hKM}, and where $\mu _{c} = -\frac{\lambda\sub{R}^{2}}{4\lambda\sub{SO}}$ is the critical energy  for $|\lambda\sub{R}|$ small enough, see \eqref{eq:critical-energy}, and show that the spin conductivity has a non-universal jump discontinuity across the critical line $w_{\mathrm{c}}^{+}$, see Lemma \ref{lem:gap}, provided that $|\lambda\sub{R}|$ is small enough. We quantify how far $w$ is from the critical line $w_{\mathrm{c}}^{+}$ by introducing $m=m(w,\lambda\sub{SO},\lambda\sub{R})$
\begin{equation}\label{eq: def-mass}
m:= w - w_{\mathrm{c}}^{+}(\lambda\sub{SO},\lambda\sub{R}) \;.
\end{equation}
With slight abuse of notation, we will write
\begin{equation}\label{eq: sigma-parameters}
\sigma^{\mathrm{s}}_{12}(t,\lambda\sub{SO},m,\lambda\sub{R},r) 
\equiv 
\sigma^{\mathrm{s}}_{12}\big( H\sub{KM}(t,\lambda\sub{SO},w_{\mathrm{c}}^{+}(\lambda\sub{SO},\lambda\sub{R}) + m,\lambda\sub{R},r),\mu _{c}\big) \;,
\end{equation}
and define the spin conductivity jump across the upper critical curve $w_{c}^{+}$ as follows
\begin{equation}\label{def-critical-jump}
\delta \sigma_{12}^{\mathrm{s}} (t,\lambda\sub{SO},\lambda\sub{R},r):= \Big( \lim _{m \to 0^{+}} - \lim _{m \to 0^{-}}\Big) \sigma^{\mathrm{s}}_{12}(t,\lambda\sub{SO},m,\lambda\sub{R},r)\;,
\end{equation}
provided that the limit exists. Note that by Proposition \ref{prop:antisymmetry} and Lemma \ref{lem:KM} , it suffices to consider only $\sigma^{s}_{12}$, as $\sigma^{s}$ is an antisymmetric tensor. Note that we could equivalently study the jump across the lower critical line $w_c^-$, compare with \eqref{eqn:wc}.

\begin{remark}
The quantity \eqref{def-critical-jump} simply measures the jump discontinuity of the spin conductivity across the critical curve and possibly identifies a first-order quantum phase transition, in the statistical mechanics sense, between the trivial and the topological insulator, compare with discussion below the proof of Proposition \ref{prop:robust-dirac-points}.
\end{remark}

Quite surprisingly, the discontinuity can be computed exactly, compare with \cite{Porta16}, where a similar computation is carried out for the charge conductivity of interacting fermionic systems.
\begin{proposition}\label{prop-discontinuity}
Let $t>0$, $\lambda\sub{SO} \neq 0$ and $r \in \R$. Then, if $|\lambda\sub{R}|$ is small enough depending on $t$, $\lambda\sub{SO}$ and $r$, the limits 
\begin{equation*}
\lim _{m \to 0^{\pm}} \sigma^{\mathrm{s}}_{12}(t,\lambda\sub{SO},m,\lambda\sub{R},r)
\end{equation*}
exist. Consequently, the spin conductivity jump $\delta \sigma_{12}^{\mathrm{s}} = \delta \sigma_{12}^{\mathrm{s}} (t,\lambda\sub{SO},\lambda\sub{R},r)$, as defined in \eqref{def-critical-jump}, also exists and is given by 
\begin{equation}\label{eq:spin-discontinuity}
\delta \sigma_{12}^{\mathrm{s}} = -\frac{1}{2 \pi} \left( 1 +\frac{ \lambda\sub{R}^{2} r}{2t\lambda\sub{SO} - \lambda\sub{R}^{2}r}  \right)\;.
\end{equation}
In particular, for the standard Kane--Mele model ($r=0$) it holds that $\delta \sigma_{12}^{\mathrm{s}} = -\frac{1}{2 \pi}$.
\end{proposition}
\begin{remark}
Even though $\delta \sigma^{\mathrm{s}}_{12} = - \frac{1}{2\pi}$ at $r=0$, we do not expect the spin conductivity to be quantised in the Kane--Mele model, and there is numerical evidence \cite{MoUl} supporting this thesis.
Additionally, note that when $ \lambda\sub{SO} \to 0$ there is no jump across the critical curve $w_{c}$; in fact, the qualitative properties of the model change drastically, as could be seen by diagonalising the fiber Hamiltonian \eqref{eq:bloch-fibre-ham}.
\end{remark}
It is straightforward to see that Theorem \ref{thm:main2} is a simple consequence of the above proposition.
\begin{proof}[Proof of Theorem \ref{thm:main2}]
Assume the contrary, that is, $\sigma^{\mathrm{s}}_{12}(H,\mu) \in \frac{1}{2 \pi} \mathbb{Z} $ for any insulator $(H,\mu)$ with $H$ time-reversal symmetric and belonging to $\mathcal{P}_{0}^{\hexagon}(\Hi)$. Then, $\delta \sigma_{12}^{\mathrm{s}}(t,\lambda\sub{SO},\lambda\sub{R},r) \in \frac{1}{2 \pi} \mathbb{Z}$ for any value of $t,\lambda\sub{SO},\lambda\sub{R},r$, which is in contradiction with \eqref{eq:spin-discontinuity}. 

We are left with proving the lower bound in \eqref{eq:quadratic-corrections}, the upper bound being stated in Theorem \ref{cor:main1}. Once more, we do this by considering the pair $(H\sub{KM},\mu_{c})$ at fixed $t,\lambda\sub{SO},r$ non-zero, at $\lambda\sub{R} > 0$ sufficiently small and at $w$ sufficiently close to $w_{c}^{+}(\lambda\sub{SO},\lambda\sub{R})$, that is, for $m \neq 0$, $m$ sufficiently small. Since all parameters but $m$ and $\lambda\sub{R}$ are fixed, let us abridge \eqref{eq: sigma-parameters} to $\sigma_{12}^{s}(m,\lambda\sub{R})$, and likewise $P\su{sc}= P\su{sc}(m)$, $w^{+}_{c} = w^{+}_{c}(\lambda\sub{R})$, and $\delta\sigma_{12}^{s}=\delta\sigma_{12}^{s}(\lambda\sub{R})$ (notice that $H\su{sc}=H\sub{KM}|_{\lambda\sub{R}=0}$).

We introduce the function
\begin{equation*}
f(m,\lambda\sub{R}):= \sigma_{12}^{s}(m,\lambda\sub{R}) -
\frac{1}{2\pi}\mbox{$S$-Chern}\big(P\su{sc}(m)\big)_{12} \;
\end{equation*}
for $m\neq 0$ and $\lambda_R>0$, both sufficiently small.
We observe that $\verti{H\su{snc}} = O(\lambda\sub{R})$ (see Remark \ref{rmk:H-non-spin-conserving} and \eqref{eqn:decscsnc}), so that the claim we need to prove is that there exist a $m\neq 0$, a constant $C>0$ and an open interval $J\ni 0$, such that $|f(m,\lambda\sub{R})| \geq C \lambda\sub{R}^{2}$ holds true for all $\lambda\sub{R} \in J $. 

We now assume the negation of the claim and show that this leads to a contradiction. We first of all note that the negation of the claim implies that
\begin{equation*}
f(m,\lambda\sub{R}) = o(\lambda\sub{R}^{2}) .
\end{equation*}
In fact, the negation of the claim is that for all $m\neq 0$, constants $C>0$, and open intervals $J \ni 0$, there exists at least one point $\lambda_0 \in J \setminus \{ 0\}$, such that $|f(m,\lambda_0)|< C \lambda_0^{2} $. Since this holds for any open interval $J \ni 0$, we can iterate the reasoning and find a sequence $(\lambda_n)_n\in \mathbb{N}$ such that $\lambda_n \to 0$ and $|f(m,\lambda_n)|< C \lambda_n^{2}$. By Remark \ref{rmk:cont}, the limit $\lim_{\lambda\sub{R}\to 0}f(m,\lambda\sub{R})$ exists for any $m\neq0$. As the constant $C>0$ can be chosen arbitrarily small,
\begin{equation*}
\lim_{\lambda\sub{R} \to 0}f(m,\lambda\sub{R})/\lambda\sub{R}^2 = 0,
\end{equation*}
that is, $f(m,\lambda\sub{R}) = o(\lambda\sub{R}^{2})$.

It remains to show that the condition $f(m,\lambda\sub{R}) = o(\lambda\sub{R}^{2})$ leads to a contradiction. To this end, we define
\begin{equation*}
g(m,\lambda\sub{R}):=f(m,\lambda\sub{R}) - f(-m,\lambda\sub{R}) \;, \qquad m>0 \;.
\end{equation*}
Since $f=o(\lambda\sub{R}^{2})$ for all $m\neq 0$ and $\lambda\sub{R}>0$ small enough, we also have $g = o(\lambda\sub{R}^{2})$. However, we can compute $g$ more precisely by means of Proposition \ref{prop-discontinuity} and prove a contradiction. First of all, we compute the spin Chern number of $P\su{sc}$. The model at $\lambda\sub{R} = 0$ consists of two copies of the Haldane model, see \eqref{eq: spin-commuting-HKM} and text around it. Then, as noted at the end of Section \ref{sec:spectral_prop}  it is straightforward to see that at $\lambda\sub{R} = 0$, for $m$ small
\begin{equation}\label{eq:free-spin-chern}
\mbox{$S$-Chern}(P\su{sc}(m))_{12} = 
\begin{cases}
0 & \qquad  m >0 \\
1  & \qquad  m< 0 \;.
\end{cases}
\end{equation}
Accordingly, by using this and Proposition \ref{prop-discontinuity} we have
\begin{equation*}
\begin{split}
g(m,\lambda\sub{R}) & = \frac{1}{2\pi} + \sigma_{12}^{s}(m,\lambda\sub{R}) -\sigma_{12}^{s}(-m,\lambda\sub{R})
\\
&= \frac{1}{2\pi} +  \delta\sigma_{12}^{s}(\lambda\sub{R}) + \Big(\sigma_{12}^{s}(m,\lambda\sub{R}) - \sigma_{12}^{s}(0^{+},\lambda\sub{R}) \Big) - \Big(\sigma_{12}^{s}(-m,\lambda\sub{R}) - \sigma_{12}^{s}(0^{-},\lambda\sub{R}) \Big)
\\
& = -\frac{1}{2\pi}\frac{ \lambda\sub{R}^{2} r}{2t\lambda\sub{SO} - \lambda\sub{R}^{2}r} + \Big(\sigma_{12}^{s}(m,\lambda\sub{R}) - \sigma_{12}^{s}(0^{+},\lambda\sub{R}) \Big) - \Big(\sigma_{12}^{s}(-m,\lambda\sub{R}) - \sigma_{12}^{s}(0^{-},\lambda\sub{R}) \Big) \;.
\end{split}
\end{equation*}
We now show that  $\sigma_{12}^{s}( \pm m,\lambda\sub{R}) - \sigma_{12}^{s}(0^{\pm},\lambda\sub{R}) =o(\lambda\sub{R}^{2})$. To see this, for any $m'=m'(\lambda\sub{R}) >0$ small enough, we write, using \eqref{eq:free-spin-chern} and the definition of $f$
\begin{equation*}
\begin{split}
\sigma_{12}^{s}(  m,\lambda\sub{R}) - \sigma_{12}^{s}(0^{+},\lambda\sub{R}) & = \sigma_{12}^{s}(  m,\lambda\sub{R}) - \sigma_{12}^{s}(  m',\lambda\sub{R}) +\sigma_{12}^{s}(  m',\lambda\sub{R}) - \sigma_{12}^{s}(0^{+},\lambda\sub{R})
\\
& = f(m,\lambda\sub{R}) - f(m',\lambda\sub{R}) + \sigma_{12}^{s}(  m',\lambda\sub{R}) - \sigma_{12}^{s}(0^{+},\lambda\sub{R})
\\
& = o(\lambda\sub{R}^{2}) \;,
\end{split}
\end{equation*}
where we used that $\sigma_{12}^{s}(  m',\lambda\sub{R}) - \sigma_{12}^{s}(0^{+},\lambda\sub{R})$ can be made $o(\lambda\sub{R}^2)$ for $m'$ small enough depending on $\lambda\sub{R}$. Then, by inspection, for all $m>0$ we get that $|g(m,\lambda\sub{R})| \geq C \lambda\sub{R}^{2}$, for some constant $C>0$ and $\lambda\sub{R}$ small enough, which is a contradiction.
\end{proof}

To prove Proposition \ref{prop-discontinuity}, we use the following imaginary-time representation of the spin conductivity. 
\begin{theorem}\label{prop:imaginary-time-repr}
Let $(H,\mu)$ be an insulator with $H \in \mathcal{P}_{0}^{\hexagon}(\Hi)$. Then, for any $i,j=1,2 $ the spin conductivity  $\sigma^{\mathrm{s}}_{ij}=\sigma^{\mathrm{s}}_{ij}(H,\mu)$ can be written as
\begin{equation}\label{almost-pontryagin}
\begin{split}
\sigma^{\mathrm{s}}_{ij}
= - \int_{\R \times \mathbb{T}_{*}^{2}} \frac{\dd \mathbf{k}}{(2 \pi)^{3}} \Tr _{\C^{4}} 
\Big(A^{\mathrm{s}}_{i}(\mu;\mathbf{k})A_{0}(\mu;\mathbf{k})A_{j}(\mu;\mathbf{k})\Big)
\end{split}
\end{equation}
where $\mathbf{k}=(k_0,k) \in \R \times \mathbb{T}_{*}^{2}$ and where for $\nu = 0,1,2$ we have introduced the quantities
\begin{align*}
A_{\nu}(\mu;\mathbf{k})&:=\left[ \partial _{k_{\nu}}\big(H(k) - (\mu -\ii k_{0} ) \mathbb{1} \big)\right]\widehat{G}(\mu;\mathbf{k}) \;,
\\
A^{\mathrm{s}}_{\nu}(\mu;\mathbf{k})&:=\frac{1}{2}\left[ \partial _{k_{\nu}}\big(H(k) - (\mu -\ii k_{0} ) \mathbb{1} \big) S_{z}+S_{z}\partial _{k_{\nu}}\big(H(k) - (\mu -\ii k_{0} ) \mathbb{1} \big)\right]\widehat{G}(\mu;\mathbf{k}) \;,
\end{align*}
with $\widehat{G}(\mu;\mathbf{k})$ defined by
\begin{equation*}
\widehat{G}(\mu;\mathbf{k}):= \big(H(k) - (\mu -\ii k_{0} ) \mathbb{1} \big)^{-1}  \;.
\end{equation*}
\end{theorem}
\begin{remark}\label{rmk-pontryagin-class}
Note that $A_{\nu}$ is the connection associated with $\widehat{G}$. If there were $A_{\nu}$ in place of  $A^{\mathrm{s}}_{\nu}$, the quantity on the r.h.s.~of \eqref{almost-pontryagin} would be proportional to the Pontryagin index associated with the connection $A_{\nu}$ \cite{Coleman,Jackiw}. In general, the presence of the spin operator in $A^{\mathrm{s}}_{\nu}$ strips formula \eqref{almost-pontryagin} of any geometrical interpretation, to the best of our knowledge.
Finally, note that $\widehat{G}(\mu;\mathbf{k})$ is the Fourier transform of the so-called imaginary-time-ordered Green function, see Appendix \ref{app:imaginary}.
\end{remark}
The proof of this representation is presented in Appendix \ref{app:imaginary}. 
The idea is to switch to the grand canonical formulation of the Kubo formula in \eqref{Kubo-formula}, and to use the Kubo--Martin--Schwinger (KMS) property to perform the Wick rotation. We learned this strategy in \cite{Porta17,Porta20} for the charge conductivity, see also \cite{MaPo} for applications to quantum spin chains.

With the aid of formula \eqref{almost-pontryagin} we shall now prove Proposition \ref{prop-discontinuity}.

\begin{proof}[Proof of Proposition \ref{prop-discontinuity}]
For the sake of brevity, we drop the explicit dependence on the fixed parameters $t,\lambda\sub{SO}$, $\lambda\sub{R}$, $r$. Because of time-reversal symmetry (see Lemma \ref{lemma:time-reversal-symmetry}), we restrict to the upper-half Brillouin torus, $\mathbb{T}^{2}_{+}:=\{(k_{1},k_{2}) \in \T^2_*\,  | \,  k_{2}\geq 0\}$, the contribution coming from the negative one being the same, that is, we have by Theorem \ref{prop:imaginary-time-repr}
\begin{equation}
\label{eqn:AAA-int}
\begin{split}
\sigma^{\mathrm{s}}_{12}
= -2 \int_{\R \times \mathbb{T}_{+}^{2}} \frac{\dd \mathbf{k}}{(2 \pi)^{3}} \Tr _{\C^{4}} 
\Big(A^{\mathrm{s}}_{1}(\mu_{c};\mathbf{k})A_{0}(\mu_{c};\mathbf{k})A_{2}(\mu_{c};\mathbf{k})\Big) \;.
\end{split}
\end{equation}
Note that, if the integrand in \eqref{eqn:AAA-int} were well-behaved, by the dominated convergence theorem we could take the limit inside the integral, thus proving that the limits $m\to 0^{\pm}$ exist and that $\delta \sigma_{12}^{\mathrm{s}}= 0$. However, the integrand has an non-integrable divergence, which makes the existence of the limit non-trivial and $\delta \sigma_{12}^{\mathrm{s}}\neq 0$. To capture the singularity of the integrand, we introduce additional notation. By Proposition \ref{prop:robust-dirac-points} we know that $(H\sub{KM}(k)-\mu)^{-1}$ can be singular only at the Fermi points $k_F^\pm$, defined in \eqref{eqn:kF}; thus the only singular point in $\mathbb{R} \times \mathbb{T}^{2}_{+}$ for $\widehat{G}(\mu;\mathbf{k})$ is $\mathbf{k}_{F}^{+} := (0,k_{F}^{+}) \in \mathbb{R} \times \mathbb{T}^{2}_{+}$. 
We extract the most singular contribution from the Green function around $\mathbf{k}^{+}_{F}$. We use the notation $\mathbf{q}=(q_{0},q) \in \mathbb{R} \times \mathbb{R}^{2}$, and denote $|\, \cdot \, |$ the Euclidean norm. Recalling that $m = w - w_{c}^{+}$, see \eqref{eq: def-mass}, and using the notation $O_{m,\mathbf{q}}(1)$ to denote terms that remain bounded as $|m|, |\mathbf{q}| \to 0$, we can write
\begin{equation*}
\widehat{G}(\mu _{c};\mathbf{k}_{F}^{+} +\mathbf{q}) = S(m,\mathbf{q}) + O_{m,\mathbf{q} }(1) \;, \qquad \big |S(m,\mathbf{q}) \big| \lesssim \big(|\mathbf{q}|+ |m|\big)^{-1} \;,
\end{equation*}
which, by direct computations with \eqref{eq:bloch-fibre-ham}, \eqref{eq:expansion-1} and \eqref{expansion-2}, is given by 
\begin{equation*}
S(m,\mathbf{q}) = \chi(m,\mathbf{q})^{-1} \begin{pmatrix}
z_{1}(\ii q_0+m)  & 0 & v_{1}(\ii q_1 + q_2)  &  z_{2} (m+\ii q_0)  \\ 
0 & 0 & 0 & 0\\
v_{1} (-\ii q_1 + q_2) & 0 & \ii z_{3} q_0 + z_{4} m & v_{2}(\ii q_1 + q_2) \\
z_{2} (m+\ii q_0) & 0 & v_{2}(-\ii q_1 + q_2) & z_{5} (\ii  q_{0}+m) 
\end{pmatrix} ,
\end{equation*}
where:
\begin{equation*}
\alpha_{\pm}:= 4 \lambda\sub{SO}^{2} \pm \lambda\sub{R}^{2} \;, \qquad \widetilde{\alpha}:= 3 t \lambda\sub{SO} -\frac{3}{2} \lambda\sub{R}^{2}r \;,
\end{equation*}
\begin{align*}
z_{1} &= - \frac{\lambda\sub{R}^{2}}{4 \lambda\sub{SO}^{2} }\alpha_{-} \;,& 
z_{2} & := \frac{\lambda\sub{R}}{2 \lambda\sub{SO}} \alpha_{-} \;,&
z_{3}&:= -\frac{1}{4 \lambda\sub{SO}^{2}} \alpha_{+} \alpha_{-} \;,&
z_{4}&:= \frac{1}{4 \lambda\sub{SO}^{2}} \alpha_{-}^{2} \;,\\
z_{5}&:= \alpha_{-} \;, & 
v_{1}&:= - \frac{\lambda\sub{R}}{4 \lambda\sub{SO}^{2}}  \widetilde{\alpha}\alpha_{-} \;, &
v_{2}& :=  \frac{1}{2\lambda\sub{SO}}  \widetilde{\alpha}\alpha_{-} ,
\end{align*}
and
\begin{equation*}
\chi(m,\mathbf{q}) =\frac{\alpha_{-}}{4 \lambda\sub{SO}^{2}} \big(m^{2}\alpha_{-} -2\ii q_{0} m \lambda\sub{R}^{2}+  q_{0}^{2}\alpha_{+} +|q|^{2} \widetilde{\alpha}^{2}\big).
\end{equation*}
Then, we extract the most singular contribution from the connection associated with $\widehat{G}$
\begin{equation*}
A_{\nu}(\mu _{c};\mathbf{k}_{F}^{+} + \mathbf{q}) = \widetilde{A}_{\nu}(m,\mathbf{q})+ O_{m,\mathbf{q}}(1) \;, \qquad \widetilde{A}_{\nu}(m,\mathbf{q}):=\partial_{k_{\nu}} \big[ \widehat{G}(\mu _{c}; \mathbf{k}_{F}^{+})^{-1} \big] S(m,\mathbf{q}) \;
\end{equation*}
and similarly for $A^{\mathrm{s}}_{\nu}$. We have
\begin{equation*}
 \partial_{k_{\nu}}\big[ \widehat{G}(\mu _{c}; \mathbf{k}_{F}^{+})^{-1}\big] =  \partial_{k_{\nu}}H(k_{F}^{+}) + \ii  \delta_{\nu,0} \mathbb{1}_{\mathbb{C}^{4}}
\end{equation*}
By inspection, $H\sub{KM}(k)$ and $S(m,\mathbf{q})$ are rational functions in $m$; therefore, $\widehat{G}(\mu_c;\mathbf{k})$ and the terms $O_{m,\mathbf{q}}(1)$ are also rational functions in $m$. Since the latter terms are, in particular, bounded, they are also continuous in $m$. Since the singular part of $A_{\nu}$ and $A^{s}_{\nu}$ diverges as $(|\mathbf{q}| + |m|)^{-1}$, all the sub-leading contributions to $\Tr _{\C^{4}} 
\Big(A^{\mathrm{s}}_{i}(\mu_{c};\mathbf{k})A_{0}(\mu_{c};\mathbf{k})A_{j}(\mu_{c};\mathbf{k})\Big)$, namely those with at least one $O_{\mathbf{q},m}(1)$, are uniformly integrable over $\mathbb{T}^{2}_{+}$ as $m \to 0$. Therefore, by the dominated convergence theorem, these terms provide the same finite contribution to the limits $m \to 0^{\pm}$; in particular, they do not contribute to  $\delta \sigma_{12}^{\mathrm{s}}$.
Accordingly, we can write
\begin{equation}\label{diff-singular-contr}
\delta \sigma_{12}^{\mathrm{s}}= \Big( \lim _{m \to 0^{+}} - \lim _{m \to 0^{-}}\Big)\widetilde{\sigma}^{\mathrm{s}}_{12}(m) \;,
\end{equation}
where, denoting $\tilde{\mathbb{T}}_{+}^{2}:=\{q \in \R^{2} \, | \, k_{F}^{+}+q \in \mathbb{T}_{+}^{2}\}$,
\begin{equation}\label{singular-Pontryagin-formula}
\begin{split}
\widetilde{\sigma}^{\mathrm{s}}_{12}(m)
= -2 \int_{\R \times \tilde{\mathbb{T}}_{+}^{2}} \frac{\dd \mathbf{q}}{(2 \pi)^{3}} \Tr_{\C^{4}} 
\Big(\widetilde{A}^{\mathrm{s}}_{1}(m;\mathbf{q})\widetilde{A}_{0}(m;\mathbf{q})\widetilde{A}_{2}(m;\mathbf{q})\Big) \;.
\end{split}
\end{equation}
We now show that both limits in \eqref{diff-singular-contr} exist (and, as a matter of fact, also proving that the limits $\lim _{m \to 0^{\pm}} \sigma^{\mathrm{s}}_{12}(t,\lambda\sub{SO},m,\lambda\sub{R},r)$ exist) and compute their exact values. With the aid of a computer software \cite{Wolfram}, we compute the trace in \eqref{singular-Pontryagin-formula}
and obtain
\begin{equation*}
\Tr_{\C^{4}} 
\Big(\widetilde{A}^{\mathrm{s}}_{1}(m;\mathbf{q})\widetilde{A}_{0}(m;\mathbf{q})\widetilde{A}_{2}(m;\mathbf{q})\Big) = -\frac{ 3 m t \widetilde{\alpha}\alpha_{-}^{2}}{4 \lambda\sub{SO} \chi(m,\mathbf{q})^{2}}
- \frac{3 t  \widetilde{\alpha}^{3} \alpha_{-}^{3}}{4 \lambda_{\mathrm{SO}}^{5} \chi(m,\mathbf{q})^{3}}\big(\alpha_{+}q_{0} -\ii m \lambda_{\mathrm{R}}^{2}\big) q_{1}q_{2} \;.
\end{equation*}
The second term does not contribute to \eqref{singular-Pontryagin-formula}, because it is odd in $q_{1}$, whereas the integration domain $\tilde{\mathbb{T}}_{+}^{2}$ is even. Therefore,
\begin{equation*}
\begin{split}
\delta \sigma_{12}^{\mathrm{s}} 
& =  \Big( \lim _{m \to 0^{+}} - \lim _{m \to 0^{-}}\Big) 
\frac{3 m t \widetilde{\alpha} \alpha_{-}^{2}}{16 \lambda\sub{SO} \pi^{3}}
\int_{\R \times \tilde{\mathbb{T}}_{+}^{2}} \frac{\dd \mathbf{q}}{\chi(m,\mathbf{q})^{2}}
\\
& =  \Big( \lim _{m \to 0^{+}} - \lim _{m \to 0^{-}}\Big) 
\frac{3 m t \widetilde{\alpha}\lambda\sub{SO}^{3}}{4 \pi^{3}}
\int_{\R^{3} } \dd \mathbf{q}
\big( q_{0}^{2}\alpha_{+} -2\ii q_{0} m \lambda\sub{R}^{2} +m^{2}\alpha_{-}+|q|^{2} \widetilde{\alpha}^{2}\big)^{-2} \;.
\end{split}
\end{equation*}
where in the second line we performed some algebraic manipulations and extended the integral to the whole $\R^{2}$ by dominated convergence, using that the integrand is uniformly integrable on $\R^{2} \setminus \tilde{\mathbb{T}}_{+}^{2}$ and continuous in $m$.

We compute the integral in $\dd q_{0}$ using that, for any $a,b,c\in \R$ such that $4ac + b^{2} >0$, we have $\int_{\R} \dd z \big(a z^{2} + \ii b z + c\big)^{-2} =  -4 \pi  a  (4 a c +b^{2})^{-\frac{3}{2}}$.
Then, we note that for $a',b'>0$ we have $\int_{\R^{2}} \dd x (a' |x|^{2} +b')^{-\frac{3}{2}} = 2 \pi \big(a' b'^{1/2}\big)^{-1}$, hence we have 
\begin{equation*}
\begin{split}
& \int_{\mathbb{R}^{2}} \dd q \int_{\mathbb{R}} \dd q_{0}
\big( q_{0}^{2}\alpha_{+} -2\ii q_{0} m \lambda\sub{R}^{2} +m^{2}\alpha_{-}+|q|^{2} \widetilde{\alpha}^{2}\big)^{-2} 
\\
& \qquad  = - 4 \pi \alpha_{+} \int_{\mathbb{R}^{2}} \dd q \big(  4 |q|^{2} \alpha_{+} \widetilde{\alpha}^{2} +  (2 m \lambda\sub{SO}^{2})^{2} \big)^{-\frac{3}{2}}
 = - \frac{\pi^{2}}{\lambda\sub{SO}^{2} \widetilde{\alpha}^{2}|m|} \;,
\end{split}
\end{equation*}
implying the existence of the limits and the sought result \eqref{eq:spin-discontinuity}.
\end{proof}

\begin{remark}\label{rmk:cont}
Note that for $m\neq 0$ the integrand in \eqref{eqn:AAA-int} is continuous in $\lambda\sub{R}$, in a neighbourhood of $0$, and uniformly bounded by some integrable function (decaying in $k_0$ and constant in $k_1$ and $k_2$). Thus, by the dominated convergence theorem $\sigma^{s}_{12}$ is also continuous in $\lambda\sub{R} $, in a neighbourhood of $0$.
\end{remark}

\appendix 

\section{Proof of Theorem \ref{prop:imaginary-time-repr}} 
\label{app:imaginary}

To prove Theorem \ref{prop:imaginary-time-repr}, we switch to the grand canonical formalism in second quantisation, where the Kubo--Martin--Schwinger (KMS) property can be used to perform the Wick rotation; see, e.g., \cite{Porta17,MaPo}.
Unlike \cite{Porta17,MaPo} we work in infinite volume directly, which is possible because we either localise observables or use the trace per unit volume functional.

\subsection{Grand canonical formalism}

Let us introduce the antisymmetric (or fermionic) Fock space associated with $\mathcal{H}$, see, e.g., \cite{Bratteli}. For $n\in \mathbb{N}$, let $\Hi^{\otimes n}$ denote the Hilbert space $n$-fold tensor product of $\Hi$. Moreover, for $n \in \N$ let $P_{n}$ be the antisymmetriser on $\mathcal{H}^{\otimes n}$, that is, 
\begin{equation*}
P_{n} f_{1} \otimes \cdots \otimes f_{n}:= (n!)^{-1}\sum_{\pi \in S_{n}} \mathrm{sgn}(\pi) f_{\pi(1)} \otimes \cdots \otimes f_{\pi(n)} \;,
\end{equation*}
$S_{n}$ denoting the group of permutations of $n$ elements and $\mathrm{sgn(\cdot)}$ the sign of the permutation, and set $\mathcal{H}^{\wedge n} := P_{n} \mathcal{H}^{\otimes n}$. The antisymmetric (or fermionic) Fock space $\F(\Hi)$ is the Hilbert space direct sum
\begin{equation*}
\mathcal{F}(\mathcal{H}):= \C \Omega \oplus \bigoplus _{n\geq 1} \mathcal{H}^{\wedge n} \;,
\end{equation*}
where $\Omega$ denotes a distinguished vector called vacuum vector.  For any $f \in \mathcal{H}$ one introduces the creation operator $a^{*}(f)$ and annihilation operator $a(f)$ on $\mathcal{F}(\mathcal{H})$:
\begin{equation*}
\begin{split}
a^{*}(f)\Omega = f \;, \qquad \quad a^{*}(f) f_1 \wedge \cdots \wedge f_n & := f \wedge f_1 \wedge \cdots \wedge f_n \;,
\\
a(f) \Omega = 0 \;, \qquad \quad
a(f) f_1 \wedge \cdots \wedge f_n & := \sum_{j=1}^{n}(-)^{j-1} \langle f,f_{j} \rangle f_1 \wedge \cdots \wedge \slashed{f_{j}} \wedge \cdots \wedge f_n\;.
\end{split}
\end{equation*}
where we used the notation $f_{1} \wedge \cdots \wedge f_{n} := (n!)^{1/2}P_{n} f_{1} \otimes \cdots \otimes f_{n}$.
These operators are the adjoint of each other, satisfy the canonical anticommutation relations
\begin{equation*}
\{a(f), a(g) \} = \{a^{*}(f), a^{*}(g) \} = 0 \;,\qquad \{a(f), a^{*}(g) \} = \langle f,g \rangle_{\mathcal{H}}\;,\qquad \forall f,g \in \mathcal{H} \;,
\end{equation*}
and are bounded, namely
\begin{equation*}
\|a(f) \|_{\mathcal{F}(\mathcal{H}) \to \mathcal{F}(\mathcal{H})} \leq \| f\|_{\mathcal{H}} \;.
\end{equation*}
We denote by $\mathrm{CAR}(\mathcal{H})$ the unital $C^{*}$-algebra generated by $\{a(f) \, | \, f \in \mathcal{H} \}$, see \cite{Bratteli} for more details. A functional $\omega: \mathrm{CAR}(\mathcal{H}) \to \mathbb{C}$ is a state if it is positive, normalised and linear. Non-interacting systems are described in terms of quasi-free states on $\mathrm{CAR}(\mathcal{H})$. In particular, we are interested in gauge-invariant quasi-free states.
\begin{definition}[Gauge-invariant Quasi-free States]\label{def-quasi-free}
A state $\omega:\mathrm{CAR}(\mathcal{H})\mapsto \mathbb{C}$ is gauge-invariant and quasi-free iff for any $n,m \in \mathbb{N}$ and $f_{1},\dots,f_{n} \in \mathcal{H}$ and $g_{1},\dots,g_{m} \in \mathcal{H}$
\begin{equation*}
\omega(a^{*}(f_{1}) \cdots a^{*}(f_{n}) a(g_{m}) \cdots a(g_{1})) = \delta_{n,m} \det \big( \omega(a^{*}(f_{i}) a(g_{j})) \big)_{1\leq i,j \leq n} \;.
\end{equation*}
\end{definition}
\begin{remark}
Note that a gauge-invariant quasi-free state is characterized by the two-point function $\omega(a^{*}(f) a(g))$, whereas $\omega(a^{*}(f) a^{*}(g))$ and $\omega(a(f) a(g))$ are identically zero.
\end{remark}
A special type of gauge-invariant quasi-free state is the grand canonical state associated with a one-particle Hamiltonian \cite[Example 5.3.2]{Bratteli}.
\begin{definition}[Grand canonical state]\label{def:gran-state}
Let $H \in \Bs(\Hi)$ be self-adjoint and $\mu \in \R$. The (infinite-volume) grand canonical state $\omega^{H,\mu}_{\beta}$ on $\mathrm{CAR}(\mathcal{H})$ at inverse temperature $\beta >0$ is the unique gauge-invariant quasi-free state with two-point function
\begin{equation}\label{eq:Fermi-Dirac-2point}
\omega^{H,\mu}_{\beta}(a^{*}(f) a(g)) = \langle g, (\mathbb{1}+\ee^{\beta(H - \mu \mathbb{1})})^{-1} f\rangle_{\mathcal{H}} \;. 
\end{equation}
\end{definition}
\begin{remark}
Note that we do not use the Gibbs prescription to define the grand canonical state because the operator $\exp(-\beta (H - \mu \mathbb{1}))$ is not trace class, as we are working in infinite volume.
\end{remark}
The operator appearing on the r.h.s.~of \eqref{eq:Fermi-Dirac-2point}, which we henceforth denote by
\begin{equation}\label{eq: Fermi-Dirac-dist}
\gamma_{\beta}^{H,\mu} := (\mathbb{1}+\ee^{\beta(H - \mu \mathbb{1})})^{-1} \;,
\end{equation}
is the one-body density matrix associated with the grand canonical state; in particular, it represents the Fermi--Dirac distribution associated with $H$ at inverse temperature $\beta$.

The grand canonical state satisfies the Kubo--Martin--Schwinger (KMS) property  with respect to the one-parameter group of Bogoliubov $^{*}$-automorphisms $(\alpha_{t})_{t \in \R}$ such that
\begin{equation}\label{eq:modular-automorphism}
\alpha_{t}\big( a(f) \big) := a(\ee^{\ii t(H-\mu \mathbb{1})}f) \;, \qquad \forall f \in \Hi \;,
\end{equation}
see \cite[Section 5.2.4]{Bratteli}. More precisely, one introduces the set of states that satisfy the KMS property with respect to $(\alpha_{t})_{t \in \mathbb{R}}$, see \cite[Definition 5.3.1]{Bratteli}.
\begin{definition}[KMS States]\label{def: KMS property}
Let $(\alpha_{t})_{t \in \R}$ be the one-parameter group of Bogoliubov $^{*}$- automorphisms as in \eqref{eq:modular-automorphism} and let $\beta >0$. A state on $\mathrm{CAR}(\mathcal{H})$ is a $(\beta,\alpha)$-KMS state if
\begin{equation}\label{eq:KMS-property}
\omega(A\alpha_{\ii \beta}(B)) = \omega(B A)
\end{equation}
for all $A,B$ in the polynomial $^{*}$-algebra generated by $\{a(f) \, | \, f \in \Hi \}$.
\end{definition}

\begin{remark}
Note that in \cite[Definition 5.3.1]{Bratteli} the condition \eqref{eq:KMS-property} is required for $A,B$ belonging to a norm-dense, $(\alpha_{t})_{t\in \mathbb{R}}$-invariant $^{*}$-subalgebra of the set of entire analytic elements for $\alpha$. This set comprises those operators $ A\in \mathrm{CAR}(\mathcal{H})$, such that the map $t \mapsto \alpha_{t}(A)$ has an extension to an entire function (in the strong or equivalently weak sense). In our setting,  since $H$ is bounded, the polynomial $^{*}$-algebra generated by $\{a(f) \, | \, f \in \Hi \}$ is precisely a norm-dense, $(\alpha_{t})_{t\in \mathbb{R}}$-invariant $^{*}$-subalgebra of entire analytic elements for $(\alpha_{t})_{t\in \mathbb{R}}$.
\end{remark}
Then, we have the following, see \cite[Example 5.3.2]{Bratteli}:
\begin{lemma}\label{prop-KMS-state}
The grand canonical state $\omega^{H,\mu}_{\beta}$ is the unique $(\beta,\alpha)$-KMS state.
\end{lemma}
\medskip

We shall now establish the connection between the grand canonical state $\omega^{H,\mu}_{\beta}$ and the state $\tau( \,\cdot \, P)$, where $\tau$ is the trace per unit volume and $P$ the Fermi projector. To this end, we shall first introduce the second quantisation: for any bounded operator $A \in \Bs(\Hi)$, we let $\dd\Gamma(A)$ denote its second quantization, that is, the (closure of the) operator such that $\dd \Gamma(A) \Omega = \Omega$ and for any $n\in \mathbb{N}$ 
\begin{equation*}
\dd \Gamma(A) \mathcal{H}^{\wedge n} = \sum_{i=1}^{n} \underbrace{\mathbb{1}_{\mathcal{H}} \otimes \cdots  \otimes A \otimes \cdots \otimes  \mathbb{1}_{\mathcal{H}}}_{\text{i-th position}} \mathcal{H}^{\wedge n} \;.
\end{equation*}
If $(f_j)_{j\in\N}$ is an orthonormal basis in $\mathcal{H}$ then
\begin{equation}\label{eq:op-second-quantization}
\dd \Gamma(A) = \sum_{i,j\in\N} \langle f_i, A f_j  \rangle_{\mathcal{H}} a^{*}(f_{i}) a(f_{j}) \;.
\end{equation}
Note that for $A,B \in \Bs(\Hi)$ one has
\begin{equation} \label{second-quantization-commutator}
\dd \Gamma\big([A,B] \big) = \big[ \dd \Gamma(A), \dd \Gamma(B) \big] \;.
\end{equation}
We also introduce the following notation: for any operator $A$ on $\Hi$, we denote by
\begin{equation*}
A_{L} := \chi_{L} A \chi_{L}
\end{equation*}
its restriction to the cell $C_{L}$.

The following holds true.
\begin{lemma}\label{lemma:grand-canonical}
Let $(H,\mu)$ be an insulator, see Section \ref{sec:linear-response}, with $H \in \mathcal{P}_{0}(\Hi)$ and let $P = P(H,\mu)$ be the associated Fermi projector. Then, for any $A \in \mathcal{P}(\Hi)$
\begin{equation}\label{eq:grand-canonical}
\begin{split}
\tau ( A P )  & =
\lim _{L \to \infty} \lim _{\beta \to \infty} |C_{L}|^{-1}\omega_{\beta}^{H,\mu} \big(\dd \Gamma(A_{L}) \big)  \;.
\end{split}
\end{equation}
Assume furthermore that $A \in \mathcal{P}(\Hi)$, let $B \in \mathcal{P}_{0}(\Hi)$ and let $X$ be the position operator. Then,
\begin{equation}\label{eq:tau-omega-comm}
\begin{split}
\tau ( [A,B] P ) &  =
\lim _{L \to \infty} \lim _{\beta \to \infty} |C_{L}|^{-1} \omega_{\beta}^{H,\mu} \big([\dd \Gamma(A_{L}),\dd \Gamma(B)] \big) \;,
\\
\tau([B,X]P) & = \lim _{L \to \infty} \lim_{\beta \to \infty} |C_{L}|^{-1} \omega^{H,\mu}_{\beta}([\dd \Gamma(B_{L}),\dd \Gamma(X)]) \;.
\end{split}
\end{equation}
\end{lemma}
\begin{remark}
Note that the grand canonical state appearing above is in infinite volume and thus extensive observables have to be localised over the cell $C_{L}$. This is one of the main difference with respect to, e.g., \cite{Porta17,Porta20,MaPo}, where the state and the observable are defined on the torus and the infinite-volume limit is performed for both at the same time.
\end{remark}
\begin{figure}
\begin{tikzpicture}
  \def\outer{4}
  \def\inner{3}

  \fill[pattern=dots, pattern color=gray!40] 
    (-\outer/2, -\outer/2) rectangle (\outer/2, \outer/2);
  \fill[white]
    (-\inner/2, -\inner/2) rectangle (\inner/2, \inner/2);

  \draw[thick] (-\outer/2, -\outer/2) rectangle (\outer/2, \outer/2);
  \draw[thick] (-\inner/2, -\inner/2) rectangle (\inner/2, \inner/2);

    \draw[->, thick, bend right=25] 
    (2.8, 0.5) to (1.8, 1.3);

  \node at (3, 0.1) {$\widetilde{C}_{L}$};
  \node at (0, 0) {$C_L \setminus \widetilde{C}_{L}$};

  \node[left=6pt] at (-2, 2) {$C_L$};

\end{tikzpicture}
\caption{Pictorial representation of the annulus $\widetilde{C}_{L}$ inside of $C_{L}$.
\label{fig:annulus}}
\end{figure}
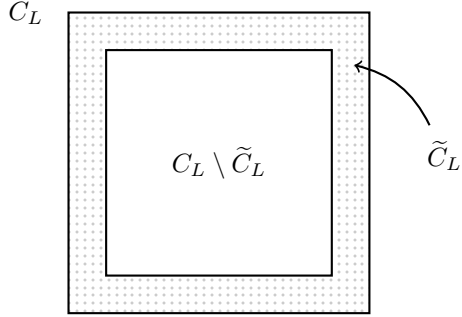
To prove this lemma, we need a couple of technical results. First of all, we note that, although the trace per unit volume localizes the whole observable, when considering products one has the following result.
\begin{lemma}\label{lemma-move-local}
Let $A \in \mathcal{P}(\Hi)$ and $B \in \mathcal{P}_{0}(\Hi)$, see \eqref{eq:gamma-per-operators}. Then, 
\begin{equation*}
\tau (A B) = \lim _{L \to \infty} |C_{L}|^{-1} \Tr_{\mathcal{H}} \big(  A_{L}  B \big) 
\end{equation*}
\end{lemma}
\begin{proof}
Letting $C^{c}_{L}:=\Cr \setminus C_{L}$ and $\chi_{L}^{c} :=\Id_{ C_{L}^{c}}$, we have the identity
\begin{equation*}
\begin{split}
 \Tr_{\mathcal{H}} \big( \chi_{L} A  B \chi_{L} \big) = \Tr_{\mathcal{H}} \big( \chi_{L} A \chi_{L} B \chi_{L} \big)  + \Tr_{\mathcal{H}} \big( \chi_{L} A \chi^{c}_{L} B \chi_{L} \big) \;.
\end{split}
\end{equation*}
We need to prove that the second piece gives a vanishing contribution when divided by $|C_{L}|$ as $L \to \infty$. To this end, we let $\widetilde{C}_{L}:= \{ x \in C_{L}  \, | \, \mathrm{dist}(x, C_{L}^{c}) \leq \log L \}$, where $\mathrm{dist}$ is in the sense of $|\, \cdot \, |_1$, be the internal annulus of thickness $L \log L$. Then,
\begin{equation*}
\begin{split}
\big|\Tr_{\mathcal{H}} \big( \chi_{L} A \chi^{c}_{L} B \chi_{L} \big)  \big| 
& \leq \| A\| \Big(  \sum_{y \in C_{L}^{c}, x \in \widetilde{C}_{L}} |B_{y,x} | + \sum_{y \in C_{L}^{c}, x \in C_{L} \setminus \widetilde{C}_{L}} |B_{y,x} |\Big) \;.
\end{split}
\end{equation*}
In the first term, we are summing over $x$ belonging to a region of area $L \log L$, thus
\begin{equation*}
\sum_{y \in C_{L}^{c}, x \in \widetilde{C}_{L}} |B_{y,x} | \lesssim (L \log L) \sup_{x \in \mathbb{Z}^{2}} \sum_{y \in \mathbb{Z}^{2}}|B_{y,x}| \lesssim L \log L \;,
\end{equation*}
where in the last step we used that $B$ is short-range, see \eqref{eq: exp-decay}.
In the other region, we note that $B_{y,x}$ is small because the points $y$ and $x$ are at distance at least $\log L$, thus for $\xi '>\xi$, with $\xi$ being the constant in \eqref{eq: exp-decay} for the operator $B$, we have 
\begin{equation*}
\begin{split}
\sum_{y \in C_{L}^{c}, x \in C_{L} \setminus \widetilde{C}_{L}} |B_{y,x} | 
&
\leq \left( \sup _{y \in C_{L}^{c}, x \in C_{L} \setminus \widetilde{C}_{L}} \ee^{- |x-y|_{1}/\xi'} \right)
\sum_{y \in C_{L}^{c}, x \in C_{L} \setminus \widetilde{C}_{L}} \ee^{|x-y|_{1}/\xi'} |B_{y,x} | 
\\
& \leq \left( \sup _{y \in C_{L}^{c}, x \in C_{L} \setminus \widetilde{C}_{L}} \ee^{- |x-y|_{1}/\xi'} \right)
L^{2} \sup_{x\in\mathbb{Z}^{2}}\sum_{y \in\mathbb{Z}^{2}} \ee^{|x-y|_{1}/\xi'} |B_{y,x} | 
\\
& \lesssim L^{2- 1/\xi'} \;,
\end{split}
\end{equation*}
where we used that $\sup _{y \in C_{L}^{c}, x \in C_{L} \setminus \widetilde{C}_{L}} \ee^{- |x-y|_{1}/\xi'} \leq L^{-1/\xi'}$. Since $|C_{L}| \sim L^{2}$, this concludes the proof.
\end{proof}
A similar result involving commutators holds true.
\begin{corollary}
Let $A \in \mathcal{P}(\Hi)$ and $B,C \in \mathcal{P}_{0}(\Hi)$. Then,
\begin{equation}\label{eq:commutator-trace}
\tau ([A,B]C) = \lim _{L \to \infty} |C_{L}|^{-1} \Tr_{\mathcal{H}} ([A_{L},B] C) \;.
\end{equation}
Furthermore, if $X$ is the position operator, we have
\begin{equation}\label{eq:commutator-X-trace}
\tau ([B,X]C) = \lim _{L \to \infty} |C_{L}|^{-1} \Tr_{\mathcal{H}} ([B_{L},X] C) \;.
\end{equation}
\end{corollary}
\begin{proof}
To prove \eqref{eq:commutator-trace}, we can open the commutator $[A,B]$ and consider the terms $ABC$ and $BAC$ separately. Since $BC \in \mathcal{P}_{0}(\Hi)$, by Lemma \ref{lemma-move-local} we have
\begin{equation*}
\tau (ABC) = \lim_{L\to \infty} |C_{L}|^{-1} \Tr_{\mathcal{H}}(A_{L} BC) \;.
\end{equation*}
For the other term, we obtain instead
\begin{equation*}
\tau (BAC) = \lim_{L\to \infty} |C_{L}|^{-1} \Tr_{\mathcal{H}}(\chi_{L}BA \chi_{L} C) \;.
\end{equation*}
We write
\begin{equation*}
\Tr_{\mathcal{H}}(\chi_{L}BA \chi_{L} C) = \Tr_{\mathcal{H}}(B \chi_{L}A \chi_{L} C) +\Tr_{\mathcal{H}}(\chi_{L}B \chi^{c}_{L}A \chi_{L} C) -\Tr_{\mathcal{H}}(\chi_{L}^{c}B \chi_{L}A \chi_{L} C) \;.
\end{equation*}
Following the proof of Lemma \ref{lemma-move-local}, one can see that the second and third terms are $O(L^{2-1/\xi})$, $\xi$ being the constant associated with $B$ according to \eqref{eq: exp-decay}, so that the first claim is proven. Concerning \eqref{eq:commutator-X-trace}, since by Lemma \ref{lemma:shortrange} $[B,X] \in \mathcal{P}(\Hi)$ we can apply Lemma \ref{lemma-move-local} and use that $\chi_{L}$ commutes with $X$.
\end{proof}
\begin{proof}[Proof of Lemma \ref{lemma:grand-canonical}]
By \eqref{eq:op-second-quantization} we have
\begin{equation}\label{eq:one-body-expectation}
\begin{split}
\omega_{\beta}^{H,\mu} \big( \dd \Gamma(A_{L}) \big)  & =\omega_{\beta}^{H,\mu} \left( 
\sum_{i,j} \langle f_{i}, A_{L} f_{j} \rangle a^{*}(f_{i}) a(f_{j}) \right)
\\
& = 
\sum_{i,j} \langle f_{i}, A_{L} f_{j} \rangle  \omega_{\beta}^{H,\mu} \left( a^{*}(f_{i}) a(f_{j}) \right)
\\
& = \Tr_{\mathcal{H}} \big( A_{L} \gamma_{\beta}^{H,\mu}  \big)
\end{split}
\end{equation}
where $(f_{i})_{i}$ is any basis of localised functions, so that the sums above are all finite, and where $\gamma_{\beta}^{H,\mu} $ was introduced in \eqref{eq: Fermi-Dirac-dist}. Since $(H,\mu)$ is an insulator, we have $\gamma^{H,\mu}_{\beta} \to P$ in the operator topology as $\beta \to \infty$, and that $P \in \mathcal{P}_{0}(\Hi)$ (see Lemma \ref{lem:Pi0shortrange}). Therefore, since $A_{L}$ is trace-class we have
\begin{equation}\label{eq:lim-beta}
\lim _{\beta \to \infty} \omega_{\beta}^{H,\mu} \big( \dd \Gamma(A_{L}) \big)  = \Tr_{\mathcal{H}} \big(  A_{L}  P\big)
\end{equation}
and we obtain \eqref{eq:grand-canonical} by Lemma \ref{lemma-move-local}. 

Regarding the commutator, one uses \eqref{second-quantization-commutator} and note that since $B \in \mathcal{P}_{0}(\Hi)$ and since $a(f_{i})$ are uniformly bounded in $i$
\begin{equation*}
[\dd \Gamma(A_{L}), \dd \Gamma (B)] = \dd \Gamma([A_{L}, B]) = \lim _{R \to \infty} \dd \Gamma([A_{L}, B_{R}])
\end{equation*}
where the limit is in the operator topology. Therefore, proceeding as in \eqref{eq:one-body-expectation}, we get
\begin{equation*}
\begin{split}
\omega^{H,\mu}_{\beta} ([\dd \Gamma(A_{L}), \dd \Gamma (B)]) & = \lim_{R \to \infty} \omega^{H,\mu}_{\beta} (\dd \Gamma([A_{L}, B_{R}])
\\
& = \lim _{R \to \infty}\Tr_{\mathcal{H}} \big([A_{L},B_{R}] \gamma^{H,\mu}_{\beta} \big) .
\end{split}
\end{equation*}
To take the limit in the last line, we note that both $A_{L}B_{R}$ and $B_{R}A_{L}$ converge in the trace-class topology to $A_{L}B$ and $BA_{L}$ respectively. To see this, note that for $R$ large enough $(B- B_{R})A_{L} = \chi_{R}^{c}B \chi_{L} A_{L}$, so that by H\"older's inequality for Schatten norms $\|(B-B_{R})A_{L} \|_{\mathrm{tr}} \leq \|\chi_{R}^{c}B \chi_{L} \|  \|A_{L} \|_{\mathrm{tr}}$ with $\|\chi_{R}^{c}B \chi_{L} \| \to 0$ as $R \to \infty$ since $B \in \mathcal{P}_{0}(\Hi)$, where $\|\, \cdot \, \|_{\mathrm{tr}}$ denotes the trace norm. Therefore,
\begin{equation*}
\lim _{R \to \infty}\Tr_{\mathcal{H}} \big([A_{L},B_{R}] \gamma^{H,\mu}_{\beta} \big)  = \Tr_{\mathcal{H}} \big([A_{L},B] \gamma^{H,\mu}_{\beta} \big) 
\end{equation*}
since $\gamma^{H,\mu}_{\beta}$ is bounded. Then, we can take the $\beta \to \infty$ as in \eqref{eq:lim-beta} since $[A_{L},B]$ is trace class, and conclude by \eqref{eq:commutator-trace}.

Finally, to prove the last claim we use that $\chi_{L}$ is a projection and that it commutes with $X$, so that $[B_{L}, X] = [B_{L}, X_{L}]$. But $X_{L}$ is bounded hence by \eqref{second-quantization-commutator}
\begin{equation*}
\dd \Gamma([B_{L},X]) = [\dd \Gamma(B_{L}), \dd \Gamma(X_{L})] = [\dd \Gamma(B_{L}), \dd \Gamma(X)] \;,
\end{equation*}
where we also used that $X = X_{L} + \chi^{c}_{L} X$ and that $[\dd \Gamma(B_{L}), \dd \Gamma (\chi^{c}_{L} X)]= \dd\Gamma\big( [B_{L},\chi^{c}_{L} X] \big) = 0$. The conclusion follows by arguing as before, and by using \eqref {eq:commutator-X-trace}.
\end{proof}
\begin{remark}
Note that because we work directly in infinite volume, the Fermi--Dirac operator $\gamma^{H,\mu}_{\beta}$ is not trace class, but this difficulty is overcome by the fact that we localise observables, so that the latter are trace class.
\end{remark}

\subsection{Wick's rotation of the Kubo formula}

To begin, we provide a Kubo formula for the spin conductivity in $ \mathcal{P}_{0}^{\hexagon}(\Hi)$. We henceforth denote $J:= \ii [H,X]$ the charge current operator associated with the Hamiltonian $H$.
\begin{proposition}[Kubo formula] \label{kubo-formula-spin}
Let $(H,\mu)$ be an insulator with $H \in \mathcal{P}_{0}^{\hexagon}(\Hi)$ and let $P=P(H,\mu)$ be the associated Fermi projector. Let $J_{i}^{\mathrm{conv}}$ be the conventional spin current and let $J(s):= \ee^{\ii s H}J \ee^{-\ii s H}$ be the time-evolved charge current. Then,
\begin{equation}\label{Kubo-formula}
\sigma^{\mathrm{s}}_{ij}=\lim _{\eta \to 0^{+}} \frac{\ii}{\eta} \bigg( \int_{-\infty}^{0}\ee^{\eta s} \tau \big( [J_{i}^{\mathrm{conv}}, J_{j}(s)]  P \big)\dd s - \tau\big([J_{i}^{\mathrm{conv}},X_{j} ]P \big) \bigg)\;.
\end{equation}
\end{proposition}
\begin{proof}
We start from Definition \ref{def-spin-cond} with $\sharp = \mathrm{conv}$. 
We write $X(t):= \ee^{\ii t H} X \ee^{-\ii t H}$. We note that $J_{i}^{\mathrm{conv}} [X_{j}(s),P] =J_{i}^{\mathrm{conv}} \ee^{\ii s H} [X_{j},P] \ee^{-\ii s H} \in \mathcal{P}(\Hi)$ by Lemma \ref{lem:Pi0shortrange} and \ref{lemma:shortrange}, uniformly in $s$, thus by Lemma \ref{lem:tau}\ref{it:tau4}
\begin{equation*}
\begin{split}
\sigma^{\mathrm{s}}_{ij}  =\re\lim_{\eta\to 0^+}\tau(  J_{i}^{\mathrm{conv}} \, L_{\eta,j})
 = \ii \lim _{\eta \to 0^{+}} \int_{-\infty}^{0}  \ee^{\eta s} \tau (J_{i}^{\mathrm{conv}}[X_{j}(s),P]) \dd s \;.
\end{split}
\end{equation*}
By the Leibniz rule, we have
\begin{equation*}
\begin{split}
\tau \big(J_{i}^{\mathrm{conv}} [X_{j}(s),P ]  \big) & + \tau \big( [J_{i}^{\mathrm{conv}},X_{j}(s)]P  \big)  
= \tau([X_{j}(s),J_{i}^{\mathrm{conv}} P]) \;.
\end{split}
\end{equation*}
We now prove that the r.h.s.~of the latter equation is identically null. Since $J(s)= \ii [H,X(s)] = \frac{\dd}{\dd t} X(s)$, we can write
\begin{equation}\label{eq:integration-X-null-term}
\tau([X_{j}(s),J_{i}^{\mathrm{conv}} P]) = \tau \big( [X_{j},J_{i}^{\mathrm{conv}} P] \big)  -  \tau \Big( \int_{s}^{0} [ J_{j}(s_{1}), J_{i}^{\mathrm{conv}} P ]\dd s_{1} \Big) \;.
\end{equation}
Because $\chi_{L}$ commutes with $X$ and $\chi_{L} X_{j} \chi_{L}$ is bounded we have
\begin{equation*}
 \Tr ( \chi_{L} [X,J_{i}^{\mathrm{conv}}P]\chi_{L}) = \Tr ( \chi_{L}X\chi_{L}J_{i}^{\mathrm{conv}}P\chi_{L} - \chi_{L}J_{i}^{\mathrm{conv}}P\chi_{L} X \chi_{L}) = 0 \;.
\end{equation*}
Concerning the second term in \eqref{eq:integration-X-null-term}, note that $J_{j}(s_{1}), J_{i}^{\mathrm{conv}}, P \in \mathcal{P}(\Hi)$ and that $J_{j}(s_{1})$ is uniformly bounded in $s_{1}$. Thus, we exchange $\tau(\cdot)$ with integration by Lemma \ref{lem:tau}\ref{it:tau4} and since by Lemma \ref{lem:tau}\ref{it:tau5} $\tau([ J_{j}(s_{1}), J_{i}^{\mathrm{conv}} P ]) = 0 $ we conclude that \eqref{eq:integration-X-null-term} is in fact identically null.
Therefore,
\begin{equation*}
\sigma^{\mathrm{s}}_{ij} = -\ii \lim _{\eta \to 0^{+}} \int_{-\infty}^{0}  \ee^{\eta s} \tau \big( [J_{i}^{\mathrm{conv}},X_{j}(s)]P \big) \dd s \;.
\end{equation*}
To obtain the two terms in \eqref{Kubo-formula}, we write $X_{j}(s)$ in terms of the charge current once more. The term with $X_{j}(0)$ can be integrated in $\dd s$ directly whereas we write $\int_{-\infty}^{0} \dd s \int_{s}^{0} \dd s_{1} = \int_{-\infty}^{0} \dd s_{1} \int_{-\infty}^{s_{1}} \dd s$ to integrate the other term.
\end{proof}

By exploiting the KMS property \eqref{eq:KMS-property}, we shall now perform a Wick rotation of formula \ref{Kubo-formula} in the grand canonical formalism. This will be presented in Proposition \ref{prop:imaginary-kubo-appendix}.
Let us begin with introducing the conventional spin current and the charge current in the second quantisation formalism. We write, for $t \in \mathbb{R}$,
\begin{equation*}
\begin{split}
\mathcal{J}_{i}(t)  :=\alpha_{t}\Big( \dd \Gamma(J_{i}) \Big) \;,
\qquad
\mathcal{J}^{\mathrm{spin}}_{L;i}  :=\dd \Gamma \big((J_{i}^{\mathrm{conv}})_{L} \big) \;;
\end{split}
\end{equation*}
note that we prefer to avoid localizing the dynamics and, as it turns out, it suffices to localize the spin current. Then, we introduce the spin-current-charge-current correlation function of the grand canonical state $K^{\beta,L}_{ij}:\mathbb{R} \to \mathbb{C}$,
\begin{equation*}
K^{\beta,L}_{ij}(t):=\omega^{H,\mu}_{\beta} \big(\mathcal{J}^{\mathrm{spin}}_{L;i} \mathcal{J}_{j}(t)  \big)\;,
\end{equation*}
where for brevity we dropped its dependence on $(H,\mu)$.
\begin{lemma}\label{lemma-KMS-J}
$K^{\beta, L}_{ij}$ has analytic extension to $\mathbb{C}$ and
\begin{equation}\label{eq:KMS}
\omega^{H,\mu}_{\beta}\big( [\mathcal{J}^{\mathrm{spin}}_{L;i}, \mathcal{J}_{j}(t)] \big) = 
K^{\beta,L}_{ij}(t) - K^{\beta,L}_{ij}(t+\ii\beta) \;,  \qquad \forall t \in \mathbb{R} \;.
\end{equation}
\end{lemma}
\begin{proof}
Since $\omega^{H,\mu}_{\beta}$ is quasi-free $K^{\beta,L}_{ij}$ can be expressed as a sum of products of two-point functions by the Wick rule, compare with Definition \ref{def-quasi-free}. Because $H$ is bounded so is $\exp (\ii t H)$ for any $t \in \mathbb{C}$ and thus the two-point function in \eqref{eq:Fermi-Dirac-2point} has analytic extension. Thus, also $K^{\beta,L}_{ij}$ has an analytic extension and \eqref{eq:KMS} follows by the KMS property \eqref{eq:KMS-property}.
\end{proof}
By performing the Wick rotation, the Kubo formula can be re-written in terms of the imaginary-time-ordered spin-current-charge-current correlation function. 
More precisely, we let
\begin{equation*}
\mathcal{K}^{\beta,L}_{ij}(t) := \omega^{H,\mu}_{\beta} \big( \mathbf{T} \mathcal{J}^{\mathrm{spin}}_{L;i} \mathcal{J}_{j}(\ii t)  \big)\;, \qquad t \in (0,\beta] \;,
\end{equation*}
where $\mathbf{T}$ denotes the fermionic imaginary-time-ordering, acting on products of the generators of the $\mathrm{CAR}(\mathcal{H})$ evolved by $\alpha_{\ii t}$ as follows
\begin{equation*}
\mathbf{T} \alpha_{\ii t_{1}}\big( a^{\sharp_{1}}(f_{1}) \big) \cdots \alpha_{\ii t_{n}}\big( a^{\sharp_{n}}(f_{n}) \big) = \mathrm{sign}(\pi) \alpha_{\ii t_{\pi(1)}}\big( a^{\sharp_{\pi(1)}}(f_{\pi(1)}) \big) \cdots \alpha_{\ii t_{\pi(n)}}\big( a^{\sharp_{\pi(n)}}(f_{\pi(n)}) \big) \;,
\end{equation*}
where $a^{\sharp}$ is either $a$ or $a^{*}$ and where $\pi$ is the permutation such that $t_{\pi(1)} \leq \dots \leq t_{\pi(n)}$ and such that the operators are ordered normally for equal $t$'s; $\mathbf{T}$ is extended by linearity to polynomials in the $a$'s and $a^{*}$'s. 

For any $\beta >0$ and $\eta \geq 0$ we let
\begin{equation*}
\eta_{\beta} := \min\; \Big\{\frac{2\pi}{\beta} \mathbb{Z} \cap [\eta,\infty) \Big\} \;,
\end{equation*}
and set
\begin{equation*}
\widehat{\mathcal{K}}_{ij}(\eta) :=   \lim _{L \to \infty} \lim _{\beta \to \infty} |C_{L}|^{-1} \int_{0}^{\beta} \ee^{\ii \eta_{\beta} t} \mathcal{K}^{\beta,L}_{ij}(t) \dd t \;.
\end{equation*}
The Wick rotation results in the following representation.
\begin{proposition}[Wick's rotation]\label{prop:imaginary-kubo-appendix}
Under the same assumptions of Proposition \ref{kubo-formula-spin}, we have
\begin{equation*}
\sigma_{ij}^{\mathrm{s}} = - \lim _{\eta \to 0^{+}}  \frac{\widehat{\mathcal{K}}_{ij}(\eta) - \widehat{\mathcal{K}}_{ij}(0)}{\eta} \;.
\end{equation*}
\end{proposition}
\begin{proof}
Consider the first term in the Kubo formula \eqref{Kubo-formula}. By Lemma \ref{lem:Pi0shortrange} and Lemma \ref{lemma:shortrange}, we note that $J_{i}^{\mathrm{conv}}$, $P$ and $J_{j}(t)$ are in $\mathcal{P}_{0}(\Hi)$ for any $t\in \mathbb{R}$. It is then easy to check that $[\mathcal{J}^{\mathrm{spin}}_{L;i} ,\mathcal{J}_{j}(t)]$ is uniformly (w.r.t.\ $t\in\R$) bounded as an operator on $\mathcal{F}(\mathcal{H})$, as long as $L$ is finite. Thus, by using the first identity in \eqref{eq:tau-omega-comm} we have
\begin{equation}\label{eq:Kubo-1-I}
\begin{split}
 \int_{-\infty}^{0}\ee^{\eta t}\tau \big( [J_{i}^{\mathrm{conv}} ,J_{j}(t)] P \big) \dd t  & =  \int_{-\infty}^{0}\ee^{\eta t}\lim_{L \to \infty} \lim _{\beta \to \infty} |C_{L}|^{-1} \omega^{H,\mu}_{\beta}\big( [\mathcal{J}^{\mathrm{spin}}_{L;i} ,\mathcal{J}_{j}(t)] \big) \dd t
 \\
 & =  \lim_{L \to \infty} \lim _{\beta \to \infty} |C_{L}|^{-1}\int_{-\infty}^{0}\ee^{\eta t} \omega^{H,\mu}_{\beta}\big( [\mathcal{J}^{\mathrm{spin}}_{L;i} ,\mathcal{J}_{j}(t)] \big) \dd t
 \;,
\end{split}
\end{equation}
where we swapped the limit with the integral by the dominated convergence theorem.
Using once more that $[\mathcal{J}^{\mathrm{spin}}_{L;i} ,\mathcal{J}_{j}(t)]$ is uniformly bounded in $t$,  we have 
\begin{equation}\label{eq:Kubo-1-II}
\begin{split}
\Big| \int_{-\infty}^{0} &\big(\ee^{\eta t} - \ee^{\eta_{\beta} t} \big) \omega^{H,\mu}_{\beta}\big( [\mathcal{J}^{\mathrm{spin}}_{L;i} ,\mathcal{J}_{j}(t)] \big) \dd t \Big| 
\\
&
\lesssim \int_{-\infty}^{0} \big(\ee^{\eta t} -\ee^{\eta_{\beta} t}\big)  \dd t
=\eta^{-1} - \eta_{\beta}^{-1}  \leq \frac{2\pi}{\eta^{2} \beta} \;,
\end{split}
\end{equation}
which is vanishing as $\beta \to \infty$, having fixed $\eta >0$. On the other hand, by Lemma \ref{lemma-KMS-J},  noting that $\ee^{\eta_{\beta} t} = \ee^{\eta_{\beta}( t + \ii \beta)}$ we write
\begin{equation}\label{eq:contour-deformation}
\begin{split}
\int_{-\infty}^{0}\ee^{\eta_{\beta} t}& \omega^{H,\mu}_{\beta}\Big( [\mathcal{J}^{\mathrm{spin}}_{L;i} ,\mathcal{J}_{j}(t)] \Big) \dd t
\\
& =  \int_{-\infty}^{0}\ee^{\eta_{\beta} t} 
\Big( K^{\beta,L}_{ij}(t) - K^{\beta,L}_{ij}(t+\ii\beta)\Big) \dd t
\\
& = \lim _{T \to \infty} \int_{-T}^{0} \Big(\ee^{\eta_{\beta} t} K^{\beta,L}_{ij}(t) 
-\ee^{\eta_{\beta}(t+\ii \beta)} K^{\beta,L}_{ij}(t+\ii\beta)
\Big)\dd t
\\
& = \ii \int_{0}^{\beta} \ee^{ \ii \eta _{\beta} t} K_{ij}^{\beta,L}(\ii t) \dd t - \ii \lim _{T \to \infty} \ee^{-\eta_{\beta} T} \int_{0}^{\beta} \ee^{\ii \eta _{\beta} t} K_{ij}^{\beta,L}(-T+\ii t) \dd t
\end{split}
\end{equation}
where we used the Cauchy integral formula for $K^{\beta,L}_{ij}$ being analytic. Because $[\mathcal{J}^{\mathrm{spin}}_{L;i} ,\mathcal{J}_{j}(-T+\ii t)]$ is uniformly bounded in $T \in \mathbb{R}$ and locally bounded in $t$, so is $K_{ij}^{\beta,L}$. Therefore, since $\eta_{\beta}>0$ by definition, the last term in \eqref{eq:contour-deformation} is identically zero. We also note that $\mathcal{J}^{\mathrm{spin}}_{L;i}\mathcal{J}_{j}(\ii t) = \mathbf{T}\mathcal{J}^{\mathrm{spin}}_{L;i}\mathcal{J}_{j}(\ii t)$ for $t>0$ and thus
\begin{equation}\label{eq:T-ordering}
K^{\beta,L}_{ij}(\ii t) = \mathcal{K}^{\beta,L}_{ij}(t) \;, \qquad t \in (0,\beta] \;.
\end{equation}
Putting together \eqref{eq:Kubo-1-I}, \eqref{eq:Kubo-1-II}, \eqref{eq:contour-deformation} and \eqref{eq:T-ordering} we obtain
\begin{equation}\label{eq:first-term-final}
\int_{-\infty}^{0}\ee^{\eta t}\tau \big( [J_{i}^{\mathrm{conv}} ,J_{j}(t)]P \big) \dd t  = \ii \lim_{L \to \infty} \lim _{\beta \to \infty} \frac{1}{|C_{L}|}  \int_{0}^{\beta} \ee^{ \ii \eta _{\beta} t} \mathcal{K}_{ij}^{\beta,L}(t)\dd t \;.
\end{equation}
Let us now analyse the second term in the Kubo formula \eqref{Kubo-formula}. Let $\mathcal{X}_{j}:= \dd \Gamma(X_{j})$. 
 By employing the second identity in \eqref{eq:tau-omega-comm}, we have
\begin{equation*}
\tau\big( [J_{i}^{\mathrm{conv}},X_{j} ] P \big)
 = \lim _{L \to \infty} \lim_{\beta \to \infty} 
|C_{L}|^{-1} \omega^{H,\mu}_{\beta} \big( \big[ \mathcal{J}^{\mathrm{spin}}_{L;i}, \mathcal{X}_{j}  \big]\big).
\end{equation*}
We want to open the commutator on the right-hand side term in the last equality and this can be done provided that we localize the operator $\mathcal{X}_{j}$ so that it is bounded. By the argument at the end of the proof of Lemma \ref{lemma:grand-canonical} we note that for any $R \geq L$ we have $ \omega^{H,\mu}_{\beta} \big( \big[ \mathcal{J}^{\mathrm{spin}}_{L;i}, \mathcal{X}_{j}  \big]\big) =  \omega^{H,\mu}_{\beta} \big( \big[ \mathcal{J}^{\mathrm{spin}}_{L;i}, \mathcal{X}_{R;j}  \big]\big)$, where $\mathcal{X}_{R;j}:= \dd \Gamma( \chi_{R} X_{j})$ is bounded when $R$ is finite. Then, by the KMS property we can write
\begin{equation*}
\begin{split}
\omega^{H,\mu}_{\beta} \big( \big[ \mathcal{J}^{\mathrm{spin}}_{L;i}, \mathcal{X}_{j}  \big]\big)
 = \omega^{H,\mu}_{\beta} \Big(  \mathcal{J}^{\mathrm{spin}}_{L;i} \big( \mathcal{X}_{R;j} -   \mathcal{X}_{R;j} (\ii \beta) \big) \Big) \;,
\end{split}
\end{equation*}
and, by slight abuse of notation, letting $ \mathcal{J}_{R;j}(\ii t) := - \ii \frac{\dd}{\dd t}\mathcal{X}_{R;j}(\ii t) = \alpha_{\ii t}\big( \dd \Gamma(\ii [H,\chi_{R}X_{j}])\big)$, we can write
\begin{equation*}
\omega^{H,\mu}_{\beta} \big( \big[ \mathcal{J}^{\mathrm{spin}}_{L;i}, \mathcal{X}_{j}  \big]\big) =
 -\ii \, \omega^{H,\mu}_{\beta} \Big( \int_{0} ^{\beta}\mathcal{J}^{\mathrm{spin}}_{L;i}  \mathcal{J}_{R;j}(\ii t)  \dd t\Big) = 
 -\ii \,\int_{0} ^{\beta} \omega^{H,\mu}_{\beta} \big(\mathcal{J}^{\mathrm{spin}}_{L;i}  \mathcal{J}_{R;j}(\ii t)  \big) \dd t \;,
\end{equation*}
where in the last step we swapped integration with expectation over the grand canonical state because the operator $\mathcal{J}^{\mathrm{spin}}_{L;i} \mathcal{J}_{R;j}( \ii t) $ is bounded, locally in $t$, and thus Bochner integrable over $[0,\beta]$. We shall now remove the localisation $R$ by taking the limit $R \to \infty$. This can be done by noting that $\mathcal{J}_{L;i} \mathcal{J}_{R;j} ( \ii t)$ converges in the operator topology to $\mathcal{J}_{L;i} \mathcal{J}_{j}( \ii t)$, which is Bochner integrable over $[0,\beta]$. Thus, we can take the limit inside the integral by the Lebesgue dominated convergence theorem  as well as inside $\omega^{H,\mu}_{\beta}$. Therefore, using \eqref{eq:T-ordering} we can write
\begin{equation*}
\tau\big( [J_{i}^{\mathrm{conv}},X_{j} ] P \big)
 = - \ii \lim _{L \to \infty} \lim_{\beta \to \infty} 
|C_{L}|^{-1} \int_{0}^{\beta} \mathcal{K}^{\beta,L}_{ij}(t) \dd t \;,
\end{equation*}
which, together with \eqref{eq:first-term-final}, implies the claim.
\end{proof}
To conclude, we obtain an explicit expression for $\widehat{\mathcal{K}}_{ij}$ and thus we can prove Theorem \ref{prop:imaginary-time-repr}.
\begin{proof}[Proof of Theorem \ref{prop:imaginary-time-repr}]
We compute $\mathcal{K}^{\beta,L}_{ij}(t) = \omega^{H,\mu}_{\beta} \big( \mathbf{T} \mathcal{J}^{\mathrm{spin}}_{L;i} \mathcal{J}_{j}(\ii t)  \big) = \omega^{H,\mu}_{\beta} \big(  \mathcal{J}^{\mathrm{spin}}_{L;i} \mathcal{J}_{j}(\ii t)  \big)$,  $t \in (0,\beta]$, by using that $\omega^{H,\mu}_{\beta}$ is quasi-free. First of all, we note that
\begin{equation}\label{eq:truncation-K-function}
\omega^{H,\mu}_{\beta}(\mathcal{J}_{j}(\ii t)) = 0  \;,\qquad j=1,2 \;.
\end{equation}
This follows by observing that $H$ is invariant under $2\pi/3$-rotations, whereas $J_{i}$ transform like vectors.
By the Wick rule and simple algebraic manipulations\footnote{One can indeed check that the Wick rule holds for the expectations with respect to $\omega^{H,\mu}_{\beta}$ of time-ordered products in the creation and annihilation operators.}, for $t \in (0,\beta]$ we find
\begin{equation*}
\begin{split}
&\omega^{H,\mu}_{\beta} \big(\mathbf{T}\mathcal{J}^{\mathrm{spin}}_{L;i}  \mathcal{J}_{j}(\ii t)  \big) 
= \Tr_{\mathcal{H}} \Big( \big(J_{i}^{\mathrm{conv}} \big)_{L} G^{H,\mu}_{\beta}(t) J_{j}G^{H,\mu}_{\beta}(-t)\Big)
 \;.
\end{split}
\end{equation*}
where the operator $G^{H,\mu}_{\beta}(t)$ is the finite-temperature imaginary-time-ordered Green function, that is, for any $f,g \in \mathcal{H}$
\begin{equation}\label{eq:T-ordered-GF}
\langle f, G^{H,\mu}_{\beta}(t) g \rangle _{\mathcal{H}} :=\omega^{H,\mu}_{\beta} \big(\mathbf{T} a(f) \alpha_{\ii t}(a^{*}(g)) \big)  \;.
\end{equation}
By \eqref{eq:Fermi-Dirac-2point} and by noting that $\alpha_{\ii t}(a^{*}(f)) = a^{*}(\ee^{- t(H - \mu \mathbb{1})}f)$, we find that $G^{H,\mu}_{\beta}(t) $ is given by
\begin{equation*}
G^{H,\mu}_{\beta}(t) : = \frac{ \ee^{-t (H-\mu \mathbb{1})}}{\mathbb{1} + \ee^{-\beta (H - \mu \mathbb{1})}} ,\qquad t\in [0,\beta)  \;,
\end{equation*}
while in the rest of the interval $[-\beta,\beta]$ it is determined by the anti-periodicity condition $G^{H,\mu}_{\beta}(t-\beta) = - G^{H,\mu}_{\beta}(t)$. Note that, by Lemma \ref{lemma:shortrange}, $G^{H,\mu}_{\beta}(t) \in \mathcal{P}_{0}(\Hi)$ for any $t \in \mathbb{R}$ and any $\beta >0$.  We also introduce the Fourier components of $G^{H,\mu}_{\beta} (t)$; because $G^{H,\mu}_{\beta}(t-\beta) = - G^{H,\mu}_{\beta}(t)$, only odd frequencies are non-zero; thus, for $\omega \in \frac{2 \pi}{\beta} (\mathbb{Z} + \frac{1}{2})$ we set\footnote{Note that our convention for the Fourier transform is opposite with respect to \cite{Porta17,Porta20,MaPo}.}
\begin{equation*}
\begin{split}
\widehat{G}^{H,\mu}(\omega) &: = \int_{0}^{\beta} \ee^{-\ii \omega t} G^{H,\mu}_{\beta} (t)\dd t
 = \big(H - (\mu-\ii \omega) \mathbb{1}\big)^{-1}  \;.
\end{split}
\end{equation*}
In particular, $\widehat{G}^{H,\mu}(\omega) \in \mathcal{P}_{0}(\Hi)$, and its matrix elements decay as $|\omega|^{-1}$. Also, note that $\widehat{G}^{H,\mu}$ does not depend explicitly on $\beta$, which on the other hand only enters the frequency domain $\frac{2 \pi}{\beta} (\mathbb{Z} + \frac{1}{2})$.

Since $J_{i}^{\mathrm{conv}}$ is localised and since $G^{H,\mu}_{\beta},J_{j} \in \mathcal{P}_{0}(\Hi)$, the operator inside the trace is trace-class, uniformly in $t$ we exchange the trace with integration and obtain
\begin{equation*}
\begin{split}
\int_{0}^{\beta} \ee^{\ii \eta _{\beta} t} & \Tr_{\mathcal{H}} \Big( \big(J_{i}^{\mathrm{conv}} \big)_{L} G^{H,\mu}_{\beta}(t) J_{j}G^{H,\mu}_{\beta}(-t) \Big) \dd t 
\\
& = \frac{1}{\beta} \sum_{\omega \in (2 \pi / \beta) \mathbb{Z}} \Tr_{\mathcal{H}} \Big( \big(J_{i}^{\mathrm{conv}} \big)_{L} \widehat{G}^{H,\mu}(\omega -\eta_{\beta} ) J_{j}\widehat{G}^{H,\mu}(\omega) \Big) \;,
\end{split}
\end{equation*}
where we also swapped trace and summation over $\omega$ because the matrix elements of $\widehat{G}^{H,\mu}$ decay as $|\omega|^{-1}$. Taking the $\beta \to \infty$ limit gives
\begin{equation*}
\widehat{\mathcal{K}}_{ij}(\eta) = (2 \pi)^{-1} \lim _{L \to \infty} |C_{L}|^{-1} \int_{\mathbb{R}} \dd \omega\Tr_{\mathcal{H}} \Big( \big(J_{i}^{\mathrm{conv}} \big)_{L} \widehat{G}^{H,\mu}(\omega - \eta ) J_{j} \widehat{G}^{H,\mu}(\omega) \Big) \;.
\end{equation*}
By similar arguments we take the limit $L \to \infty$ inside the integral and, because $\widehat{G}^{H,\mu}$, $J_{i}^{\mathrm{conv}}$ and $J_{j}$ are short-range, by Lemma \ref{lemma-move-local} and Lemma \ref{lem:tau}\ref{it:tau3} we obtain
\begin{equation*}
\widehat{\mathcal{K}}_{ij}(\eta) = \int_{\mathbb{R} \times \mathbb{T}^{*}_{2}} \frac{\dd \mathbf{k}}{(2 \pi)^{3}} \Tr _{\mathbb{C}^{4}} \Big(  \widehat{J}^{\mathrm{conv}}_{i}(k) \widehat{G}^{H,\mu}(\omega - \eta;k) \widehat{J}_{j}(k) \widehat{G}^{H,\mu}(\omega;k) \Big) \;,
\end{equation*}
where $\mathbf{k} \equiv (\omega, k)$ and where $\widehat{G}^{H,\mu}(\omega;k):= \big(H(k) - (\mu - \ii \omega) \mathbb{1} \big)^{-1}$. To conclude, since 
$\partial_{k_{i}} H(k) = \partial_{k_{i}}\( \widehat{G}^{H,\mu}(\omega;k) ^{-1}\)$,
it is easy to check that 
\begin{equation*}
\begin{split}
\widehat{J}^{\mathrm{conv}}_{i}(k) = \frac{1}{2} \big\{  \partial_{k_{i}} \widehat{G}^{H,\mu}(\omega;k) ^{-1}, S_{z} \big\} \;, \qquad
\widehat{J}_{j}(k) = \partial_{k_{j}} \widehat{G}^{H,\mu}(\omega;k) ^{-1} \;,
\\
-\partial_{\eta} \widehat{G}^{H,\mu}(\omega - \eta;k)\big|_{\eta =0} = - \widehat{G}^{H,\mu}(\omega ;k) \big[ \partial_{\omega} \widehat{G}^{H,\mu}(\omega;k) ^{-1} \big] \widehat{G}^{H,\mu}(\omega ;k) \;.
\end{split}
\end{equation*}
The claim then follows because $(H,\mu)$ is an insulator so that we can take the limit $\eta \to 0^{+}$ inside the integral by the Lebesgue dominated convergence theorem.
\end{proof}

\goodbreak

\let\oldaddcontentsline\addcontentsline
\renewcommand{\addcontentsline}[3]{}

\let\addcontentsline\oldaddcontentsline


\bigskip \bigskip

{\footnotesize

\begin{tabular}{ll}
(L.~Fresta) 
		        	&  \textsc{Mathematics and Physics Department, University Roma Tre} \\ 
        	&   Largo San Leonardo Murialdo 1, 00146 Roma, Italy \\        
        	\\
        		&  {E-mail address}: \href{mailto:luca.fresta@uniroma3.it}{\texttt{luca.fresta@uniroma3.it}}\\
\\
(G.~Marcelli)   
		&  \textsc{Mathematics and Physics Department, University Roma Tre} \\ 
        	&   Largo San Leonardo Murialdo 1, 00146 Roma, Italy \\        
        	\\	&  {E-mail address}: \href{mailto:giovanna.marcelli@uniroma3.it}{\texttt{giovanna.marcelli@uniroma3.it}}\\
\end{tabular}

} 


\begin{thebibliography}{00}
\bibitem[AG]{AizenmanGraf}
\textsc{Aizenman, M.; Graf G.M.} : Localization bounds for an electron gas, {\it J. Phys. A: Math. Gen.} {\bf 31}, 6783 (1998).

\bibitem[AW]{AizenmanWarzel}
\textsc{Aizenman, M.; Warzel, S.}. {\it Random Operators}. Graduate Studies in Mathematics 168, American Mathematical Society, 2015.


\bibitem[AMP]{Porta18}
\textsc{Antinucci G.; Mastropietro V.; Porta M.} :
Universal Edge Transport in Interacting Hall Systems.
\textit{Commun. Math. Phys.} \textbf{362}, 295--359 (2018).


\bibitem[ASV]{Avila}
\textsc{Avila, J.C., Schulz-Baldes, H., Villegas-Blas, C.} : Topological invariants of edge states for periodic two-dimensional models. \textit{Math. Phys., Anal. Geom.} \textbf{16} (2), 127--170 (2013).


\bibitem[AS]{AvronSeiler}
\textsc{Avron, J. E.; Seiler, R.} : Quantization of the Hall conductance for general, multiparticle Schr\"odinger Hamiltonians. {\it Phys. Rev. Lett.} {\bf 54}, 259--262 (1985).

\bibitem[AS$^2_1$]{Avron83}
\textsc{Avron J. E.; Seiler R.; Simon B.} : Homotopy and quantization in condensed matter physics. \textit{Phys. Rev. Lett.} \textbf{51}, 51 (1983).

\bibitem[AS$^2_2$]{Avron94}
\textsc{Avron J. E.; Seiler R.; Simon B.} : Charge deficiency, charge transport and comparison of dimensions. \textit{Commun.Math. Phys.} \textbf{159}, 399--422 (1994).

\bibitem[BdRF1]{Bachmann18}
\textsc{Bachmann S.; De Roeck W.; Fraas M.} : The adiabatic theorem and linear response theory for extended quantum systems. \textit{Communications in Mathematical Physics} \textbf{361}, 997--1027 (2018).


\bibitem[BBdRF3]{Bachmann20II}
\textsc{Bachmann S.; Bols A.; De Roeck W.; Fraas M.} :
A many-body index for quantum charge transport.
\textit{Communications in Mathematical Physics} \textbf{375}, 1249--1272 (2020).

\bibitem[BBR]{BBR}
\textsc{Bachmann S.; Bols A.; Rahnama M.} :
Many-body Fu--Kane--Mele index. Preprint  arXiv:2406.19463.

\bibitem[Be]{Bel}
\textsc{Bellissard, J.} : K-theory of $C^*$-algebras in solid state physics. In: Dorlas, T.C., Hugenholtz, N.M., Winnink, M. (eds.) Statistical Mechanics and Field Theory: Mathematical Aspects. Vol. 257 in Lecture Notes in Physics, pp. 99–156. Springer, Berlin (1986)

\bibitem[BES]{Bellissard94}
\textsc{Bellissard, J.; van Elst, A.; Schulz-Baldes, H.} :
The non-commutative geometry of the quantum Hall effect.
{\it J. Math. Phys.} {\bf 35}, 5373 (1994). 

\bibitem[BHZ$_1$]{Bernevig}
\textsc{Bernevig, B. A.; Hughes, T. L.; Zhang, S.-C.}: Quantum spin Hall effect and topological phase transition
in HgTe quantum wells. \textit{Science} \textbf{314}, 1757--1761 (2006).

\bibitem[BHZ$_2$]{Bernevig2}
\textsc{Bernevig, B. A.; Hughes, T. L.; Zhang, S.-C.}:
Quantum Spin Hall Effect.
\textit{Phys. Rev. Lett.} \textbf{96}, 106802 (2006).

\bibitem[BGKS]{Bouclet}
\textsc{Bouclet, J. M.; Germinet, F.; Klein, A.; Schenker, J. H.} : Linear response theory for magnetic Schr\"odinger operators in disordered media. {\it J. Funct. Anal.} {\bf 226}, 301--372 (2005).

\bibitem[BR]{Bratteli}
\textsc{Bratteli O.; Robinson  D. W.} : \textit{Opertor
algebras and quantum statistical mechanics}, Volume 1-2, Springer (2002).

\bibitem[Mol$_{3}$]{Mol3}
\textsc{Br\"une, C., Roth, A., Buhmann, H., Hankiewicz, E.M., Molenkamp, L.W., Maciejko, J., Qi, X.-L., Zhang, S.-C.} : 
Spin polarization of the quantum spin Hall edge states. \textit{Nature Physics} \textbf{8}, 485--490 (2012).

\bibitem[Col]{Coleman}
\textsc{Coleman, S.} : \textit{Aspects of Symmetry: Selected Erice Lectures}. Cambridge University Press, 1988. 
%

\bibitem[FK]{FrohlichKerl}
\textsc{Fr\"ohlich J. ; Kerler, T.} :
Universality in quantum Hall systems. \textit{Nucl.Phys.} \textbf{B} 354, 369--417 (1991).

\bibitem[FS]{Frohlich}
\textsc{Fr\"ohlich J. ; Studer U.~M.} :
Gauge invariance and current algebra in nonrelativistic many-body theory.
\textit{Rev. Mod. Phys.} \textbf{65}, 733 (1993).

\bibitem[FST]{FrohlichST}
\textsc{Fr\"ohlich J. ; Studer U.~M.; Thiran E.} :
Quantum Theory of Large Systems of Non-relativistic Matter. \textit{Les Houches Lectures} 1994, Elsevier, New York (1995). arXiv:cond-mat/9508062

\bibitem[FW]{FW}
\textsc{Fr\"ohlich J. ; Werner P.} :
Gauge theory of topological phases of matter. \textit{EPL} \textbf{101} 47007 (2013).

\bibitem[FZ]{FrohlichZee}
\textsc{Fr\"ohlich J. ; Zee A.} : Large scale physics of the quantum Hall fluid. \textit{Nucl.Phys.} \textbf{B} 364, 517--540 (1991).


\bibitem[FKM]{FuKaneMele}
\textsc{Fu, L.; Kane, C.L.; Mele, E.J.} : Topological insulators in three dimensions, {\it Phys. Rev. Lett.} {\bf 98}, 106803 (2007).

\bibitem[GJMP]{Porta16}
\textsc{Giuliani A.; Jauslin, I., Mastropietro V.; Porta M.} : Topological phase transitions and universality in the Haldane--Hubbard model. \textit{Physical Review B} \textbf{94} (20), 205139 (2016).
%
\bibitem[GMP$_{1}$]{Porta17}
\textsc{Giuliani A.; Mastropietro V.; Porta M.} : Universality of the Hall conductivity in interacting electron systems.
\textit{Commun. Math. Phys.} \textbf{349} (3), 1107--1161 (2017).

\bibitem[GMP$_{2}$]{Porta20}
\textsc{Giuliani A.; Mastropietro V.; Porta M.} : Quantization of the interacting Hall conductivity in the critical regime.
\textit{Journal of Statistical Physics} \textbf{180} (1), 332--365 (2020).

\bibitem[Gr]{Graf review}
\textsc{Graf, G.M.} : Aspects of the Integer Quantum Hall Effect, \emph{Proceedings of Symposia in Pure Mathematics} \textbf{76}, 429--442 (2007).

\bibitem[GP]{GrafPorta}
\textsc{Graf, G.M.; Porta, M.} : Bulk-edge correspondence for two-dimensional topological insulators, {\it Commun. Math. Phys.} {\bf 324}, 851--895 (2013).

\bibitem[Hal]{Haldane88}
\textsc{Haldane, F.D.M.}: Model for a Quantum Hall Effect without Landau levels: condensed-matter realization of the ``parity anomaly'', {\it Phys. Rev. Lett.} {\bf 61}, 2017 (1988).

\bibitem[HK]{HasanKane}
\textsc{M.Z. Hasan and C.L. Kane}. Colloquium: topological insulators. {\it Reviews of modern physics} {\bf 82}, 3045 (2010).

\bibitem[Has]{Hasanexp}
\textsc{Hsieh, D.; Qian, D.; Wray, L.; Xia, Y.; Hor, Y.S.; Cava, R.J.; Hasan, M.Z.}: A topological Dirac insulator
in a quantum spin Hall phase. \textit{Nature} \textbf{452}, 970 (2008).

\bibitem[HaMi]{Hastings}
\textsc{Hastings M.B.; Michalakis S.} : Quantization of Hall conductance for interacting electrons on a torus. \textit{Commun. Math. Phys.}, \textbf{334} 433--471, (2015).

\bibitem[HaWe]{Hastings2}
\textsc{Hastings M.B.; Wen X.G.}: Quasiadiabatic continuation of quantum states: The
stability of topological ground-state degeneracy and emergent gauge invariance.
\textit{Phys. Rev. B}, \textbf{72} 045141, (2005).

\bibitem[Jack]{Jackiw}
\textsc{Jackiw R.W.} :
Nonperturbative and Topological Aspects of Gauge Theory. In: Encyclopedia of Mathematical Physics, 
Elsevier, 2006.

\bibitem[KM$_1$]{KaneMele2005}
\textsc{Kane C.L.; Mele E.J.} : $\Z_2$ Topological Order and the Quantum Spin Hall Effect. {\it Phys. Rev. Lett.} {\bf 95}, 146802 (2005).

\bibitem[KM$_2$]{KaneMele_graphene}
\textsc{Kane C.L.; Mele E.J.} : Quantum Spin Hall Effect in graphene. {\it Phys. Rev. Lett.} {\bf 95}, 226801 (2005).

\bibitem[Kit]{Kitaev}
\textsc{Kitaev, A.} : Periodic table for topological insulators and superconductors, {\it AIP Conf. Proc.} {\bf 1134}, 22 (2009).

\bibitem[LPGetal]{LPGetal}
\textsc{Lin, KS., Palumbo, G., Guo, Z. et al.}: Spin-resolved topology and partial axion angles in three-dimensional insulators, {\it Nat. Commun.} {\bf 15}, 550 (2024).

\bibitem[Mol$_1$]{Molenkamp}
\textsc{König, M.; Wiedmann, S.; Brüne, C.; Roth, A.; Buhmann, H.; Molenkamp, L.W.; Qi, X.-L.; Zhang, S.-C.}:
Quantum spin Hall insulator state in HgTe quantum wells. \textit{Science} \textbf{318}, 766 (2007).

\bibitem[Mol$_2$]{Mol2}
\textsc{König, M.; Buhmann, H.; Molenkamp, L.W.; Hughes, T.; Liu,C.-X.; Qi,X. L.; Zhang S.-C.} :
The Quantum Spin Hall Effect: Theory and Experiment
\textit{Journ. Phys. Soc. Jap.}, \textbf{77}, 031007 (2008).

\bibitem[Ku]{Kubo}
\textsc{Kubo, R.} : Statistical-mechanical theory of irreversible processes I: General theory and simple applications to magnetic and conduction problems. {\it J. Phys. Soc. Jpn.} {\bf 12}, 570--586 (1957).

\bibitem[Kuc]{Kuchment}
\textsc{Kuchment, P.} : An overview of periodic elliptic operators. {\it Bull. AMS} {\bf 53}, 343--414 (2016).

\bibitem[MaMo]{MaMo}
\textsc{Marcelli, M; Monaco, D.} :
From charge to spin: Analogies and differences in quantum transport coefficients.
J. Math. Phys. \textbf{63}, 072102 (2022)

\bibitem[MaPaTa]{MaPaTa}
\textsc{Marcelli, M; Panati, G.; Tauber, C.} : Spin conductance and spin conductivity in topological insulators: Analysis of Kubo-like terms. {\it Ann. Henri Poincar\'{e}} {\bf 20}, 2071--2099 (2019).

\bibitem[MaPaTe]{MaPaTe}
\textsc{Marcelli, M; Panati, G.; Teufel, S.} : A New Approach to Transport Coefficients in the Quantum Spin Hall Effect. {\it Ann.~Henri Poincar\'{e}} {\bf 22}, 1069--1111 (2021).

\bibitem[MaPo$_1$]{MPS}
\textsc{Mastropietro, V.; Porta, M.} :
Spin Hall insulators beyond the helical Luttinger model.
\textit{Phys. Rev. B} \textbf{96}, 245135  (2017).

\bibitem[MaPo$_2$]{MaPo}
\textsc{Mastropietro, V.; Porta, M.} :
Canonical Drude Weight for Non-integrable Quantum Spin Chains. \textit{J. Stat. Phys.} \textbf{172} 379--397 (2018).

\bibitem[MoUl]{MoUl}
\textsc{Monaco, D.; Ulcakar, L.}
Spin Hall conductivity in insulators with nonconserved spin.
\textit{Phys. Rev. B} \textbf{102} (2020).

\bibitem[Pr]{Prodan2009}
\textsc{Prodan, E.} : Robustness of the spin-Chern number. {\it Phys. Rev. B} {\bf 80}, 125327 (2009).

\bibitem[RS4]{RS4}
\textsc{Reed, M.; Simon, B.}
Methods of modern mathematical physics, Vol. IV, Analysis of operators. Academic Press, 1978.


\bibitem[SSGR]{Ryu-inv}
\textsc{Shiozaki, K.; Shapourian, H.; Gomi, K.; Ryu, S.} :  Many-body topological invariants for fermionic short-range entangled topological phases protected by antiunitary symmetries. \textit{Phys. Rev. B} \textbf{98}, 035151
(2018).

\bibitem[Sch]{SchulzBaldes}
\textsc{Schulz-Baldes, H.} : Persistence of spin edge currents in disordered Quantum Spin Hall systems, {\it Commun. Math. Phys.} {\bf 324}, 589--600 (2013).

\bibitem[SSTH]{ShengShengTingHaldane2005}
\textsc{L.~Sheng, D.N.~Sheng, C.S.~Ting, F.D.M.~Haldane}:
Nondissipative spin Hall effect via quantized edge transport.
{\it Phys. Rev. Lett.} {\bf 95}, 136602 (2005).  

\bibitem[SWSH]{ShengWengShengHaldane2006}
\textsc{D.N.~Sheng, D. N.; Weng,  Z.Y.; Sheng, L.; Haldane, F.D.M.} : 
Quantum spin-Hall effect and topologically invariant Chern numbers. 
{\it Phys. Rev. Lett.} {\bf 97}, 036808 (2006).  

\bibitem[SZXN]{ShiZhangXiaoNiu}
\textsc{Shi, J.; Zhang, P.; Xiao, D.; Niu, Q.} : Proper definition of spin current in spin-orbit coupled systems. {\it Phys. Rev. Lett.} {\bf 96}, 076604 (2006).


\bibitem[SCNSJM]{Sinovaetalii}
\textsc{J.~Sinova, D.~Culcer, Q. Niu, N.A.~Sinitsyn, T.~Jungwirth, A.H.~MacDonald}:
Universal intrinsic spin Hall effect. {\it Phys. Rev. Lett.} {\bf 92}, 126603 (2004).

\bibitem[SVWBJ]{Sinova}
J.~Sinova, S.~O.~Valenzuela, J.~Wunderlich, C.~H.~Back, and T.~Jungwirth.
Spin Hall effects. 
\textit{Rev. Mod. Phys.} \textbf{87} (2015).

\bibitem[Teu1]{Teufel1}
\textsc{Teufel, S.} : Adiabatic Perturbation Theory in Quantum Dynamics. No. 1821 in Lecture Notes in Mathematics. Springer, Berlin (2003).

\bibitem[Teu2]{Teufel}
\textsc{Teufel, S.} : Non-equilibrium almost-stationary states and linear response for gapped
quantum systems. \textit{Comm. Math. Phys.} \textbf{373} 621--653 (2020).

\bibitem[TKNN]{TKNN}
\textsc{Thouless, D.J.; Kohmoto, M.; Nightingale, M.P.; de Nijs, M.} : Quantized Hall conductance in a two-dimensional periodic potential, {\it Phys. Rev. Lett.} {\bf 49}, 405--408 (1982).

\bibitem[vKDP]{vonKlitzing}
\textsc{von Klitzing, K.; Dorda, G.; Pepper, M.} : New method for high-accuracy determination of the fine-structure
constant based on quantized Hall resistance. \textit{Phys. Rev. Lett.} \textbf{45}, 494 (1980).

\bibitem[WaQiZh]{WangQiZhang}
\textsc{Whang, Z.; Qi, X.-L.; Zhang, S.-C.} : Topological Order Parameters for Interacting Topological Insulators. \textit{Phys. Rev. Lett.} \textbf{105}, 256803 (2010).

\bibitem[Wolf]{Wolfram}
Wolfram Research, Inc., Mathematica, Version 13.1, Champaign, IL (2022).

\bibitem[ZWSXN]{Zhang}
\textsc{Zhang, P.;  Wang, Z.;  Shi, J.; Xiao, D.; Niu, Q.} : 
Theory of conserved spin current and its application to a two-dimensional hole gas. 
{\it Phys. Rev. B} {\bf 77}, 075304 (2008).


%

%
%

%
%
%
%
%
%
%
%
%
%
%
%
%
%
%
%
%
%
%
%

%


%

%
%
%
%
%
%
%
%
%
%
%
%
%
%
%
%
%
%
%
%
%
%
%
%
%
%
%
%
%
%
%
%
%

\end{thebibliography}
\end{document}